\documentclass[english,10pt,letterpaper,twocolumn,floatfix,allowfontchangeintitle,amsmath,amssymb,,nofootinbib,accepted=2024-08-30]{quantumarticle}

\pdfoutput=1

\usepackage[utf8]{inputenc}
\usepackage[english]{babel}
\usepackage[T1]{fontenc}
\usepackage{amsmath}
\usepackage{hyperref}

\usepackage{dsfont}
\usepackage[normalem]{ulem}
\usepackage{mathtools}
\usepackage[makeroom]{cancel}
\usepackage{bbm}
\usepackage{enumitem}
\usepackage{xcolor}
\usepackage{graphicx}
\usepackage{amsthm}
\usepackage{amssymb}        
\usepackage{thmtools} 
\usepackage{thm-restate}  
\usepackage{cases}
\usepackage{tikz-cd}
\usepackage{tcolorbox}

\usepackage{tabularx}
    \newcolumntype{L}{>{\raggedright\arraybackslash}X}

\newtheorem{lemma}{Lemma}

\newtheorem*{theorem*}{Theorem}
\newtheorem{example}{Example}

\theoremstyle{definition}
\newtheorem{definition}{Definition}
\newtheorem*{principle*}{Principle}
\usepackage{algorithm}
\usepackage{algpseudocode}
\usepackage{comment}
\newtheorem{observation}{Observation}

\definecolor{js}{rgb}{0.578125,0.23828125,0.9609375}
\definecolor{purple}{rgb}{0.578125,0.23828125,0.9609375}
\definecolor{green}{rgb}{0.4470588235,0.968627451,0.6117647059}
\definecolor{Next}{rgb}{0,0.56,1}
\definecolor{ATTENZIONE}{rgb}{1,0.1,0.1}
\definecolor{dacompletare}{rgb}{0.945,0.749,0.274}

\newcommand{\id}{\mathbb{I}}

\newcommand{\tr}{\mathrm{Tr}} %old

\usepackage{braket}
 
\newcommand{\lland}{\wedge}
\newcommand{\llor}{\vee}

\newcommand{\TTT}{\textbf{T}}
\newcommand{\FFF}{\textbf{F}}
\newcommand{\AAA}{\textbf{A}}
\newcommand{\NNN}{\textbf{N}}

\newcommand{\uncert}{\mathcal{X}}
\newcommand{\indisti}{\uncert}
\newcommand{\indisc}{\uncert}
\newcommand{\inacc}{\uncert}
\newcommand{\SMUB}{\tilde{S}}
\newcommand{\mus}[1]{\vec{p}_{\frac{1}{#1}}}
\newcommand{\musq}[1]{\vec{s}_{\frac{1}{#1}}}
\newcommand{\musr}[1]{\vec{r}_{\frac{1}{#1}}}

\newcommand{\AR}{amount of relation}
\newcommand{\REL}{\tilde{\mathcal{R}}}
\newcommand{\lerifloor}[1]{\left\lfloor{#1}\right\rfloor}

\newtcolorbox{BoxWT_green}[1]{sharp corners, colback=green!5!white,colframe=green!75!black,fonttitle=\bfseries,title=#1,size=small}
\newtcolorbox{BoxWT_purple}[1]{sharp corners,colback=purple!5!white,colframe=purple!75!black,fonttitle=\bfseries,title=#1,size=small}

    \makeatletter
\def\@fnsymbol#1{\ensuremath{\ifcase#1\or \dagger\or *\or \ddagger\or
   \mathsection\or \mathparagraph\or \|\or **\or \dagger\dagger
   \or \ddagger\ddagger \else\@ctrerr\fi}}
    \makeatother

\usepackage[numbers]{natbib}

\newcommand*{\mySpecialfootnotes}[1]{%
  \patchcmd{\@footnotetext}{\floatingpenalty\@MM}{\floatingpenalty#1\relax}%
           {}{\errmessage{Couldn't patch \string\@footnotetext}}%
}

\begin{document}
\title{A Theory of Inaccessible Information${}^1$}
	\author{\href{mailto:jacopo.surace@gmail.com}{Jacopo Surace}} 
	\affiliation{Perimeter Institute for Theoretical Physics, 31 Caroline Street North, Waterloo, ON N2L 2Y5, Canada}
	\affiliation{ICFO - Institut de Ciencies Fotoniques, The Barcelona Institute of Science and Technology, Castelldefels (Barcelona), 08860, Spain}	
\begin{abstract}

What would be the consequences if there were fundamental limits to our ability to experimentally explore the world? In this work we seriously consider this question. 
We start by assuming the existence of statements whose truth value is not experimentally accessible. That is, there is no way, not even in theory, to directly test if these statements are true or false. We further develop a theory in which experimentally accessible statements are a union of a fixed minimum number of inaccessible statements. For example, the value of truth of the statements \textbf{a} and \textbf{b} is not accessible, but the value of truth of the statement “\textbf{a} or \textbf{b}" is accessible. We do not directly assume probability theory, we exclusively define experimentally accessible and inaccessible statements and build on these notions using the rules of classical logic. We find that an interesting structure emerges. Developing this theory, we relax the logical structure, naturally obtaining a derivation of a constrained quasi-probabilistic theory rich in structure that we name \textit{theory of inaccessible information}. Surprisingly, the simplest model of theory of inaccessible information is the qubit in quantum mechanics. Along the path for the construction of this theory, we characterise and study a family of multiplicative information measures that we call inaccessibility measures.
	\end{abstract}
\maketitle

\section{Introduction}
	\footnotetext{The consequences of seriously accepting that some statements cannot be directly proven or disproven experimentally.}

%For quite some time nature has been whispering to us the following harsh words “You cannot know!”. \\
%Call it Heisenberg uncertainty relations, call it intrinsic quantum randomness or with any other name: one lesson of quantum mechanics is that there is a fundamental limit to the information we can acquire experimentally. 
%Besides its mysterious postulates, the lesson quantum mechanics is teaching us is maybe more fundamental than the postulates themselves: at the core of our experimental experience there seems to be a fundamental limit to the information we can experimentally gather. 
%Little we know if quantum mechanics is here to stay as the most fundamental theory we manage to formulate, but the lesson it teach us is something that seems more fundamental than the postulates of quantum mechanics itself.
%What if, forgetting (but not for long) about quantum mechanics, we simply assume as a basic postulate the existence of a fundamental limit to the information we can acquire experimentally and explore the structure of the theories emerging from this assumption?
%In this work we try to do exactly this. Starting from one of the most basic structure, the language of logic, we build up what we will call a \textit{theory of inaccessible information}.

Starting with the discovery of the uncertainty relations and moving to the development of the concept of intrinsic quantum randomness \cite{dhara2014,acin2016,herrero-collantes2017,senno2022}, quantum mechanics taught us that there is a fundamental limit to the information we can acquire experimentally.
%Little do we know about whether quantum mechanics is here to stay as the most fundamental theory we manage to formulate, but the lesson it teaches us is something that seems more fundamental than the postulates of quantum mechanics itself.
Little do we know about whether quantum mechanics will remain the most fundamental theory we can formulate. However, the lesson it teaches us is something that seems more fundamental than the postulates of quantum mechanics itself.

What if, forgetting (but not for long) about quantum mechanics, we simply assume as a basic postulate the existence of a fundamental limit to the information we can acquire experimentally, and explore the structure of the theories emerging from this assumption?
In this work, we explore this question, studying the consequences of what we call the Inaccessibility Hypothesis
\\
\\
{\bf Inaccessibility Hypothesis $\mathcal{I}\mathcal{H}$}: \\
There are some statements about the physical world that are absolutely experimentally inaccessible, i.e. there is even in principle no experimental procedure that could confirm or disconfirm such statements. 
\\

 Following a path partially different from previous works in this area \cite{spekkens2007b,spekkens2016}. Without assuming probability theory, we start from the basic language of logic \cite{piron1976,knuth2006}, to build up what we will call a \textit{theory of inaccessible information}.

We  develop this idea starting from the separation of all the well formed statements into two kinds: experimentally accessible or inaccessible.
%We start by  defining two complementary sets: the set of experimentally accessible statements and the set of experimentally inaccessible statements. 
%Every well-formed statement must be assigned to one of these sets. 
 We define experimentally inaccessible statements as those for which there does not exist, not even in theory, an experimental method for assessing their truth value. We will not deal with metaphysical questions such as what does it mean for something to be true or false while being experimentally inaccessible.

In our theory all experimentally accessible statements are formed by combining a minimum of $d$ experimentally inaccessible statements, with $d$ representing what we call the \textit{accessibility-depth} of our experimental knowledge.  For example, if $d=3$ and the three atomic statements ${a, b, c}$ are inaccessible, the statement $r=a\vee b$ (meaning ${\mbox{"$a$ or $b$"}}$) cannot be experimentally accessed, but the statement $s=a\vee b\vee c$, composed by $d=3$ inaccessible statements, can be.  We can now see why $d$ is interpreted as a \textit{accessibility-depth}: even if an experiment shows that $s$ is true, this knowledge will leave us completely ignorant of the truth values of the three experimentally inaccessible statements $a,b,c$ that make up $s$.
%Essentially, our experimental knowledge has limits, and some information is inaccessible. 
This is how we characterise the fundamental limitations of our experimental knowledge. 
% The goal of this work is to explore the complex structures that arise from these simple assumptions.

By developing and relaxing the initial structure of the theory, a notion of quasi-probability associated to experimentally inaccessible statements and a notion of probability associated to experimentally accessible statements will naturally emerge. In this relaxed scenario, the value of the  accessibility-depth parameter $d$ will impose some information-theoretic constraints on the probability vectors we will associate to the experimentally accessible statements. Surprisingly, we will find that, among the many possible models admitted, the structure of the qubit originates as the simplest model in the theory inaccessible information.
%Our purpose will be uncovering the properties of a theory of inaccessible information and characterising the set of probability distributions admitted by the structure of the theory.
 
%as a side result,
 % if one is interested in describing the set  of allowed probability vectors emerging from our theory using the language of hermitian matrices, part of the standard formalism of quantum mechanical states 
%and unitary transformations
 %is recovered. In particular, we find out that 

%This theory of inaccessible information should not be intended as tentative derivation of quantum mechanics, but as a study of the consequences of accepting that some statements cannot be directly proven or disproven experimentally.
%
%Along the path for the construction of this theory of inaccessible information, we will additionally describe how to derive quasi-probabilities from the structure of the algebra of logical statements and we will characterise and study set of \textit{multiplicative} information measure that we call \textit{inaccessibility measures} (measures that can be mapped to Shannon's, Renyi's and Tsallis' entropies).
% 

We caution the reader that our theory of inaccessible information is not intended as a derivation of quantum mechanics, but rather as a study of the consequences of accepting that some statements cannot be directly proven or disproven experimentally. 
Along the path of our construction we will additionally uncover how to derive quasi-probabilities from the structure of the algebra of logical statements and introduce a set of multiplicative information measures that we call inaccessibility measures.
%Our work sheds light on the rich structure that emerges from the simple assumptions underlying the \textit{theory of inaccessible information}.

\section{Experimentally accessible and inaccessible statements}
\label{sec:section-II}

The aim of this section is to introduce the framework for our theory.
We begin by defining a formal language, consisting of a set of symbols and a set of rules for their composition. The elements of the set of symbols we will use are referred to as \textit{atomic statements}. An atomic statement is intended to represent a proposition that declares a fundamental physical property of the system.
Importantly, we want atomic statements to encode mutually exclusive properties of a system: it is absurd for any system to be characterised by two atomic statements simultaneously. This idea will be formalised shortly, but for the time being, we can provide an example.
Consider, for example, the roll of a die. The value of the top face of a die is a fundamental property of the die. Additionally, a die has only one top face, and it is absurd to talk about two top faces (it would not be a die). Thus statements like $s_1$="The die has a 6 on the top face" or $s_2$="The die has a 1 on the top face" are good examples of atomic statements.\footnote{Other examples of atomic statements can be "the spin has direction pointing as $\vec{r}$" for a complete set of directions, or "pressing this button I get a $i$ on the screen" for all possible outcomes $i$. An example of statement that will be considered inaccessible is the statement "The value of $X$ and $Z$ for the qubit is $1$".}  Formally, a model of \textit{dimension} $D$ will be characterised by a set of $D$ atomic statements $\mathcal{A}(\mathcal{S})=\{s_1,s_2,\dots,s_D\}$.

We combine atomic statements to form new statements using a set of logical constants from classical propositional logic. These include the unary operator $\neg$, denoting the negation of a statement, and the binary operators "or" ($\vee$) and "and" ($\wedge$).  

Additionally, we introduce two special elements,  $\top$ and $\perp$, which denote the intuitive ideas of the tautology statement and the absurd statement, respectively.

To the collection $(\mathcal{A}(\mathcal{S}), \vee, \wedge, \neg, \top, \perp)$ of all the ingredients we just defined, we add a set of rules specifying the equivalence between different statements. The first equivalence we specify uses $\perp$ to formalise the property of atomic statements being mutually exclusive. Specifically, atomic statements are defined such that the composition of any two atomic statements $s_i$ and $s_j$ with the "and" operator forms a statement equivalent to the absurd statement, $s_i \wedge s_j = \perp$ for every $i \neq j$. %, and every possible statements obtained with any composition is equivalent to a statement that is composition of atomic statements with $\vee$. 
Further equivalence relations are specified by the rules in appendix \ref{appendix:algebra-of-statements} with equations \eqref{eq:axioms-lattice},\eqref{eq:axioms-complete},\eqref{eq:axioms-distributivity}. These rules encode, for example, the intuition that the statement $s_1$ and the statement $s_1 \vee s_1$ should be considered equivalent, as defined by the idempotency rule \eqref{eq:axioms-lattice}.\footnote{As another example, perhaps more complex, one can show that for $3$ atomic statements $s_1,s_2,s_3$ the equivalence $s_1=(s_1\vee s_3)\wedge(s_1\vee s_2)$ holds. In fact $(s_1\vee s_3)\wedge(s_1\vee s_2)=s_1\vee(s_3\wedge s_2)$ from rule \eqref{eq:distributivity-or}, this is equivalent to $s_1\vee\perp$ because of the definition of mutually exclusive atomic statements, in turn this is equivalent to $s_1\vee(s_1\wedge \neg s_1)$ using rule \eqref{eq:and-complementation}. Finally one obtains the desired results using the absorption rule \eqref{eq:absorbtion}.}
Combining atomic statements in all possible ways, one obtains the quotient space $\mathcal{S}$, which is the set of equivalence classes of statements formed by combinations of atomic statements, or, equivalently, the set of representatives of each equivalence class of statements. The set $\mathcal{S}$ represents the set of all possible statements we can formulate. 
It is possible to see that every statement in $\mathcal{S}$ is equivalent to a composition of atomic statements using only the binary operator $\vee$, and the $\top$ element is equivalent to the composition of all atomic statement using the binary operator $\vee$.
We obtain the Boolean algebra of statements $(\mathcal{S}, \vee, \wedge, \neg, \top, \perp)$, or equivalently, a complete Boolean lattice of statements built from the set of atomic statements (see appendix \ref{appendix:algebra-of-statements} and \cite{piron1976,knuth2004,knuth2005a}).

Now to each element of the lattice we attach a label that we call \textit{truth value} of the statement. This label can assume two values $\TTT$ or $\FFF$, read as "true" and "false" respectively.
An assignment of truth value to each element of the lattice is considered admissible if it respects the following rules. First, it must be consistent with the rule specified by the truth tables  \ref{table:truth-table}.  For example, if we assign to $s_1$ the value $\TTT$ and to $s_2$ the value $\FFF$, then the only possible assignment to $s_{1\wedge2}=s_1 \wedge s_2$ is $\FFF$.
Furthermore we ask for the assignment of the truth value to be the same on equivalent statements. Finally, we want to assign to the absurd statement always the value $\FFF$ since we want to encode the idea that atomic statements are mutually exclusive and thus, for any two atomic statements $s_i,s_j$ with $i\neq j$ one expect expects that $\perp=s_i \wedge s_j$ has assigned truth value $\FFF$

Specifying the truth value of each atomic statement is enough for computing the truth value of every other statement in the lattice by simply requiring the assignment to be admissible. For this reason one can consider the truth value as propagating through the lattice starting from the atomic statements. It can be further verified that admissible assignments of the truth value are the one and only the one where there is one single true atomic statement and all the others are false.

 \begin{table}
\begin{tabular}{||c | c || c||} 
 \hline
 a & b & a $\wedge$ b\\ [0.5ex] 
 \hline
\textbf{T} &\textbf{T}&\textbf{T}\\
\textbf{T}&\textbf{F}&\textbf{F}\\
\textbf{F}&\textbf{T}&\textbf{F}\\
\textbf{F}&\textbf{F}&\textbf{F}\\
 \hline
\end{tabular}$\qquad$ \begin{tabular}{||c | c || c||} 
 \hline
 a & b & a $\vee$ b\\ [0.5ex] 
 \hline
\textbf{T}&\textbf{T}&\textbf{T}\\
\textbf{T}&\textbf{F}&\textbf{T}\\
\textbf{F}&\textbf{T}&\textbf{T}\\
\textbf{F}&\textbf{F}&\textbf{F}\\
 \hline
\end{tabular}$\qquad$ \begin{tabular}{||c || c||} 
 \hline
 a  & $\neg$ a\\ [0.5ex] 
 \hline
\textbf{T}&\textbf{F}\\
\textbf{F}&\textbf{T}\\
 \hline
\end{tabular}
\caption{Truth tables.}
\label{table:truth-table}
\end{table}

%They say, for example, that if $s_1$ is true and $s_2$ is false, then $s_{1\wedge2}=s_1 \wedge s_2$ is false. We say that a truth assignment is admissible if it is compatible with the truth assignments In general for boolean lattices $\top$ is always true. 
%This is because the only truth assignments compatible with a lattice where $\top$ is false is the trivial lattice where all the atomic elements are false.

We attach to each statement an \textit{additional} label that we call \textit{accessibility value} of the statement. This label can assume two values $\AAA$ and $\NNN$, read as "experimentally accessible" and "experimentally inaccessible" respectively.
An assignment of accessibility value to each element of the lattice is considered admissible if it respects the following rules. First, it must be consistent with the rule specified by the accessibility tables  \ref{table:accessibility-table}. Furthermore we ask for the assignment of the accessibility value to be the same on equivalent statements.  
 
 \begin{table}
\begin{tabular}{||c | c || c||} 
 \hline
 a & b & a $\wedge$ b\\ [0.5ex] 
 \hline\hline
\textbf{A}&\textbf{A}&\textbf{A}\\
\textbf{A}&\textbf{N}&\textbf{A/N}\\
\textbf{N}&\textbf{A}&\textbf{A/N}\\
\textbf{N}&\textbf{N}&\textbf{A/N}\\
 \hline
\end{tabular}$\qquad$ \begin{tabular}{||c | c || c||} 
 \hline
 a & b & a $\vee$ b\\ [0.5ex] 
 \hline\hline
\textbf{A}&\textbf{A}&\textbf{A}\\
\textbf{A}&\textbf{N}&\textbf{A/N}\\
\textbf{N}&\textbf{A}&\textbf{A/N}\\
\textbf{N}&\textbf{N}&\textbf{A/N}\\
 \hline
\end{tabular}$\qquad$ \begin{tabular}{||c || c||} 
 \hline
 a  & $\neg$ a\\ [0.5ex] 
 \hline
\textbf{A}&\textbf{A}\\
\textbf{N}&\textbf{N}\\
 \hline
\end{tabular}
\caption{Accessibility tables. When written $A/N$ it means that the accessibility value is not defined and can be chosen. Notice that composition and negations of experimentally accessible statements always returns an experimentally accessible statement.}
\label{table:accessibility-table}
\end{table}

We expect that by combining or negating experimentally accessible statements we should obtain an experimentally accessible statement. Furthermore, we should also expect that negating a statement should not change its accessibility value otherwise one would easily gain access to an inaccessible statement just by considering its negation. 
In general we are going to consider lattices for which $\top$ is always experimentally accessible. This is because the lattice where $\top$ is experimentally inaccessible is the trivial lattice where all the statements are experimentally inaccessible (see lemma \ref{lem:top-is-accessible} and its proof in appendix \ref{appendix:proofs}).
We note that the accessibility value of the elements of a lattice is not completely specified by the accessibility value of the atomic elements.

\section{Logical inaccessible information models}

\begin{figure*}[t]
	\centering
	\includegraphics[width=0.7\textwidth]{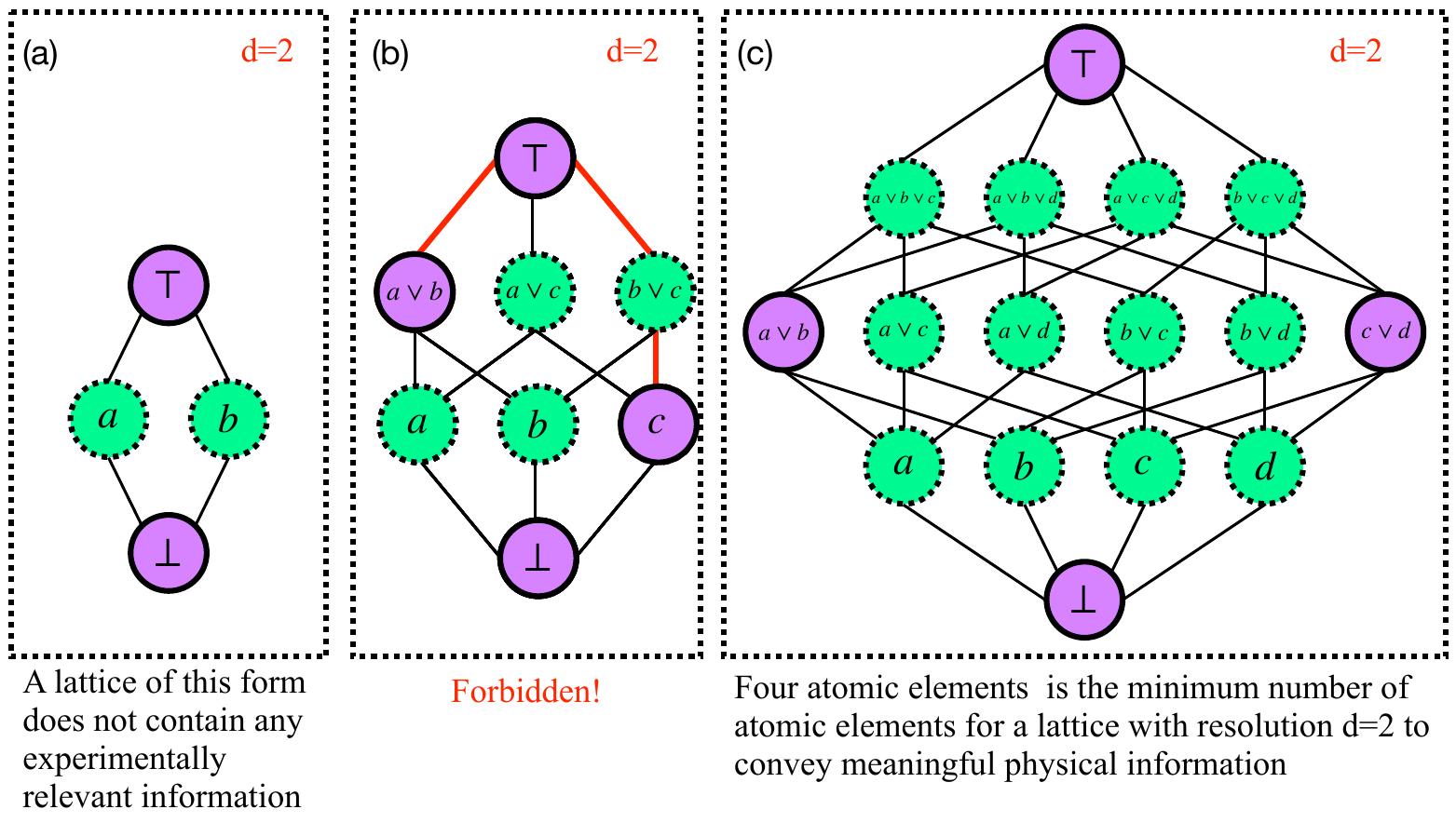}
	\caption{Pictorial representations of the three examples. The statements are represented by coloured balls. The colour of the ball represents its accessibility value. A purple ball corresponds to an experimentally accessible statement, a green ball corresponds to an experimentally inaccessible statement. We represents models of accessibility-depth $d=2$ for different dimensions $D$. In panel \textbf{(a)} is represented the case of Example \ref{ex:1}. For $D=2$ the only accessible statements are $\perp$ and $\top$, thus this is a useless model as it does not convey any experimentally relevant information. In panel \textbf{(b)} is represented the case of Example \ref{ex:2}. For $D=3$, the only accessible statements are $\perp$ and $\top$, thus this is a useless model. In the panel a step of the proof is highlighted: choosing the statement $a\vee b$ to be accessible would imply its negation $\neg(a\vee b)=c$   to be accessible, but $c$ cannot be accessible by definition, thus we have a contradiction. In panel \textbf{(c)} is represented the case of Example \ref{ex:3}. This is the first non-trivial model. From the figure it is easy to notice how the accessible (purple) statements form a sub-lattice with $D=2$ of a classical model.}
	\label{fig:Fig1}
\end{figure*}

\label{sec:assumption1}

In the previous section, we introduced the framework for discussing accessible and inaccessible statements and the $\mathcal{I}\mathcal{H}$. In this section we add some structure to the definitions of the last section. We informally state here the main assumption of the theory of inaccessible information.

%\correction{In the previous section, we introduced the framework for discussing accessible and inaccessible statements. Now, we can formally define the $\mathcal{I}\mathcal{H}$ as Assumption 1 of our theory and explore its implications.}{In this section we add some structure to the definitions of the last section. We informally state here the main assumption of the theory of inaccessible information.}

\begin{BoxWT_green}{Assumption 1 for the logical lattices}
Only disjunctions of at least $d$ atomic statements are accessible, for some $d$. That is, our knowledge is limited to be always uncertain about a fixed number of statements $d$. We call $d$ the \textit{accessibility-depth} of the model.
\end{BoxWT_green}

%\correction{Assumption 1 formalises $\mathcal{I}\mathcal{H}$ within the context of our theory, stating that experimentally accessible information is a coarse-grained representation of some experimentally inaccessible information.
%}{The idea behind assumption $1$ is that our theory should describe systems where the experimentally accessible information is a coarse graining of some experimentally inaccessible information.
%\begin{quote}
%	\textbf{Assumption 1:} Every inaccessible information model has a fixed \textit{accessibility-depth} $d$. The accessibility-depth of a model is the minimum number of experimentally inaccessible statements necessary to form an experimentally accessible statement.
%\end{quote}
%This is the intuition guiding us, but pragmatically we are just exploring what are the consequences of imagining a theory based on this assumption, that is a theory where some information is inaccessible.}

The idea behind assumption $1$ is that our theory should describe systems where the experimentally accessible information is a coarse graining of some experimentally inaccessible information.
%\begin{quote}
%	\textbf{Assumption 1:} Every inaccessible information model has a fixed \textit{accessibility-depth} $d$. The accessibility-depth of a model is the minimum number of experimentally inaccessible statements necessary to form an experimentally accessible statement.
%\end{quote}
This is the intuition guiding us, but pragmatically we are just exploring what are the consequences of imagining a theory based on this assumption, that is a theory where some information is inaccessible.

We have that the accessibility-depth $\mathbb{N}\ni d \geq 1$ is the minimum number of atomic statements necessary to form an experimentally accessible statements through the connective $\vee$. Alternatively, numbering the level of the lattice from the bottom (with $\perp$ being at level $0$), $d$ is the level of the lattice at which experimentally accessible statements are allowed (see appendix \ref{appendix:graphical-notation}). 

\begin{definition}
A model $(D,d)$ is a list $(L_D,A_d,V)$, where $L_D$ is the Boolean algebra of statements, or Boolean lattice, introduced beforehand (with  $D$ atomic elements),  $A_d$ is the set of all “accessibility-assignments” compatible with the requirement that the “accessibility-depth” is $d$, and  $V$ is the set of admissible “truth-assignments” constrained in the aforementioned fashion. 
\end{definition} 

\begin{definition}
A \textit{configuration} of a model is $(L_D,a_d,v)$ with $a_d\in A_d$ and $v \in V$.
\end{definition} 

For each model the allowed configurations are the ones compatible with the rule of the accessibility tables and with the accessibility-depth $d$ of the model.  

 For example for the model $(D,d)=(4,2)$ we have that the choice of the accessibility-depth $d=2$ imposes that all the atomic statements are experimentally inaccessible while allowing statements $s_i\vee s_j$, with $i,j \in \{1,2,3,4\}$ and $i\neq j$ to be experimentally accessible. Furthermore the request of being compatible with the rules of the accessibility tables imposes further constraints as we will see soon with some examples.

For a given model, we will prefer configurations that maximise the number of experimentally accessible statements. We call these configurations \textit{ideal configurations} and we define them as follows

\begin{definition}
An \textit{ideal configuration} of a model $(D,d)$ is an allowed configuration with the assignments of accessibility values that maximises the number of accessible statements at level $d$. 
 \end{definition}

What distinguishes a theory of inaccessible information from a classical theory, is the distinction between accessible and inaccessible statements. Dropping this distinction should allow us to recover a classical theory. This corresponds to setting the accessibility-depth of the model to be $d=1$ and considering the unique ideal configuration that assigns accessible value to each and every statement of the theory.
 
 \begin{definition}
 	A \textit{classical model} of dimension $m$ is a model of inaccessible information with $(D=m,d=1)$ in an ideal configuration.
  \end{definition}
  
  We reserve the use of the symbol $m$ to denote the dimension of classical models.
 
 In fact, for $d=1$, the assignment of accessibility values that maximises the number of accessible statements at level $d$ corresponds to simply choosing all the atomic statements to be accessible and thus recovering a theory without an effective accessibility label. 
 
 Configurations for which the only experimentally accessible statements are $\top$ and its negation $\perp$ do not convey any relevant physical information. These are the only allowed configurations for models with $D=d$. We call these models \textit{useless models}.

We can see through the following examples that Assumption $1$ is already strong enough to reveal some interesting properties of the theory.

\begin{example}
\label{ex:1}
Consider the models with $D=2$ atoms.
It is easy to check that the only allowed models are the classical model and the useless model.
Referring to panel $(a)$ of Figure \ref{fig:Fig1} we can see that $``a"$ is the negation of $``b"$ thus $``a"$ and $``b"$ must have the same accessibility value. We have that either both $``a"$ and $``b"$ must be accessible (classical model) or both $``a"$ and $``b"$ must be inaccessible (useless model).
\end{example}

\begin{example}
\label{ex:2}
Consider a lattice with $D=3$ in an ideal configuration.
For $d=1$ it is a classical model, for $d=3$ it is a useless model. We study the case $d=2$.
Because of the structure of the lattice we have that $s_2\llor s_3=\neg s_1$. Since $d=2$ we have that $s_1$ is $\NNN$, this implies that $s_2\llor s_3$ must also be $\NNN$. We can proceed with the same reasoning for $\neg p_2$ and $\neg p_3$ to find out that, if $d=2$ then no configurations with accessible statements at level $3$ are allowed. 
Thus a lattice with three atomic elements can only effectively be a classical or a useless model.
See Figure \ref{fig:Fig1} for a pictorial representation.

\end{example}

\begin{example}
\label{ex:3}
Consider a lattice with $D=4$ in an ideal configuration.
This is the first case of a non-trivial lattice. In fact we can consider a model with accessibility-depth $d=2$. Without loss of generality we can assign to $s_1\llor s_2$ the accessibility value $\AAA$. This would imply that $s_3\llor s_4=\neg(s_1\llor s_2)$ is also $\AAA$.
At the same time $s_1 \llor s_3$ cannot be $\AAA$. In fact, if $s_1 \llor s_3$ were $\AAA$ we would have that $s_1=(s_1 \llor s_3)\lland (s_1 \llor s_2)$ is also $\AAA$. This would contradict $d=2$. See Figure \ref{fig:Fig1} for a pictorial representation.
\end{example}

%We each value of $d$ we define its \textit{main model} the model 
%
%\begin{lemma}
%	For each $d$ the the smallest $D$ for which the lattice $(D,d)$ contains a classical sublattice
%\end{lemma}
%	

In appendix \ref{appendix:proofs} we prove the following lemmas:

%\js{This lemmas section has to be verbalised, narrativised.}

 \begin{restatable}{lemma}{topisaccessible}
 	If $\top$ is $\NNN$ then all the statements of the lattice are $\NNN$.
 	\label{lem:top-is-accessible}
 \end{restatable}

 \begin{restatable}{lemma}{notintersect}
 \label{lem:not-intersect}
 For each model $(D,d)$,the sets of atomic statements out of which two accessible statements at level $d$ of an admissible configuration are composed must not have any element in common. \end{restatable}

  \begin{restatable}{lemma}{integermultiples}
\label{lem:integer-multiples}
	For each model $(D,d)$, the accessible statements at level $d$  in an ideal configurations are $\frac{D}{d}$ if $\bmod(D,d)=0$, and $\lerifloor{\frac{D}{d}}-1$ if $\bmod(D,d)\neq0$. 
\end{restatable}

  \begin{restatable}{lemma}{uniquenessidealconfigurations}
\label{lem:uniqueness_ideal_configurations}
	For each model $(D,d)$, the ideal configurations are all equivalent up to relabelling of the atomic statements.
\end{restatable}

In the following whenever we will specify a model we will consider it in its ideal configuration,  since this configuration maximises the experimentally relevant information of the model. 

% TODO Maybe put this information later, when you talk about the tension for maximising expressibility of the theory.

 As seen in example \ref{ex:3}, it is possible to embed a classical system inside an ideal configuration of a model of inaccessible information. In this case we can say that the model of inaccessible information is an \textit{inflation of the classical model}.

 \begin{definition}
An \textit{inflation of a classical model} of dimension $m$ is any model $(D,d)$  in ideal configuration with exactly $m$ accessible statements at level $d$.
 \end{definition}

 \begin{restatable}{lemma}{classicalsublattices}
 \label{lem:classical-sub-lattices}
 Every classical model with $m$ atomic elements can be inflated only to models $(D,d)$ where, for each accessibility-depth $d$, the set of allowed $D$ is
\begin{equation}
D\in\mathcal{D}_{m,d}\coloneqq\{md\}\cup \bigcup_{i=1}^{m-1}\{(m+1)d+i\}
\end{equation}
 \end{restatable}
 
 From lemma \ref{lem:classical-sub-lattices} we have that classical models cannot be inflated to models with arbitrary number of atomic statements $D$. For example, the classical model with $m=2$ cannot be inflated to a model of inaccessible information with $D=5$.

We aim to interpret models of inaccessible informations as refinements of classical models. Specifically, we're intrigued by the implications of considering that each classical model might be a part within a larger model of inaccessible information.  For this reason, from now on we are going to focus on models that are inflations of classical models. In the next section we are going to characterise them.

% TODO the lesson here is that the best models are inflations of classical models.

\section{Composition of models}
\label{sec:composition-of-models}

\begin{figure}[h!]
	\centering
	\includegraphics[width=1\linewidth]{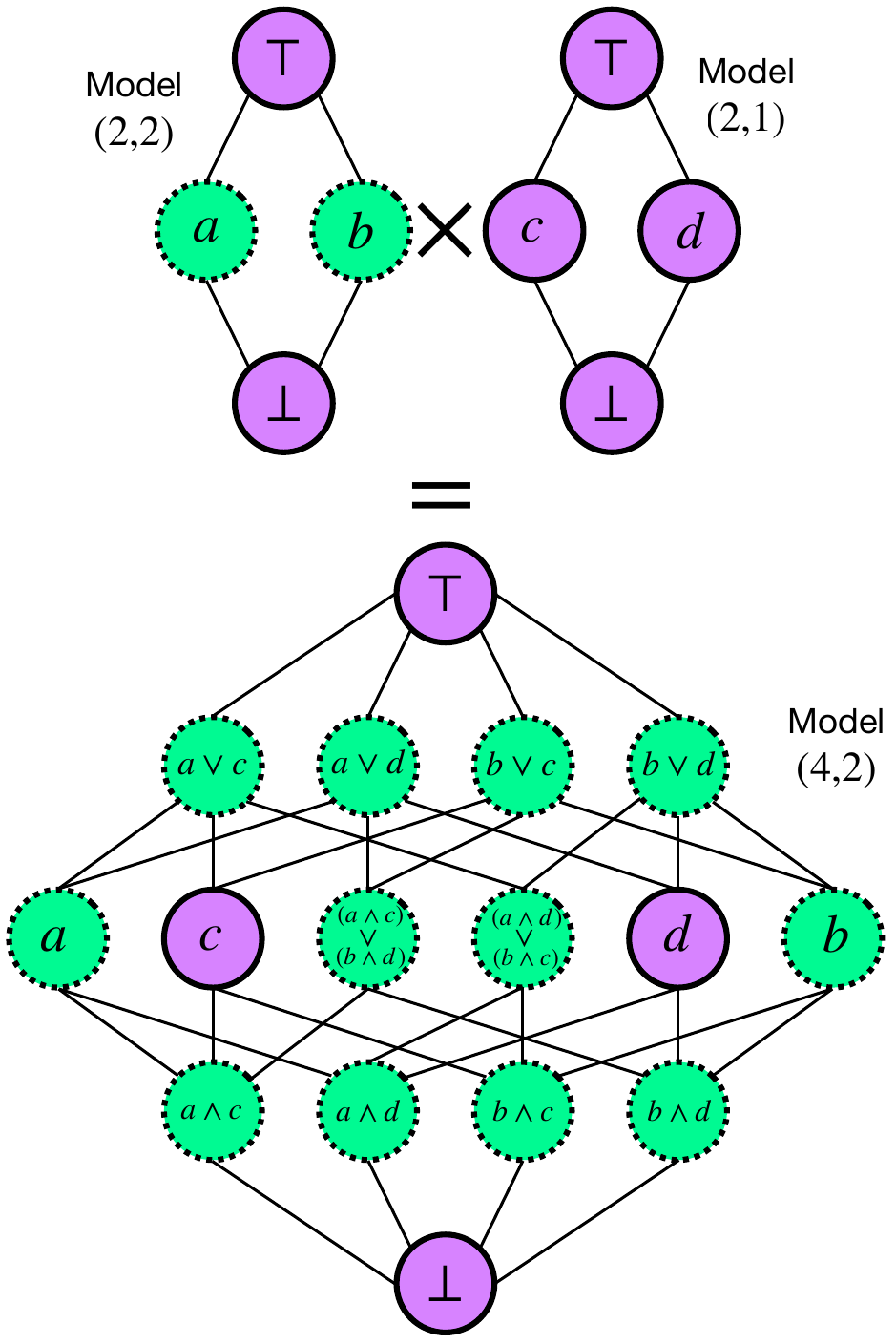}
	\caption{Composition of model $(2,2)$ with model $(2,1)$.}
	\label{fig:composition-vertical}
\end{figure}	

It is not always possible to have a full picture of the system we want to describe straight from the beginning. One usually describes the many different parts of a system one after the other, hoping to reach a complete description through a process of reconstruction. For this reason, guided by this principle of reductionism, we might be interested in having a rule for joining different models in a single model. In this section we are going to see how to glue together two inaccessible information models.

The composition of two models $(D_1,d_1)$ and $(D_2,d_2)$ with atomic statements $\{r_1,r_2,\dots,r_{D_1}\}$ and $\{s_1,s_2,\dots,s_{D_2}\}$ is is a new lattice of statements with $D_1 D_2$ atomic elements $\{p_{\alpha}\}_{\alpha=1,\dots,D_1\cdot D_2}$ and accessibility-depth $d_1 d_2$. We write $(D_1,d_1)\times (D_2,d_2)=(D_1 D_2, d_1 d_2)$. Each atomic element of the composed lattice is the product of the atomic elements of the two original lattices as
\begin{equation}
	\{p_{\alpha}\}_{\alpha=1,\dots,D_1\cdot D_2}\coloneqq\{r_i \cdot s_j\}_{i=1,\dots,D_1 \atop j=1,\dots,D_2},
\end{equation}
where $\alpha=(i-1)D_2+j$.

The composition of two models can be thought as the model where the atomic elements are all the combinations of one element of the first model and one element of the second model connected by the binary operator $\wedge$. This  is the reason why the number of atomic elements multiply.
Concurrently, we expect for the experimentally accessible statements of the two models to still be experimentally accessible in the composed model. For example consider the statement $r_1 \vee \dots  \vee r_{d_1}$ at level $d_1$ of the first lattice to be experimentally accessible. When composing the two lattices we expect for this statement to still be experimentally accessible. Statements that were respectively at level $d_1$ and $d_2$ of the two models, are at level $d_1 d_2$ of the composed model. 

\begin{example}
\label{ex:composition}
In figure \ref{fig:composition-vertical} we represent the composition of a $(2,2)$ model with a $(2,1)$. The composite model is a $(2\cdot 2=4,2\cdot 1=2)$ model with accessibility values inherited from the two initial models.
\end{example}

Since we focus on model of inaccessible informations that are inflations of classical models, we want to find a general procedure that associates to each classical model of dimension $m$ an inaccessible information model such that this procedure is compatible with the composition of models. We call this procedure an \textit{inflation procedure}. Asking for the inflation procedure to be compatible with the composition of models means that we want the composition of inflations of classical models to be equivalent to the inflation of the composition of classical model.
This procedure is given in the following lemma proved in the appendix \ref{appendix:proofs}.

\begin{restatable}{lemma}{scalingdimension}
\label{lem:scaling-dimension}
There exists one unique family of inflation procedures, parametrised by $c\in\mathbb{N}_{>0}$. Each of these procedures associates to each classical model of dimension $m$ the inaccessible information model $(m^{c+1},m^c)$.
\end{restatable}

As stated at the end of last section, we are interested in inaccessible information models that are inflations of classical models. 
Thus, we would like that for any chosen accessibility-depth $d$, there exists a model $(D,d)$ and, consequently, a set of $D$ atomic statements, such that the ideal configuration of the model $(D,d)$ is an inflation of a classical model. By asking for an inflation procedure with this additional property we obtain that there exists a unique inflation procedure, the one obtained setting $c=1$, since choosing $d=2$ the equation $m^c=2$ with $m,c\in \mathbb{N}$ has the unique solution $m=2,c=1$. We define this as the standard inflation procedure.
\begin{definition}
The general inflation procedure is defined as the map
\begin{equation}
m \to (m^2,m),
\end{equation}
mapping a classical model of dimension $m$ to the model of inaccessible information $(m^2,m)$.
\end{definition}

Thus we have that for inflations of classical models, the number of atomic elements have to scale as the square of the accessibility-depth of the model. The model of panel $(c)$ of Figure \ref{fig:Fig1} can be interpreted as an inflation of a classical model with dimension $m=2$, or as an inflation for accessibility-depth $d=2$.

\begin{observation}
	Since our focus is on models that are inflations of a classical model, from now on we are going to identify a model uniquely by its accessibility-depth $d$, so that a model of accessibility-depth $d$ is a $(d^2,d)$ model.
\end{observation}

In the following we are going to study the properties of systems with $d$ prime. Systems with $d$ not prime can be thought as unique compositions of systems with $d$ prime.

At this point, this theory seems to not be of any practical use. In fact, there is no real communication between the classical sub-lattice and the whole lattice. One can argue that the formalism of the theory of inaccessible information is just an overparameterised method for  describing a classical model. 

Consider for example the model $(4,2)$. To describe a specific configuration of this model, one has to specify the truth value of $4$ inaccessible elements, the model can then be used to express the accessible information of a classical lattice of dimension just $2$.  
%Each configuration of this model is characterised by the truth value of $4$ inaccessible elements. Each configuration expresses the accessible information of a classical lattice of dimension just $2$. 
Considering that one is interested in the description of experimentally accessible information, one is describing a two parameters space with four parameters. This freedom is reflected in the fact that more than one configuration of the inaccessible truth value corresponds to the same configuration of accessible truth values. One has, for example, that the configuration where $a\vee b$ is true and $c \vee d$ is false corresponds to the two atomic elements  $(a,b,c,d)$ configurations $(\FFF,\TTT,\FFF,\FFF)$,$(\TTT,\FFF,\FFF,\FFF)$. Thus, for the moment,  the model with inaccessible information is just an over-parametrisation.

Moreover, the structure of the theory is so rigid that, for an ideal model, there is no way to change the accessibility of any statement from $\NNN$ to $\AAA$ without loosing the consistency with the rules of the accessibility table.

In the next section we are going to relax the compositional rules and a richer structure will be possible.

 \section{From the logical lattice to quasi-probability}
 \label{sec:relaxing}
We discussed the concept of inaccessible information in the context of binary logic, where a binary truth value is assigned to each statement. However the modern language of science is not restricted to binary assignments of truth values, but it rather uses the language of probability. The passage to this modern language historically happened through a lengthy development process.
However, using the ideas firstly advanced by R. T. Cox  \cite{cox1946b,cox1961}  one can attribute to this paradigm shift a rationale and formally see probability theory as a generalisation of logic. In fact, one can derive the rules of probability theory from the logic of statements when a concept of "plausibility" is introduced and requested to behave consistently with the rule of logic. Important to this point is that probability theory is not assumed, but it is rather derived from logical axioms and what Jaynes calls "common sense" \cite{jaynes2003b,jaynes1988}. In this work, inspired by these seminal works, we proceed following a similar path with the aim of extending the domain of applicability of this construction to the theory of inaccessible information. Our procedure differs from the canonical one  of R. T. Cox and E. T. Jaynes, as we choose to reduce the assumptions of  "common sense" to the bare minimum. What we obtain is a generalisation of binary logic that behaves as a theory of probability with the difference that the "probability" assignments can be negative and greater than one. For this reason we say that, generalising binary logic while preserving consistency with the structure of the algebra of statements we obtain a theory of quasi-probability. We refer to this procedure as the relaxation of the logical lattice models to quasi-probability models.
A full derivation is given in the appendix \ref{appendix:quasi-probabilities-from-lattice} and the main differences between our derivation and the canonical derivation of R.T. Cox and E. T. Jaynes are highlighted in appendix \ref{appendix:additional-assumptions}.
A key result of the relaxation to quasi-probability models is the construction of a function $\tilde{Q}(y \mid x)$ that takes as argument two statements $x$ and $y$ and return a real number:
\begin{align}
	\tilde{Q}(y \mid x)=\left\{\begin{array}{l}
1 \text { if } (x\wedge y = x)\wedge(x \vee y=y) \\
0 \text { if } x \wedge y=\perp \\
p \text { otherwise, where } p\in \mathbb{R},p \neq 0,p \neq 1.
\end{array}\right.
\end{align}
The function $\tilde{Q}(y \mid x)$ returns 1 if $x$ implies $y$ in the logical sense, and return $0$ if $x$ and $y$ are "uncorrelated" statements.
Asking for function $\tilde{Q}$ to be consistent with the rules of the composition of statements, we find that $\tilde{Q}$ has the following properties
\begin{enumerate}
	\item the sum rule: \begin{equation} 
		\tilde{Q}(x\llor y| z)= \tilde{Q}(x|z) +\tilde{Q}(y|z)-\tilde{Q}(x\lland y | z),
		\end{equation}
	\item the product rule: \begin{equation}
		\tilde{Q}(x \lland y | z)= \tilde{Q}(x| z)\tilde{Q}(y|x \wedge z),
		\end{equation}
	\item the Bayes' theorem: \begin{equation}
		\tilde{Q}(y|x\lland z)=\frac{\tilde{Q}(y|z)\tilde{Q}(x|y\lland z)}{\tilde{Q}(x|z)}.
	\end{equation}
\end{enumerate}
This shows that, to be consistent with the rules of composition of statements, the function $\tilde{Q}$ has to satisfy the properties of a conditional probability function. The difference between $\tilde{Q}$ and a conditional probability function is that $\tilde{Q}$ is not constrained to return values in the interval $[0,1]$.

Using the function $\tilde{Q}$ we can proceed with the relaxation to a quasi-probability model by replacing the binary truth labels $\TTT$ and $\FFF$ associated to each statement $s_i$ with a real number $Q(s_i):=\tilde{Q}(s_i\mid \top)\in\mathbb{R}$ such that the following properties hold:
\begin{itemize}
\item $Q(\top)=1$,
\item $Q(\perp)=0$,
\item $Q(s_i)\in \mathbb{R}$,
\item $\sum_{s\in \mathcal{A}(\mathcal{S})}Q(s)=1$,
\item $Q(s_i \vee s_j)=Q(s_i)+Q(s_j)-Q(s_i\wedge s_j)$,
\item $Q(s_i \wedge s_j)= Q(s_j)Q(s_i\mid s_j)$,
\end{itemize}
where $\mathcal{A}(\mathcal{S})$ is the set of atomic elements of a set of statements $\mathcal{S}$ (see appendix \ref{appendix:algebra-of-statements}).
The  value of $Q(s_i)$ is the quasi-probability assigned to the statement $s_i$.
Because of the property of $Q$ the quasi-probability assignment of every composite statement is completely specified by the assignment of $Q$ on the atomic statements $\mathcal{S}=\{s_1,s_2,\dots,s_{d^2}\}$.
Thus a configuration of the quasi-probability model of the theory of inaccessible information is completely characterised by the vector
\begin{align}
	\vec{q}&\coloneqq \{Q(s_1),Q(s_2),\dots,Q(s_{d^2})\}.
\end{align}
The explicit value of $Q(s)$ on a generic statement $s=s_i\vee s_j \vee \dots \vee s_k$, where ${s_i,s_j,\dots,s_k}$ are atomic statements, can be easily computed from the knowledge of $\vec{q}$ as  
\begin{align}
\label{eq:propagation-of-q}
	Q(s_i\vee s_j \vee \dots \vee s_k)&=Q(s_i)+ Q(s_j) + \dots + Q(s_k)\\
	&=q_i+q_j+\dots +q_k.
\end{align}
Using this, one can verify that $\vec{q}$ is a proper quasi-probability vector, in fact:
\begin{equation}
1=Q(\top)= \{Q(s_1),Q(s_2),\dots,Q(s_{d^2})\}=\sum_{i=1}^{d^2}q_i.
\end{equation}
The fact that we obtain an assignment of quasi-probabilities instead of probabilities is in complete agreement with the key ideas of the theory of inaccessible information. Experimentally inaccessible statements are not physically defined object and thus, since they cannot be measured, it makes no sense to associate to them a probability distribution. 
Nevertheless, as stated in the previous section, we are focusing on models of inaccessible information that are inflations of classical models. Thus every model will have some statements that are accessible. We expect that to each accessible statement we can associate a well defined probability. 
Thus, to ensure consistency, for any model with accessibility-depth $d$ we ask for the function $\tilde{Q}$ restricted to the classical sub-lattice to behave as a function $\tilde{P}$ defined as
\begin{align}
	\tilde{P}(y \mid x)=\left\{\begin{array}{l}
1 \text { if } x \leq y \\
0 \text { if } x \wedge y=1 \\
p \text { otherwise, where } p\in (0,1),
\end{array}\right.
\end{align}
 allowing us to recover standard probability theory on the experimentally accessible statements, i.e. ${P(s)\coloneqq \tilde{P}(s\mid \top)}$ is positive for accessible elements.
 This restriction imposes some constraints on the admissible $\vec{q}$.
For example, if $s_1,s_2$ are two atomic statements, asking for $s_1\vee s_2$ to be experimentally accessible, corresponds to asking for $P(s_1 \vee s_2)=Q(s_1 \vee s_2)=q_1+q_2\in[0,1]$.

We notice that this restriction on the admissible quasi-probability assignments is effectively connecting the two notions of accessibility and truth value that in the logical lattice were decoupled. Moving from the logical lattice to a quasi-probabilistic lattice we have relaxed both the truth value assignments and the accessibility value assignments. We relaxed the binary truth value to a continuous value specified by a function of the element of the lattice.
At the same time, requesting for the consistency of the probability assignment with the classical sub-lattice, the accessibility label, distinguishing between $\AAA$ and $\NNN$, relaxes to the distinction between statements being characterised by probabilistic or quasi-probabilistic values.

In the logical lattice truth values and accessibility values were not communicating. In the quasi-probability lattice the relaxed version of truth values and accessibility values become intertwined creating a rich structure of allowed distributions. In the quasi-probability models, the constraints on allowed configurations are not dictated by consistency with truth and accessibility tables, but by a series of inequalities and equalities determining the probability or quasi-probability character of the assignments of $Q$ to every statement. These assignments propagate in the lattice through the simple rule of equation \eqref{eq:propagation-of-q}.

%We want to generalise this idea to the function $Q$ and lattices where the elements have the additional property of being $\AAA$ or $\NNN$. The properties of function $Q$ that we have just noted tell us that the vector of the values that $Q$ assumes on the atomic elements of the lattice is indeed a quasi-probability vector $\vec{q}$. By the theorem of Rota, the value of $Q$ on every element of the lattice is completely characterised by $\vec{q}$.

We notice that by restricting to models that are the inflation of classical models, we always have a proper probability vector at the level of the accessibility-depth $d$. In fact, at level $d$, all the accessible statements of the logical lattice are non-intersecting and their composition with the binary operator $\vee$ is $\top$, this means that the sum of their probability assignment is $1$. 
 
Before moving on to the next section, we restate here that the quasi-probability assignments of a model are completely characterised by the values of the vector $\vec{q}$. We call $\vec{q}$ the \textit{state} of the model.

It will be the purpose of the next sections to explicitly characterise the constraints that characterise the allowed states $\vec{q}$.
\section{The quasi-probability models}
% TODO Cambiare quasi-probability lattice con qualcosa di altro. Quasi probability models?

%\js{Our assumption on the accessibility-depth of the system becomes the request for statements at level $d$ to be specified by probability distributions.}

\subsection{Maximal expressive standard model}
Consider the model with accessibility-depth $d=2$ in panel $(c)$ of Figure \ref{fig:Fig1}. Relaxing to the quasi-probability lattice we associate to the model a quasi-probability vector $\vec{q}=(q_1,q_2,q_3,q_4)=(Q(a),Q(b),Q(c),Q(d))$ describing the state of the model. As prescribed by our assumptions, on the classical sub-lattice we restrict to a probability distribution. This implies that the allowed vectors $\vec{q}$ are such that $Q(a\vee b)=P(a \vee b)\in [0,1]$ and $Q(c\vee d)=P(c \vee d)\in [0,1]$.
This imposes the constraints $q_1+q_2\in [0,1]$ and $q_3+q_4\in [0,1]$.

At the end of section \ref{sec:composition-of-models} we observed how the logical lattice structure was not helping us in modelling new kind of structures, but it was, instead, just an overparameterisation of some classical lattice. 

For the relaxed lattice this is not the case anymore. We can indeed restrict other elements of the lattice to be probability distributions without inconsistencies. 

\begin{example}
\label{ex:three-probabilities}
	Consider the model with accessibility-depth $d=2$ with associated $\vec{q}=(Q(a),Q(b),Q(c),Q(d))$.
	The vector  $\vec{p}^{[1]}=(Q(a \vee b), Q(c \vee d))$ on the classical sub-lattice with atomic elements $a \vee b$, $c \vee d$ is restricted to be a probability vector.
	
	In the logical lattice we could have defined the classical sub-lattice on $a \vee c$ and $b \vee d$, \textbf{instead} that on $a\vee b$ and $c \vee d$. 
	
	In the quasi-probability model we can restrict \textbf{also} $\vec{p}^{[2]}=(Q(a \vee c), Q(b \vee d))$ to be a probability vector, as well as $\vec{p}^{[3]}=(Q(a \vee d), Q(b \vee c))$.
	The choice of these three different probability vectors imposes the constraints
	\begin{align}
		q_i+q_j\in[0,1]
	\end{align}
	for each $i,j \in \{1,2,3,4\}$ and $i\neq j$.
	We can see that these constraints are not equivalent to the constraint $q_i\in[0,1]$ for each $i\in\{1,2,3,4\}$, thus $\vec{q}$ is still a quasi-probability vector with possibly negative entries.  
	Differently from the logical lattice, in the quasi-probability formulation we can have three accessible probability vectors at level $d$ while leaving the vector $\vec{q}$ free to take negative values and thus without constraining the model to be classical. 
	A graphical representation of this example is found in Figure \ref{fig:MESd2}.
\end{example}

Following the notation introduced in example \ref{ex:three-probabilities}, we denote the different probability vectors obtained at level $d$ with the upper-right index in square bracket. We call these probability vectors the \textit{ associated accessible probability vectors} of the model. The accessible probability vector associated with the accessible statements of the original logical lattice will always be denoted as $\vec{p}^{[1]}$.

In example \ref{ex:three-probabilities} we note that statements within the same probability vector combines to be a probability again. For example $\vec{p}^{[1]}_1+\vec{p}^{[1]}_2\in[0,1]$. Elements from different probability vectors are not assured to combine to a probability.  For example $\vec{p}^{[1]}_1+\vec{p}^{[2]}_1=q_1+q_2+q3$ can lie outside the interval $[0,1]$. 
\footnote{Mapping these properties to the case of the logical lattice, the different accessible probability vectors of the quasi-probability model correspond to different families of accessibility. In the case of example \ref{ex:three-probabilities} we would have that each statement has three accessibility labels. We call these labels accessibility-$1$, accessibility-$2$, accessibility-$3$. The statements of $\vec{p}^{[1]}$ are accessible for accessibility-$1$, while inaccessible for accessibility-$i\neq 1$. The composition of different accessibilities is not assured to be accessible. The rule of compositions do not mix accessibilities, in the sense that, while the composition of two accessible-$1$ statements is accessible-$1$, the composition of a statement accessible-$1$ with a statement accessible-$2$ is not guaranteed to be accessible as they are accessible with respect to two different accessibility labels.}

We see that there is not a hard constraint on the number of statements that we can constrain to be part of a probability vector. For a generic model with accessibility-depth $d$ we ask how many probability vectors $\vec{p}^{[i]}$ at level $d$  it makes sense to consider. We note that requesting for more than one accessible probability distribution in the lattice creates a tension between two properties of our models:
from one side, imposing more elements of the lattice to be restricted to probabilities increases the \textit{expressivity} of the model, in the sense that the more probability distributions are possible to be encoded in the model the more experimentally relevant information the model contains.
At the same time, requiring for more elements of the lattice  to be restricted to probabilities, imposes more restrictions on the admissible $\vec{q}$, these restrictions will then propagate to the possible probability distribution of the model and thus on the possible experimentally relevant information the model contains.
We want our theory of inaccessible information to be as expressive as possible imposing as few restriction on $\vec{q}$ as possible.

In particular we would like to have the least constrained theory for which each $\vec{q}$ is in one-to-one correspondence with a unique set of accessible probability vectors $\{\vec{p}^{[i]}\}_{i}$.

With this objective in mind, we define the maximal expressive standard models.
\begin{definition}[Maximal Expressive Standard models (MES models)]
\label{def:MESmodel}
	A model with accessibility-depth $d$, with $d$ prime, is a Maximal Expressive Standard (MES) model of accessibility-depth $d$ if it maximises the number of probability vectors $p^{[i]}$ at level $d$ with the constraint that every two statements from different accessible probability vectors at level $d$ have in common just one atomic statement.
\end{definition}

\begin{example}
	Consider the MES model with accessibility-depth $d=2$ and its three probability vectors $\vec{p}^{[1]},\vec{p}^{[2]},\vec{p}^{p[3]}$ (see Figure \ref{fig:MESd2}). Taking a statement of $\vec{p}^{[1]}$ and any other statement in $\vec{p}^{[2]}$ or $\vec{p}^{[3]}$, we have an overlap of at most one statement. For example we take $a \vee b$ from $\vec{p}^{[1]}$ and $b\vee d$ from $\vec{p}^{[2]}$, the overlap is just the statement $b$.
\end{example}

\begin{restatable}{lemma}{MESdplusone}
\label{lem:MESdplusone}
A MES model of accessibility-depth $d$ has $d+1$ probability vectors at level $d$.
\end{restatable}

\begin{definition}[Set of states of a MES model]
	\label{def:MESd}
For a MES model with accessibility-depth $d$, $\text{MES}_d$ is the the set of states $\vec{q}$ compatible with the model. 
\begin{equation}
\text{MES}_d=\left\{ \vec{q}\in\mathbb{R}^{d^2} \mid \sum_{i=1}^{d^2}q_i=1; \; p_j^{[k]}\in [0,1],\; \forall\; \substack{k=1,\dots,d+1\\ j=1,\dots,d}\right\}
\end{equation}
where $\{\vec{p}^{[k]}\}_{k=1}^{d+1}$ are the associated probability vectors.
\end{definition}

The definition of MES models is justified by the following lemma.

\begin{restatable}{lemma}{MESonetoone}
\label{lem:MES-1-to-1}
	MES models are such that there is a one-to-one correspondence between $\vec{q}$ and the accessible probability distributions at level $d$, with the minimum number of constraints on the admissible states $\vec{q}$.
\end{restatable}

%On the elements $a\vee c, a \vee d, b\vee c, b\vee d$ we have in general a quasi-probability distribution, in fact, the simple request that the the classical sub-lattice has to support a probability distribution is not strong enough to force a probability distribution on all the statements at level $d$.
%To be consistent with the idea behind the definition of the accessibility-depth $d$ we ask a reinforce version of the assumption. Namely, we ask for the lattice of quasi-probabilities that at level $d$ all statements must be described by a probability distribution.
%As an example we consider again the model $(4,2)$. At level $2$ we will have the three probabilities vectors $\vec{P}^{[1]},\vec{P}^{[2]},\vec{P}^{[3]}$ associated to the three different probability distributions $\vec{P}^{[1]}=(P(a\vee b),P(c \vee d))$, $\vec{P}^{[2]}=(P(a\vee c),P(b \vee d))$, $\vec{P}^{[3]}=(P(a\vee d),P(b \vee c))$.
%The assumption on the accessibility-depth of the models do not solely requires for the statements at level $d$ to be experimentally accessible, that is probability vectors at level $d$. It also requires for statements at lower level then $d$ to be experimentally inaccessible. 
%
\subsection{Assumption 1 in the quasi-probability lattice and the inaccessibility measure}
In this section we deal with a doubt one encounters relaxing from the logical lattice to the quasi-probability models. We decided to assign to the atomic elements of the lattice a vector of quasi-probabilities such that, through the propagation rules, this transforms to probability assignments for the associated probability vectors at level $d$.
The doubt arises from the fact that the set of probability vectors is a sub-set of the set of quasi-probability vectors. Thus, if we cannot  distinguish when we are assigning a vector of quasi-probabilities or a probability vector to the atomic elements of the lattice, how can the distinction between accessible and inaccessible information be reasonably expressed in the quasi-probability models?

 To solve this question consistently, we have to go back and remember the intuition behind assumption 1 and $\mathcal{I}\mathcal{H}$. We want to describe systems where each experimentally accessible statement is composed by $d$ experimentally inaccessible statements. In the logical lattice this fact was clearly captured by the labelling $\AAA$/$\NNN$ and how how the accessibility assignements propagates through the lattice following the rules specified by the accessibility tables. Take for example the model $(4,2)$ of figure \ref{fig:Fig1}. Knowing the statement $a \vee b$ is true, does not give us any clue in inferring which one of the two atomic statements $a$ or $b$ was true. We have complete ignorance about $d$ atomic elements. 
 Asking for all the statements below level $d$ to be experimentally inaccessible, can be read as asking for a complete uncertainty on at least $d$ elements. 
 
We can frame this same request in the language of probability asking for the total information contained in the probability distributions at level $d$ to be no greater than the information we would have if we had complete uncertainty about at least $d$ atomic elements out of $d^2$.
How do we measure this information? To address this question, we will examine some examples in order to identify the key properties that an appropriate information measure for this task should possess.

 %  	Consider the quasi-probability lattice model $(4,2)$ together with the specific assignment $\vec{q}=(1,0,0,0)$. The vector $\vec{q}$ is a proper quasi-probability vector, and $\vec{P}^{[1]}=(1,0)$, $\vec{P}^{[2]}=(1,0)$, $\vec{P}^{[3]}=(1,0)$ are proper probability vectors. Everything is consistent with the rules imposed so far, but not everything is consistent with our initial intuition. In fact, a $\vec{q}$ so defined implies a certainty for $a$, thus a value of truth for $a$ in the logical lattice. We want to constrain the admitted values for the quasi-probability vector $\vec{q}$ in a way that such problems do not occur.
 \begin{example}
 \label{ex:q-total-cert}
	 	Consider the MES model with $d=2$ together with the specific assignment $\vec{q}=(1,0,0,0)$. The vector $\vec{q}$ is a proper quasi-probability vector, and $\vec{p}^{[1]}=(1,0)$, $\vec{p}^{[2]}=(1,0)$, $\vec{p}^{[3]}=(1,0)$ are proper probability vectors. A $\vec{q}$ so defined implies a certainty for the statement $``a"$, thus a value of truth for $``a"$ in the logical lattice. This is forbidden.
\end{example}
 \begin{example}
 \label{ex:q-2-uncert}
	 	Consider the MES model with $d=2$ together with the specific assignment $\vec{q}=(1/2,1/2,0,0)$. The vector $\vec{q}$ is a proper quasi-probability vector, and $\vec{p}^{[1]}=(1,0)$, $\vec{p}^{[2]}=(1/2,1/2)$, $\vec{p}^{[3]}=(1/2,1/2)$ are proper probability vectors. A $\vec{q}$ so defined has complete uncertainty on two elements so it is consistent with the spirit of assumption 1.
\end{example}

 \begin{example}
	 	Consider the MES model with $d=2$ together with the specific assignment $\vec{q}^{(2)}=(1/2,1/2,0,0)$ and the MES model with $d=3$ together with the assignement $\vec{q}^{(3)}=(1/3,1/3,1/3,0,0,0,0,0,0)$. In the first model we have complete uncertainty over $2$ elements, in the second model we have complete uncertainty over $3$ elements. Now consider the composition of the two models. The state of the composition is described by the composition of $\vec{q}^{(2)}$ and $\vec{q}^{(3)}$ into the vector $\vec{q}^{(6)}=(1/6,1/6,1/6,1/6,1/6,1/6,0,0,\dots,0)$. In the composite model we have complete uncertainty over $6$ elements. We can see that the uncertainty is multiplicative.
\end{example}

In the previous examples we highlighted some desirable properties we would like for our information measure. In particular we noticed that the information measure we are looking for is multiplicative under composition of systems. 
In appendix \ref{appendix:distinguishability} we introduce the one parameter family of \textit{inaccessibility measures} $\uncert_c:\Delta_D\to\mathbb{R}_{+}$, where $c\in \mathbb{R}_{+}$ and $c\neq 1$, mapping a $D$-dimensional probability vector of the probability simplex $\Delta_D$ to a positive real number.\\
The action of the inaccessibility measure $\inacc_c(\vec{q})$ on a probability vector $\vec{p}$ is defined as
\begin{equation}
\indisti_{c}(\vec{p})=\frac{1}{\sqrt[c-1]{\sum_{i=1}^{D}p_i^c}}.
\end{equation}
An inaccessibility measure is a function of a probability vector that quantifies on how many elements we have complete ignorance, or better said, how many elements are inaccessible to our knowledge. For a probability vector $\vec{p}\in \Delta_D$ one has that for every $c$,  $1\leq \uncert_{c}(\vec{p})\leq D$. The main difference between Shannon entropy and the inaccessibility measures is that the latter ones are multiplicative, whereas Shannon entropy is additive. Considering examples \ref{ex:q-total-cert} and \ref{ex:q-2-uncert}, we have that for every $c$, $\uncert_c{(1,0,0,0)}=1$ and $\uncert_c{(1/2,1/2,0,0)}=2$. The derivation of the family $\uncert_c$ compatible with our assumptions can be found in appendix \ref{appendix:distinguishability}, while the comparison of $\uncert_c$ with other information measures is done in appendix \ref{appendix:comparisons-info-measure}.

In a theory of inaccessibile information we are interested in measuring the inaccessibility of quasi-probability vectors. For this reason, in appendix \ref{appendix:extension-to-quasi-prob} we extend the domain of the inaccessibility measure $\inacc$ to the space of quasi-probability states of MES models. Asking for compatibility with the structure of MES models we find that in a theory of inaccessible information there is a unique inaccessibility measure that exists by its own, \textit{the} inaccessibility measure $\inacc$ acting on quasi-probability vectors $\vec{q}$ as:
\begin{equation}
\inacc(\vec{q})=\frac{1}{\sum_{i}q_i^2}.
\end{equation} 
All the inaccessibility measures defined for $c>2$ depends on the inaccessibility measure $\inacc$.

%The assumption 1 for quasi-probability models asks for a MES model of accessibility-depth $d$ to have inaccessibility $\uncert(\vec{q})$ of any admissible state $\vec{q}$ lower bounded by $d$. 
%We are finally ready to \correction{formalise $\mathcal{I}\mathcal{H}$ as an }{state } assumption $1$ for quasi-probability models.

We are finally ready to state assumption $1$ for quasi-probability models.

\begin{BoxWT_purple}{Assumption 1 for quasi-probabilities}
\label{ass:assumption1-MES}
%The admissible states $\vec{q}$ of a MES models of accessibility-depth $d$ 
%%are the ones constrained by the $d+1$ accessible probability vectors at level $d$ and such that their 
% have indiscernibility lower bounded by $\uncert_2(\vec{q})\geq \uncert_2(\mus{d})=d$.
 The admissible states $\vec{q}$ of a MES model of accessibility-depth $d$ 
%are the ones constrained by the $d+1$ accessible probability vectors at level $d$ and such that their 
 have inaccessibility lower bounded by $d$, i.e. $\uncert(\vec{q})\geq d$.
\end{BoxWT_purple}

We have a formulation of assumption 1 that encodes the intuitive idea of accessible \textit{accessibility-depth} in the language of quasi-probabilities. 
%
%The relaxation of the logical lattice to the quasi-probability lattice introduced three major transformations of the ingredient of the theory: 
%\begin{itemize}
%\item to each statement is assigned a binary truth value $\rightarrow$ to each statement is assigned real number,
%\item to accessible statements is assigned the label  $\AAA$ to inaccessible statements is assigned the label $\NNN$ $\rightarrow$ to accessible statements is assigned a probability to inaccessible statements is assigned a quasi-probability,
%\item the existence of an accessibility-depth $d$ corresponds to an additional rule on the propagation of the accessibility value $\rightarrow$ the existence of an accessibility-depth $d$ corresponds to the imposition of a lower bound to the inaccessibility.
%\end{itemize}
%
%
Finally we can formulate a theory of inaccessible information.
(For convenience, a schematic comparison between logical lattice and the quasi-probability lattice can be found in table \ref{tab:riassunto} in appendix \ref{appendix:riassunto}.)

\section{A theory of inaccessible information}
A theory of inaccessible information is a MES model obeying assumption 1.
Assumption 1 for MES models (see definition \ref{ass:assumption1-MES}) restricts the admissible states $\vec{q}$ to the ones with inaccessibility $\indisti(\vec{q})$ bigger than the accessibility-depth $d$ of the model.
For a MES model with accessibility-depth $d$, states $\vec{q}$ such that $\indisti(\vec{q})\lneqq d$ are not admitted.
For each accessibility-depth $d$, we denote the set of admissible states as
\begin{equation}
	\mathcal{Q}_d=\Big\{ \vec{q}\in \text{MES}_d \mid \indisti(\vec{q})\geq d \Big\}.
\end{equation}
We define as pure states of the theory the states $\vec{q}\in \mathcal{Q}_d$ that maximise the knowledge on the system. These states are the one that minimise the inaccessibility.
\begin{definition}[Pure states]
	The pure states of a theory of inaccessible information of a MES model of accessibility-depth $d$ are all the admissible states $\vec{q}$ such that $\indisti(\vec{q})=d$. 
\end{definition} 
Because of lemma \ref{lem:quasi-concave} in appendix \ref{appendix:properties-inaccessibility} we know that $\indisti$ is a quasi-concave function over $\mathcal{Q}_d$, and thus, any convex combination of pure states is an admissible state with inaccessibility greater than or equal to $d$.

\section{The theory of inaccessible information with accessibility-depth d=2}

In this section we study the simplest model of our theory of inaccessible information. We study a MES model with accessibility-depth $d=2$.
The model is characterised by $D=d^2=4$ atomic statements $(a,b,c,d)$ to which we associate a state in the form of the quasi-probability vector $\vec{q}=(Q(a),Q(b),Q(c),Q(d))$.
Since we are considering a MES model with accessibility-depth $d=2$, at level $d=2$, from lemma \ref{lem:MESdplusone} we have $d+1=3$ associated accessible probability vectors. Following the construction appearing in the proof of lemma \ref{lem:MESdplusone} these associated accessible probability vectors are:
\begin{align}
	\vec{p}^{[1]}&=(P(a\vee b),P(c\vee d))=(q_1+q_2,q_3+q_4),\nonumber \\
	\vec{p}^{[2]}&=(P(a\vee c),P(b\vee d))=(q_1+q_3,q_2+q_4),\nonumber \\
	\vec{p}^{[3]}&=(P(a\vee d),P(b\vee c))=(q_1+q_4,q_2+q_3).
\end{align}
From an inspection of the lattice representation in figure \ref{fig:MESd2} it is possible to see how these probability distributions obey the \hyperref[def:MESmodel]{definition of MES model}, that is every two statements from different probability vectors at level $d$ have in common just one atomic statement.
It is also evident how the accessibility-depth $d=2$ is a special case, as $d=2$ is the only accessibility-depth for which all statements at level $d$ belong to an accessible probability distribution.

 \begin{figure}[h!]
	\centering
	\includegraphics[width=1\linewidth]{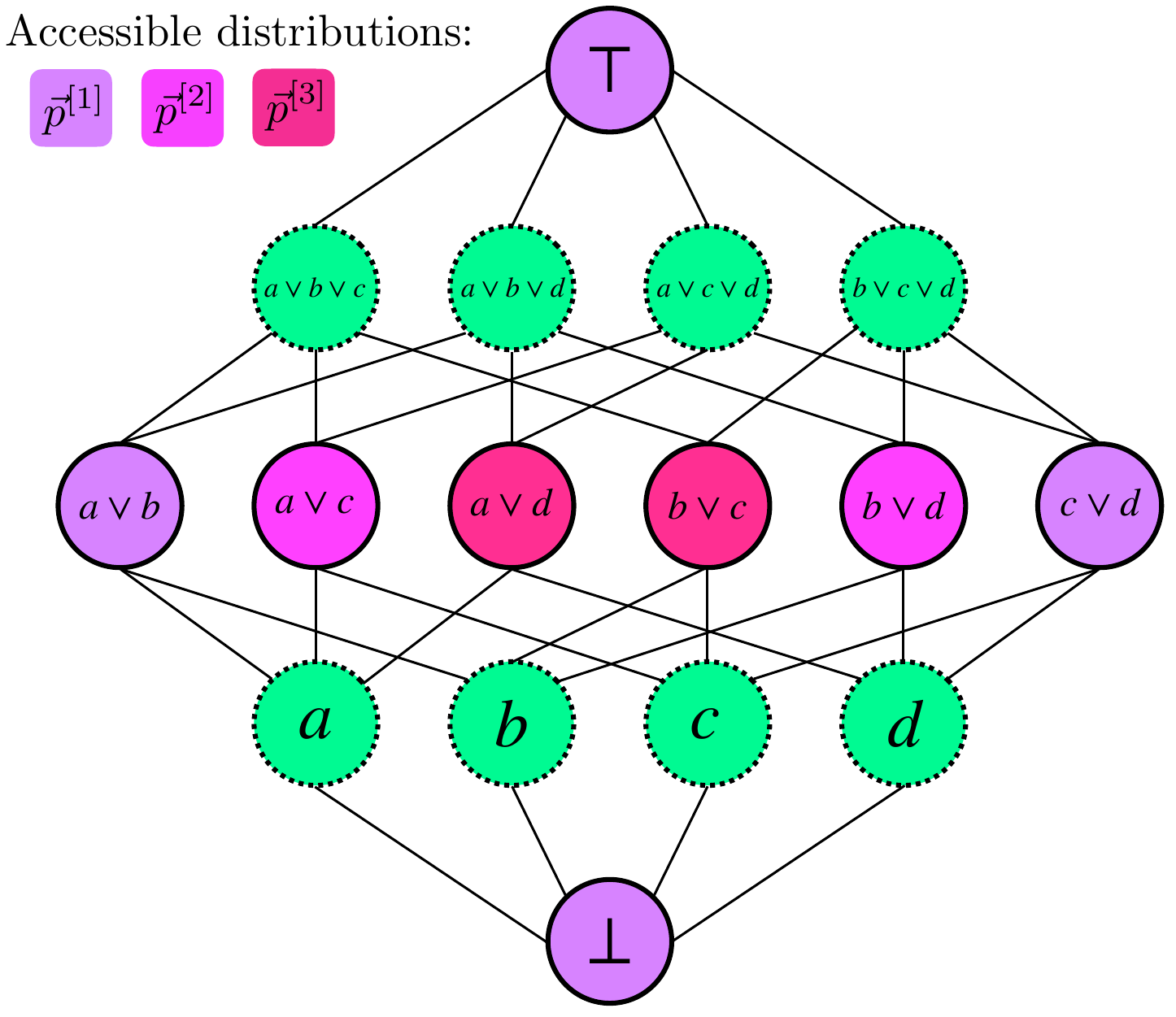}
	\caption{Lattice representation of the MES model with accessibility-depth $d=2$. The $d+1=3$ accessible distribution at level $d$ contains the whole set of statements at level $d$.}
	\label{fig:MESd2}
\end{figure}

We want now to characterise the set $\mathcal{Q}_2$ of admissible states $\vec{q}$.
The constraints on $\vec{q}$ are:
\begin{itemize}
	\item Quasi-probability constraint: $\vec{q}\in\mathbb{R}^4$ with $\sum_{i=1}^4q_i=1$. See section \ref{sec:relaxing}.
	\item Accessibility constraint: $q\in MES_2$ implies that all $\{\vec{p}^{[i]}\}_{i=1}^{d+1}$ must be probability vectors. For the case $d=2$ this implies $(q_i+q_j)\in[0,1]$ for all $i\neq j$. See definition \ref{def:MESmodel} and \ref{def:MESd}.
	\item inaccessibility constraint: $\vec{q}$ has to obey the assumption 1, that is $\indisti_2(\vec{q})\geq 2$. See definition \ref{ass:assumption1-MES}.
\end{itemize}
So $\mathcal{Q}_2$ is the set of all the $\vec{q} \in \mathbb{R}^4$ such that
\begin{numcases}{ } 
	\sum_{i=1}^4q_i=1, \label{eq:condition1}\\ 
	(q_i+q_j)\in[0,1]\;,\; \forall i\neq j, \label{eq:condition2}\\
	\indisc_2(\vec{q})\geq 2.\label{eq:condition3}
\end{numcases}
Equation \eqref{eq:condition1} combined with equation \eqref{eq:condition3} implies equation \eqref{eq:condition2}, thus $\mathcal{Q}_2$ is the set of all the $\vec{q} \in \mathbb{R}^4$ such that
\begin{align}
	\begin{cases} 
	\label{eq:constraints-MES2-X}
		\sum_{i=1}^4q_i=1\\ 
		\indisc_2(\vec{q})\geq 2,
	\end{cases}
\end{align}
or
\begin{align}
	\begin{cases} 
	\label{eq:constraints-MES2}
		\sum_{i=1}^4q_i=1\\ 
		\sum_{i=1}^4 q_i^2\leq \frac{1}{2},
	\end{cases}
\end{align}
where we rewrote the constraint $\indisc_2(\vec{q})\geq 2$ as $\sum_{i=1}^4 q_i^2\leq \frac{1}{2}$.
This is the intersection of an hyperplane and a solid hypersphere in $4$ dimensions. The intersection is a solid sphere in $3$ dimensions, thus the set $\mathcal{Q}_2$ is a convex set with extreme points on the surface of the sphere, that is all the vectors $\vec{q}$ such that conditions \eqref{eq:constraints-MES2} are satisfied and the vectors on the surface of the sphere  are the pure states of the theory as they all have inaccessibility equal to $2$.

\subsection{The qubit}
It should not come as a surprise if in this section we depart for a moment from our theory of inaccessible information for a digression to quantum systems. The analogies were building up and we need to address them explicitly.\\
We consider a qubit system. A state is described by its density matrix $\rho \in \mathcal{H}_2=\mathbb{C}^2$. Within the Bloch sphere formalism one can consider the observables associated to the measure of the spin along one of the three orthogonal axis $x,y,z$.
For each axis $x,y,z$ one defines two projectors $\{\Pi_{\uparrow}^{[\alpha]},\Pi_{\downarrow}^{[\alpha]}\}_{\alpha=x,y,z}$ such that $P^{[\alpha]}_{\uparrow}=\tr[\Pi_{\uparrow}^{[\alpha]}\rho]$ and $P^{[\alpha]}_{\downarrow}=\tr[\Pi_{\downarrow}^{[\alpha]}\rho]$  are the probabilities that the state $\rho$ is measured up or down respectively along the axis $\alpha$.\\
Let us define the four matrices
\begin{align}
\label{eq:Wootters-frame}
\hat{a} &= \frac{1}{4}\left(\id +\sigma_x+ \sigma_y + \sigma_z \right),\nonumber\\
\hat{b} &= \frac{1}{4}\left(\id +\sigma_x- \sigma_y - \sigma_z \right),\nonumber\\
\hat{c} &= \frac{1}{4}\left(\id -\sigma_x+ \sigma_y - \sigma_z \right),\nonumber\\
\hat{d} &= \frac{1}{4}\left(\id -\sigma_x- \sigma_y + \sigma_z \right),
\end{align}
where $\sigma_x$, $\sigma_y$, and $\sigma_z$ are the Pauli matrices.
 \begin{figure}[t]
	\centering
	\includegraphics[width=1\linewidth]{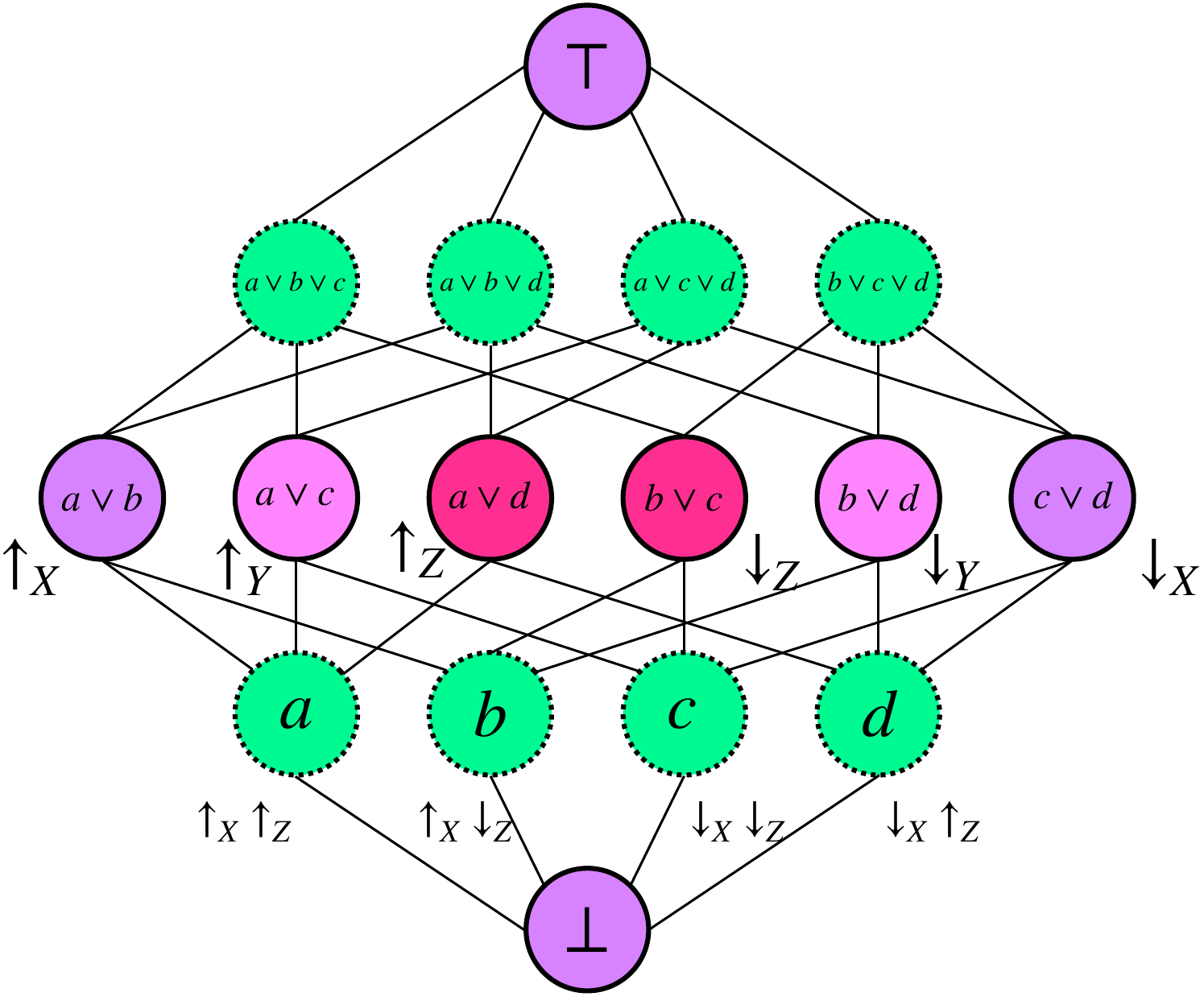}
	\caption{Analogy qubit and MES with accessibility-depth 2.}
	\label{fig:analogy-qubits}
\end{figure}
One can see that
\begin{align}
\Pi^{[x]}_{\uparrow}& =\hat{a}+\hat{b}, &\Pi^{[y]}_{\uparrow}& =\hat{a}+\hat{c}, &\Pi^{[z]}_{\uparrow}& =\hat{a}+\hat{d}, \\ 
\Pi^{[x]}_{\downarrow}& =\hat{c}+\hat{d}, &\Pi^{[y]}_{\downarrow}& =\hat{b}+\hat{d}, & \Pi^{[z]}_{\downarrow}& =\hat{b}+\hat{c}.
\end{align}
If we call $P^{[\alpha \wedge \beta]}_{i,j}$ with $\alpha,\beta \in \{x,y,z\}$ and $i,j\in\{\uparrow,\downarrow\}$ the "probabilities" of measuring the qubit simultaneously along the two axis $\alpha,\beta$ in the directions $i$ and $j$ respectively, we have the following interpretation
\begin{align} 
   P^{[x\wedge z]}_{\uparrow\uparrow}& = \tr[\rho \hat{a}], & P^{[x\wedge z]}_{\uparrow,\downarrow}& =\tr[\rho \hat{b}],\\
   P^{[x\wedge z]}_{\downarrow,\downarrow}& = \tr[\rho \hat{c}], & P^{[x\wedge z]}_{\downarrow,\uparrow}& =\tr[\rho \hat{d}].\\
  \end{align}
We notice that $\hat{a},\hat{b},\hat{c},\hat{d}$ is an orthogonal basis of $\mathcal{H}_2$ with respect to the Hillbert-Schmidt product.
Thus a state $\rho$ is completely characterised by the vector
\begin{equation}
\vec{q}(\rho)=( \tr[\rho \hat{a}], \tr[\rho \hat{b}], \tr[\rho \hat{c}], \tr[\rho \hat{d}]),
\end{equation}
and
\begin{align}
\sum_{i=1}^{4}q(\rho)_i=\tr[\rho]=1.
\end{align}
We called $P^{[\alpha \wedge \beta]}_{i,j}$ "probabilities" with quotation marks, because they are in fact quasi-probabilities. The simultaneous measurements of two different spins is not an observable and $P^{[\alpha \wedge \beta]}_{i,j}$ can be negative.
  
 Furthermore we notice that for the qubit example
\begin{align}
\uncert_2(\vec{q}(\rho))=\frac{2}{Tr[\rho^2]}.
\end{align}

Thus, for a general pure qubit state $\ket{\psi}$ the general quasi-probability vector associated
\begin{equation}
	\vec{q}_{\ket{\psi}}= \left( \bra{\psi}\hat{a}\ket{\psi},\bra{\psi}\hat{b}\ket{\psi},\bra{\psi}\hat{c}\ket{\psi},\bra{\psi}\hat{d}\ket{\psi} \right)
\end{equation}
is characterised by a fixed inaccessibility $\uncert_2(\vec{q}_{\ket{\psi}})=2$ since $\tr{(\ket{\psi}\bra{\psi})^2}=1$.
For a general qubit state $\rho$ the quasi-probability vector associated $\vec{q}(\rho)$ is characterised by a constrained inaccessibility $\uncert_2(\vec{q}(\rho))\geq2$ since $\tr{\rho^2}\leq1$.

We see that the constraints characterising the hermitian matrices describing the qubits, namely
\begin{align}
	\begin{cases} 
		\tr[\rho]= 1\\ 
		\tr[\rho^2]\leq 1.
	\end{cases}
\end{align}
are directly mapped to the constraints \eqref{eq:constraints-MES2-X} characterising the theory of inaccessible information with accessibility-depth $d=2$.

We find that qubits states can be formulated as the states of a MES model of accessibility-depth $d=2$ in a theory of inaccessible information. In figure \ref{fig:analogy-qubits} we see an analog of panel of figure \ref{fig:MESd2} written in terms of quantum states.
In particular, there is a one-to-one correspondence between $\vec{q}$ allowed by a theory of inaccessible information with accessibility-depth $d=2$ and quantum states describing a qubit.

This result can be interpreted within the framework of quantum frames \cite{ferrie2009}. Starting from the description of the qubit, one can choose any specific frame of $4$ operators and describe the qbit state as a vector of \mbox{(quasi-)probabilities $\vec{\mu}$}.  We can see that, choosing the specific frame of equations \eqref{eq:Wootters-frame}\footnote{We can call this the Feynman-Wootters frame as it appears in Feynman \cite{hiley1987a} in relation to a discussion on quasi-probabilities, and it is generalised to any prime dimension by Wootters in \cite{wootters1987}.}, the constraints on the frame representation vector $\vec{\mu}=\vec{q}$ have a natural interpretation as constraints on the inaccessibility of the vector, we have that $\vec{\mu}$ has to satisfy $\indisti(\mu)\geq2$.
We can compare these constraints with the ones obtained with choice of a different frame. For example, choosing as frame a SIC-POVM, as done in the QBist reconstruction of quantum mechanics, the situation is different. The vector $\vec{\mu}$ would be a probability vector $\vec{\mu}=\vec{p}$, and the bound on the inaccessibility would consequently change to a bound of another function of $\vec{\mu}$ (equation $(7)$ of \cite{appleby2017}). The form of the bound for the set of $\vec{\mu}$ describing the qubit states within the SIC-POVM frame representation seems to not have a natural interpretation. Quoting \cite{appleby2017} (second paragraph after equation $(9)$) "[...] \textit{what a strange-looking set the quantum states become when written in these terms! What could account for it except already knowing the full-blown quantum theory as usually formulated?}"\footnote{This statement refers to the set of constraints characterising quantum-states in all dimensions, not just for the case of the qubit. Specifically, when dealing with systems with dimensions greater than $2$, an additional constraint must be considered (eq. (8) in \cite{appleby2017}). In our work, we are only dealing with the qubit, and thus only considering half (the simpler half) of the constraints referred to in this statement.}.
Within the formalism of the theory of inaccessible information we cannot derive equation $(8)$ of \cite{appleby2017}, but equation $(7)$ of \cite{appleby2017} becomes the lower bound on the inaccessibility. Part of the constraints can be derived logically and look less weird when considering theories of inaccessible information\footnote{We remind once again that the theory of inaccessible information is not a tentative reconstruction of quantum mechanics.}.

We conclude this section with a last analogy, noticing that the set of accessible probability distributions at level $2$ corresponds to a set of mutually unbiased operators for the qubit.

\section{The theory of inaccessible information with accessibility-depth d=3}

The model is characterised by $D=d^2=9$ atomic statements $(a,b,c,d,e,f,g,h,i)$ to which we associate a state in the form of the quasi-probability vector $\vec{q}=(Q(a),Q(b),Q(c),Q(d),Q(e),Q(f),Q(g),Q(h),Q(i))$.
Since we are considering a MES model, at level $d=3$ we have $d+1=4$ accessible probability vectors. Within the freedom of relabelling, these are
	\begin{align}
		\vec{p}^{[1]}&= \left(P(a\vee b\vee c),P(d\vee e\vee f),P(g\vee h\vee i)\right),\\
		\vec{p}^{[2]}&= \left(P(a\vee d\vee g),P(b\vee e\vee h),P(c\vee f\vee i)\right),\\
		\vec{p}^{[3]}&= \left(P(a\vee e\vee i),P(b\vee f\vee g),P(c\vee d\vee h)\right),\\
		\vec{p}^{[4]}&= \left(P(a\vee f\vee h),P(b\vee d\vee i),P(c\vee e\vee g)\right).
	\end{align}

In the case of accessibility-depth $d=3$ not all statements at level $d$ belong to an accessible probability distribution. In fact at level $d=3$ one finds $\frac{d^2!}{d!(d^2-d)!}=84$ statements, while only $d(d+1)=12$ enter into the accessible probability distributions.
Analogously to the case $d=2$, here the set $\mathcal{Q}_3$ of admissible states $\vec{q}\in \mathbb{R}^9$ is characterised as
\begin{numcases}{ } 
	\sum_{i=1}^9 q_i=1, \label{eq:condition1-d3}\\ 
	(q_i+q_j+q_k)\in[0,1]\;,\; \forall i\neq j, i\neq k, j\neq k, \label{eq:condition2-d3}\\
	\sum_{i=1}^9 q_i^2\leq \frac{1}{3}.\label{eq:condition3-d3}
\end{numcases}

Differently from the case $d=2$, here all three conditions are necessary. In fact, the quasi-probability vector
\begin{equation}
\tilde{\vec{q}}=\left(\frac{1}{4}-\frac{1}{24} \left(3+\sqrt{33}\right),0,0,\frac{1}{4},\frac{1}{4},\frac{1}{4},0,0,\frac{1}{24} \left(3+\sqrt{33}\right)\right),
\end{equation}
satisfies condition \eqref{eq:condition1-d3} and condition \eqref{eq:condition3-d3}, but not condition \eqref{eq:condition2-d3}.

\section{Outlook and Conclusions}

Long time ago, at the time when quantum mechanics was being developed, in the field of logic there was a discussion regarding the distinction between true and provable statements \cite{nagel2005}. Can we find a proof for all true statements? Are all true statements provable? Among many others, K. Gödel helped giving an answer to this questions introducing the famous incompleteness theorems confirming once and forever that there are some statements that are true, but that cannot be proven\footnote{Given a fixed formal system.}.
In our work we have been inspired by these results, drawing a suggestive analogy between logic and physics. We imagined the process of an experiment as a tool for proving statements about reality. If a logical proof is an algorithm involving the tools of logic, an experiment is an algorithm involving the tools of the physical world. As there are statements that cannot be proved with a logical
proof, we imagined there can be statements that cannot be proved with an experimental inspection. This suggestions was clearly inspired by the lessons of quantum mechanics. In fact, since its inception, quantum mechanics has been suggesting that there must be some limits to our knowledge.
A remark here has to be done, the wording "proved with an experimental inspection" may be hard to read, as it seems to collide with the wording used in the discourses over falsifiability. What we are talking about are statements as "doing \textit{this} and \textit{that} I will obtain this numbers on this screen". By proving this statement with an experimental inspection we mean performing the aforementioned \textit{this} and \textit{that} and verify if we obtained the predicted numbers on the screen or not.\\
Even though this is the rationale behind this work, we decided to focus on the less philosophical part, and simply study the consequences of the $\mathcal{I}\mathcal{H}$, formalising the distinction between accessible and inaccessible statements. Accessible statements being the one for which the value of truth can be assessed experimentally. We decided to formalise the $\mathcal{I}\mathcal{H}$ in the language of classical propositional logic for two main reasons. The first reason is simplicity; we consider classical propositional logic to be one of the simplest set of rules for the manipulation of statements. The second reason is our familiarity with the extensive literature on the derivation of probability theory from logic, which presented an opportunity to derive probability theory within our framework rather than attaching it externally (for example, by allowing convex combinations of states justified by operational reasons).
We further assume an \textit{accessiblity-depth}  $d$ on the accessible knowledge. Each accessible statement is composed by $d$ inaccessible statements. This means that, even if we are able to assess the truth value of a statement we will nonetheless be uncertain over $d$ inaccessible statements. The accessibility-depth $d$ individuates a fundamental degree of "unknowledge". 
From these basic assumptions a rich structure emerges. Surprisingly, a relation between the accessibility-depth and the number of total statements to be considered in the model is found. For an accessibility-depth $d$ we need to consider at least $d^2$ statements.

Building up on results on the foundations of probability theory, we have been able to relax and extend the logical model of the theory of inaccessible information to a quasi-probabilistic model. We derive the quasi-probabilistic structure from first principle based solely on consistency criterion, showing how a quasi-probabilistic description is the natural language of choice when allowing for the existence of inaccessible statements. \footnote{At the end of section \ref{sec:composition-of-models} we observed that adding the accessibility label appeared to be merely an overparameterisation, simply involving the addition of a new label with an attached interpretation to each statement. However this very interpretation of the distinction between accessible and inaccessible statements within the logical lattice proved crucial for determining how to relax the logical binary assignments, as noted in Appendix \ref{appendix:additional-assumptions}.} Our construction reminds, in part, about hidden variable models, and suggests that, for modelling the knowledge associated with hidden variables, quasi-probabilities are a great tool. In fact our approach rises the following questions, why  should we associate probabilities to hidden variables? How could we justify this specific restriction? We cannot build frequencies or beliefs on something we cannot access, something that is hidden.
Within the quasi-probabilistic extension of the theory of inaccessible information, we found ourselves in the need of defining a measure of uncertainty. The necessity of quantifying the degree of inaccessibility to preserve the concept of knowledge's  \textit{accessibility-depth} leads to the definition of inaccessibility measure. Given a quasi-probability vector, the inaccessibility measure quantifies its knowledge's depth.

The theory of inaccessible information becomes simply defined as  a theory of quasi-probabilities with two constraints: quasi-probabilities are expected to behave as proper probabilities when restricted to accessible statements, and the admissible quasi-probability vectors have a bounded value of inaccessibility.
Nothing more than this (almost).

Given our premises it did not come as a surprise when we uncovered the connection between the theory of inaccessible information and the formalism of quantum mechanics. In fact, what we found, is that the structure of the qubit can be derived as the structure of the states of the simplest model of inaccessible information\footnote{We note here that the structure of the qubit can be (and has been historically) derived in many different ways, starting from many different sets of principles. For this reason the qubit can be considered a fundamental structure. For the same reason, at this point, being able to construct the qubit from first principles can be deemed as an exercise of style.}. Albeit we deem this result as secondary (the aim of this work is the study of the theory of inaccessible information on its own, not to derive the structure of the qubit), the same result can be recasted as an additional incentive for further development and study of the theory of inaccessible information.

All in all, the theory of inaccessible information, even though based on simple assumptions, revealed itself for a theory rich in structure. On the journey to the development of our theory we have come across interesting results. We derived quasi-probabilities from logical principles, establishing a basic rationale for their use as a tool. We derived from few principles the family of the inaccessibility measures,  a set of multiplicative information measures that has appeared scattered around in the literature with many different names and forms (see \cite{hill1973,pellens2016,popescu2006,dunlop2021,brukner1999}). We recovered the structure of the qubit as the simplest model of inaccessible information, that is as a model where each accessible statement is always composed by two inaccessible statements.
It has been a nice journey, but not a finished journey. Many are the possible directions that can be undertaken. From an abstract point of view, we notice that the theory of inaccessible information seems to not incorporate the concept of "object". By "object" we mean something to which the statements are referring to. The theory as it is just talks about statements related to each others by the algebra of statements and by the constraints of the theory. Ideally a definition of "object" would allow to define "transformations" of statements via compositions. It is easy to think about this concept by restricting to the case of quantum mechanics where, for the case of the qubit, these transformations would be the unitary transformations allowing to transform statements about the spin along one direction to statements about the spin along any direction.
Another direction is formulating properly the process of state update starting when new information is gathered.
Another further direction is studying carefully models with accessibility-depth higher then $d=2$.

\emph{Acknowledgements.} 
First of all I want to thanks Matteo Scandi, best "compagno di merende", with whom I had many lovely discussions, he helped me in staying excited about this work. Then I would like to thank Francesco Buscemi who kindly hosted me in his group while I was developing part of this work, inspired different part of this work and shared his passion on the foundational aspects of probability theory. Finally I would like to thank Markus P. Müller, who, after a presentation of an embryonic version of this work asked me the right questions that lead to the definition of the inaccessibility measures. We thank the referee for the many suggestions and corrections that greatly improved this work. The name "Inaccessibility Hypothesis" was suggested by the referee.
This project has received funding from the Government of Spain (FIS2020-TRANQI and Severo Ochoa CEX2019-000910-S), Fundacio Cellex, Fundació Mir-Puig, Generalitat de Catalunya (SGR 1381 and CERCA Programme).Research at the Perimeter Institute for Theoretical Physics is supported by the Government of Canada through the Department of Innovation, Science and Economic Development Canada and by the Province of Ontario through the Ministry of Research, Innovation and Science.

%\emph{Acknowledgements.} 
%First of all I want to thanks Matteo Scandi, best "compagnio di merende", with whom I had many nice discussions and helped me in staying excited about this work. Then I would like to thank Francesco Buscemi that kindly hosted me in his group while I was developing this work, inspired different part of this work and shared his passion on the foundational aspects of probability theory. Finally I would like to thanks Markus P. Müller, that, after a presentation of an embryonic version of this work asked me the right questions that lead to the definition of the inaccessibility measures.
%This project has received funding from the Government of Spain (FIS2020-TRANQI and Severo Ochoa CEX2019-000910-S), Fundacio Cellex, Fundació Mir-Puig, Generalitat de Catalunya (SGR 1381 and CERCA Programme)

%\bibliographystyle{unsrtnat}
%\bibliographystyle{quantum}
%\bibliography{Inaccess.bib}

\bibliographystyle{quantum}
\bibliography{Inaccess.bib}
 
\onecolumn

 \begin{appendix}
 
 \section{Proofs}
 
   	In order to simplify the language, we identify a collection of statements $s_i \vee s_j \vee \dots \vee s_k$ connected by the binary operator $\vee$ with the set of their indices $(i,j,\dots,k)$ (see section \ref{appendix:algebra-of-statements}). Thus the rule of composition of statements maps to the rules of union and intersection of sets as in these examples
 	\begin{align}
 		s_i &\leftrightarrow (i),\\
 		s_i \vee s_j &\leftrightarrow (i,j)\\
 		(s_i\vee s_j) \wedge (s_j\vee s_k)=s_j &\leftrightarrow  (i,j)\cap(j,k)=(j).
 	\end{align}
 	
 \label{appendix:proofs}
 
 \topisaccessible*
 \begin{proof}
 	Let us set the accessibility of  $\top$ to $\NNN$. Choose a statement $s_i$ in the lattice. If $s_i$ is $\AAA$ than $\neg s_i$ is $\AAA$. This implies that $\top=s_i \vee \neg s_i$  is $\AAA$. This is a contradiction. 
 \end{proof}

% \begin{lemma}
%\label{lem:integer-multiples-old}
%	For each resolution $d$, for $D\neq m\cdot d$, with $m\in \mathbb{N}$, the only models obtainable are classical models, truism models or model with just $div(D,d)$ $\AAA$ statement at level $d$, .
%\end{lemma}
%\begin{proof}
%	Suppose $D$ is not an integer multiple of $d$. Without loss of generality consider $p_1\vee p_2\vee \dots \vee p_d$, at level $d$ to be $\AAA$. This implies that the complement $\neg(p_1\vee p_2\vee \dots \vee p_d)$ at level $D-d$ is $\AAA$. If $D>2d$ this is allowed. Selecting another element from level $d$, $d-1$ scenarios are possible. We select a statement $\tilde{p}_0$ such that $(p_1\vee p_2\vee \dots \vee p_d) \wedge \tilde{p}_0=\perp$, the two composite statements do not have atomic statements in common. We can select a statement $\tilde{p}_1$ such that $(p_1\vee p_2\vee \dots \vee p_d) \wedge \tilde{p}_1=p_i$, the two composite statements  have $1$ atomic statements in common. Up to the scenario in which we select a statement $\tilde{p}_{d-1}$ such that $(p_1\vee p_2\vee \dots \vee p_d) \wedge \tilde{p}_{d-1}=p_{i_1}\vee p_{i_2}\vee \dots \vee p_{i_{d-1}}$, the two composite statements  have $d-1$ atomic statements in common. Ecc ecc.
%\end{proof}

%
%\begin{lemma}
%	To every classical model with $m$ atomic elements inflated to an $(md,d)$ inaccessible information model with $d\in\mathbb{N}$ adding $i\in\{1,\dots,d-1\}$  atomic elements make classical sublattice of dimension $m-1$. \js{qualcosa del genere per ultimo step della prova del prossimo lemma}
%\end{lemma}
%\begin{proof}
%	
%\end{proof}

 \notintersect*
 \begin{proof}
 Consider the model $(D,d)$.
Suppose $(i_1,i_2,\dots,i_d)$ and $(j_i,j_2,\dots,j_d)$ are two accessible statements at level $d$, and $(i_1,i_2,\dots,i_d)\cap(j_i,j_2,\dots,j_d)=(k_1,\dots,k_{r\lneqq d})$ (we impose ${r\neq d}$ because otherwise the two statements would be identical).  Because of the rule of accessibility propagation from Table \ref{table:accessibility-table}, the statement $(k_1,\dots,k_{r\lneqq d})$ at level $r$ must be accessible, but the accessibility-depth $d$ of the model forbids every statement below level $d$ to be accessible. Thus at level $d$ all accessible statements must not intersect.
 \end{proof}

\integermultiples*
 \begin{proof}
	Consider a generic $(D,d)$ model. The set of atomic elements is $\{s_1,\dots, s_D \}$.\\
 	Because of lemma \ref{lem:not-intersect}, at level $d$ we can consider accessible at most $\lerifloor{\frac{D}{d}}$ statements. 
 	In the case $D=md$, with $m\in\mathbb{N}$, the number of accessible statements can be maximised to be $m$. A model of this kind is an inflation of a classical model with $m=\frac{D}{d}$ atomic elements. \\ %, as it is clear from lemma \ref{lem:classical-sub-lattices}.\\
 	We have to prove the lemma for $\text{mod}(D,d)\neq 0$.
 	First we show that it is impossible to choose more than $\lerifloor{\frac{D}{d}}$ statements to be accessible at level $d$.
 	We prove that we cannot set $\lerifloor{\frac{D}{d}}$ statements at level $d$ to be accessible by showing that doing so would imply that some statement at level lower then $d$ must be accessible.
	We set $\lerifloor{\frac{D}{d}}$ not-intersecting statements at level $d$ to be accessible. The union $\tilde{s}$ of all the accessible statements at level $d$ is a statement composed by $d\lerifloor{\frac{D}{d}}=D-\text{mod}(D,d)$ atomic statements. Because of the rules of accessibility propagation we have that $\tilde{s}$  must be an accessible statement.  Applying the rules of accessibility propagation another time we have that $\neg \tilde{s}$, must be accessible too, but $\neg \tilde{s}$ is composed by $D-(D-\text{mod}(D,d))=mod(D,d)\lneqq d$ atomic statements thus cannot be accessible by definition. Thus we cannot take more then $\lerifloor{\frac{D}{d}}-1$ accessible statements at level $d$.

 	Now we set $\lerifloor{\frac{D}{d}}-1$ not-intersecting statements at level $d$ to be accessible.
 	The union $\tilde{s}$ of all the accessible statements at level $d$ is a statement composed by $d(\lerifloor{\frac{D}{d}}-1)=D-\text{mod}(D,d)-d$ atomic statements. In this case $\neg \tilde{s}$ is composed by $D-(D-\text{mod}(D,d)-1)=mod(D,d)+d\gneqq d$ atomic statements, so the argument above does not apply.
 \end{proof}

  \uniquenessidealconfigurations*
  \begin{proof}
  %The configuration of a model $(D,d)$ is characterised by the collection of accessible statements at level $d$.
  For a model $(D,d)$ with $\bmod(\frac{D}{d})=0$, from lemma \ref{lem:not-intersect} and lemma \ref{lem:integer-multiples} it follows that the choice of the set of accessible statements at level $d$ corresponds to a \textit{partitioning} of the complete set of atomic statements in $\frac{D}{d}$ non-overlapping sets each of cardinality $d$. Given a partitioning of this type, all the other possible partitioning are obtained by relabelling.  For a model $(D,d)$ with $\bmod(\frac{D}{d})\neq0$, from lemma \ref{lem:not-intersect} and lemma \ref{lem:integer-multiples}, it follows that the choice of the set of accessible statements at level $d$ corresponds to a \textit{partitioning} of the set of atomic statements in $\lerifloor{\frac{D}{d}}-1$ non-overlapping sets each of cardinality $d$ plus the set of remaining $d+\bmod(D,d)$ atomic statements. This is a complete partitioning of the set of atomic statements, thus all the other possible partitioning of this kind are obtained by relabelling. 
  \end{proof}

 \classicalsublattices*
 \begin{proof}
 	It is a direct application of lemma \ref{lem:integer-multiples}.
 \end{proof}

\scalingdimension*
\begin{proof}
We want to characterise the set of general procedures of inflation that are compatible with the composition of models.
That is, we want to find a function that given the dimension $m$ of a classical model it returns the number of atomic statements $D$ and the accessibility-depth $d$ of an ideal configuration of an inaccessible information model $(D,d)$. Furthermore we want for this function to be compatible with the composition of models: the inflation of the composition of models must correspond to the model obtained by the composition of inflations.

We know that for any classical model of dimension $m$ any couple $(D,d)$  that satisfies the conditions of  lemma \ref{lem:classical-sub-lattices} determines an inaccessible information model that, in its ideal configuration, is an inflation of the classical model $m$. 
%The number of atomic statements $D$ depends on the classical dimension $m$ and on the chosen accessibility-depth $d$. 
In order to explicitly write the general procedure of inflation, we consider a function $F$ that takes as argument the classical dimension $m$ and returns the accessibility-depth $d=F(m)$ of the inflated model. From lemma \ref{lem:classical-sub-lattices} we have than the the allowed number of atomic statements in the inflated moments are the members of the set $\mathcal{D}_{m,F(m)}$.
It follows that the general procedure of inflation associates to each classical model an inaccessible information model 
\begin{equation}
m\to (\tilde{D},F(m))
\end{equation}
where $\tilde{D}\in\mathcal{D}_{m,F(m)}$.

Now let us consider two generic classical models with dimensions $m_1$ and $m_2$. Their inflations are 
$(\tilde{D}_1,F(m_1))$ and $(\tilde{D}_2,F(m_2))$, where $\tilde{D}_1\in\mathcal{D}_{m_1,F(m_1)}$ and $\tilde{D}_2\in\mathcal{D}_{m_2,F(m_2)}$.
The composition of the classical models is a classical model with dimension $m_1\cdot m_2$ with inflation 
\begin{equation}
\label{eq:inflation-of-composition}
(\tilde{D}_{12},F(m_1\cdot m_2))\qquad \text{ (inflation of composition)}
\end{equation}
 where $\tilde{D}_{12}\in\mathcal{D}_{m_1\cdot m_2,F(m_1\cdot m_2)}$. 
 
The composition of the inflations of each classical model is the model 
\begin{equation}
\label{eq:composition-of-inflations}
(\tilde{D}_1\cdot \tilde{D}_2,F(m_1)\cdot F(m_2)) \qquad \text{ (composition of inflations)}
\end{equation}
To ask for the the general procedure of inflation to be compatible with composition we want the model \eqref{eq:inflation-of-composition} to be the same as the model \eqref{eq:composition-of-inflations}.
We start by equating the accessibility-depths of the two composed models, obtaining that for the function $F$ it must hold that
\begin{equation}
F(m_1\cdot m_2)=F(m_1)F(m_2).
\end{equation}
From \cite{aczel1966} we know that the solutions of this functional equation are $F(x)=x^c$ for some parameter $c$ independent from $m$.
Since the accessibility-depth of a model is an integer number greater then $1$ we want $F$ to return an integer value greater than $1$, thus we have that $c\in\mathbb{N}_{>0}$.
We obtain then that the models  \eqref{eq:inflation-of-composition} and \eqref{eq:composition-of-inflations} can be written as
\begin{align}
(\tilde{D}_{12},(m_1\cdot m_2)^c),\qquad \text{ (inflation of composition),}\\
(\tilde{D}_1\cdot \tilde{D}_2,m_1^c\cdot m_2^c), \qquad \text{ (composition of inflations),}
\end{align} 
where $\tilde{D}_{12}\in\mathcal{D}_{m_1\cdot m_2,(m_1\cdot m_2)^c}$,  $\tilde{D}_1\in\mathcal{D}_{m_1,m_1^c}$ and $\tilde{D}_2\in\mathcal{D}_{m_2,m_2^c}$.
Now we have to find elements $\tilde{D}_{12}\in\mathcal{D}_{m_1\cdot m_2,(m_1\cdot m_2)^c}$,  $\tilde{D}_1\in\mathcal{D}_{m_1,m_1^c}$ and $\tilde{D}_2\in\mathcal{D}_{m_2,m_2^c}$ such that
\begin{equation}
\tilde{D}_{12}=\tilde{D}_1\cdot \tilde{D}_2,
\end{equation}
independently on $m$.
Since $c$ must be independent from $m$, otherwise the inflation procedure would not be general, we have that the only solution is for the general procedure of inflation to select $D=dm$ out of $\mathcal{D}_{m,d}$. Thus we have $\tilde{D}_{12}=m_1\cdot m_2 \cdot (m_1\cdot m_2)^c$, $\tilde{D}_{1}=m_1\cdot m_1^c$, and $\tilde{D}_{2}=m_2\cdot m_2^c$.
We have then that the models  \eqref{eq:inflation-of-composition} and \eqref{eq:composition-of-inflations} can be written as
\begin{align}
((m_1\cdot m_2)\cdot (m_1\cdot m_2)^c,(m_1\cdot m_2)^c),\qquad \text{ (inflation of composition),}\\
((m_1\cdot m_1^c)\cdot(m_2\cdot m_2^c),m_1^c\cdot m_2^c), \qquad \text{ (composition of inflations).}
\end{align} 
We obtained a family of general procedures of inflations compatible with the composition and parametrised by a positive integer $c$ such that for any fixed $c$, the inflation procedure associates to each classical model  of dimension $m$ an inaccessible information model as follows
\begin{equation}
m\to (m^{c+1},m^c).
\end{equation}
\end{proof}

\MESdplusone*
\begin{proof}
% TODO Vettori partono da 0 dappertutto?
	The proof consists in showing that for a set $\mathcal{L}=\{s_0,\dots,s_{d^2-1}\}$ of $d^2$ elements with $d$ prime, it is possible to choose at most $d(d+1)$ sets $\{\ell_i\}_{i=0}^{d(d+1)-1}$ each of $d$ elements such that each set has an overlap of at most $1$ element with any other set, $|\ell_i \cap \ell_j|\leq 1$, $\forall i\neq j$.\\
	We prove the statement by construction.\\
	We construct the first $d$ sets $\{\ell_i\}_{i=0}^{d-1}$ by choosing a partition of $\mathcal{L}$ in $d$ non-overlapping sets of $d$ elements. We choose the sets $\ell_0=\{s_0,\dots,s_{d-1}\}$, $\ell_1=\{s_{d},\dots,s_{2d-1}\}$, $\dots$,$\ell_{d-1}=\{s_{d(d-1)},\dots,s_{d^2-1}\}$.
	The remaining $d^2$ sets $\{\ell_i\}_{i=d}^{d(d+1)-1}$ constructed with the following rule dictating to which set $\ell$ an element $s_i$ belong.
	Consider the following equation
	\begin{equation}
	\label{eq:modular-equation}
		x=\bmod\left(\lerifloor{\frac{j}{d}}y+\bmod(j,d),d\right).
	\end{equation}
	For each $j\in\{ 0,1,\dots,d^2-1\}$ and for each $y\in\{0,1,\dots,d-1\}$ there is one single $x\in\{0,1,\dots,d-1\}$ that satisfies the equation.  In other words, for each value of $j$ there exist $d$ pairs $(x,y)\in\{0,1,\dots,d-1\}\times\{0,1,\dots,d-1\}$ that satisfy equation \eqref{eq:modular-equation}.
	For each $j$ the statement $s_j$ belong to the $d$ sets $\ell_{d(x+1)+y}$ where $x$ and $y$ are the solutions of equation \eqref{eq:modular-equation}.\\
	We note that for each $x\in\{0,1,\dots,d-1\}$ the sets $\{\ell_{d(x+1)+y}\}_{y=0}^{d-1}$ form a non-overlapping partitioning of $\mathcal{L}$, thus are associated to an accessible probability distribution.
	
	To show that this construction is valid it is sufficient to notice that  any two elements set $\tilde{\ell}=\{s_m,s_n\}$ is a subset of the unique $\ell_g$ constructed above, where $g=\lerifloor{\frac{m}{d}}$ if $\lerifloor{\frac{m}{d}} = \lerifloor{\frac{n}{d}}$ and $g=d(x+1)+y$ where $(x,y)$ is the unique solution of the system of equations
	\begin{equation}
    	\begin{cases}
			x=\bmod\Big(\left(\lerifloor{\frac{m}{d}}y+\bmod(m,d)\right),d\Big)\\
			x=\bmod\Big(\left(\lerifloor{\frac{n}{d}}y+\bmod(n,d)\right),d\Big)
		\end{cases},
	\end{equation}
	if $\lerifloor{\frac{m}{d}} \neq \lerifloor{\frac{n}{d}}$.
	This proves that the set is maximal as the addition of any new set would have an overlap of at least $2$ elements, this also proves that the overlap of any two sets chosen from $\{\ell_i\}_{i=0}^{d(d+1)-1}$ is at most of $1$ element.
\end{proof}

 \MESonetoone*
 \begin{proof}
 	Let us start with a counting argument.
 	At the accessible level there are $d+1$ accessible probability distributions, thus $d(d+1)$ variables with $d+1$ constraints (the elements of each of the $d+1$ accessible probability distributions has to sum to $1$), for a total of $(d-1)(d+1)=d^2-1$ free variables.
 	At the atomic level there is one quasi-probability vector $\vec{q}$, so there are $d^2$ variables with the single constraint constraint $\sum_{i=1}^{d^2}q_i=1$, for a total of $d^2-1$ free variables.
 	Now let's look at the one-to-one mapping between $\vec{q}$ and $\{\vec{p}^{[i]}\}_{i=1}^{d^2}$.
 		Following the construction of lemma \ref{lem:MESdplusone} we consider the set of atomic statement $\mathcal{L}=\{s_0,\dots,s_{d^2-1}\}$ with $d$ prime and the sets $\{\ell_i\}_{i=0}^{d(d+1)-1}$. 
 		For every $i\in\{ 0,1,\dots,d^2-1\}$ we consider the $d$ sets $\{\ell_{dx+y}\}_{(x,y)\in \mathcal{I}(i)}$ that contains the statement $s_i$. 
 		The set $\mathcal{I}(i)$ is explicitly defined as follow
 		\begin{equation}
 			\mathcal{I}(i)=\{(x,y)\in\{0,1,\dots,d-1\}\times\{0,1,\dots,d-1\} | x=\bmod(\lerifloor{\frac{i}{d}}y+\bmod(i,d),d)\}.
 		\end{equation}
 		To each set $\{\ell_{dx+y}\}_{(x,y)\in \mathcal{I}(i)}$ corresponds an accessible probability value $\{p_{dx+y}\}_{(x,y)\in \mathcal{I}(i)}$. Summing this accessible probability values one obtain
 		\begin{equation}
 			\sum_{(x,y)\in\mathcal{I}(i)} p_{dx+y} = dq_i+1,
 		\end{equation}
 		where $q_i=Q(s_i)$ is the inaccessible quasi-probability associated to the statement $s_i$.
 		Thus from the accessible probability distributions at level $d$ it is possible to reconstruct the vector $\vec{q}$. The opposite is clearly true, from the vector $\vec{q}$ one reconstruct the accessible probability distributions at level $d$.
% 		The fact that the number of constraint imposed is minimum is verified noticing that $d+1$ probability vectors of length $d$ can be completely specified by $d^2$ parameters, 
	The fact that the number of constraint imposed is minimum is verified noticing that removing just one accessible probability distribution (i.e. selecting an $x$ and discarding all the probability values $\{p_{dx+y}\}_{y=0}^{d-1}$) would not allow to recover $\vec{q}$ from the accessible probability distributions anymore.
	For a generic MES${}_d$ model we have then an invertible mapping from the single quasi-probability distribution $\vec{q}$ to the set of $d+1$ accessibile probability distributions $\{\vec{p}^{[i]}\}_{i=1}^{d+1}$.
 \end{proof}

 \section{Statements and quasi-probabilities}

% Most of the results in this section are heavily based on the works of K.H.Knuth \cite{knuth2004,knuth2005a}, we report these results here for convenience and for allowing us to easily introduce the method for naturally deriving quasi-probabilities from the algebra of statements.
% We invite the curious readers to enjoy the lecture of the works of Cox, Jaynes, Caticha and Knuth  \cite{cox1946b,cox1961,jaynes2003b,knuth2004,knuth2005a,knuth2009,caticha2009}.
 
 \subsection{Algebra of Statements}
 \label{appendix:algebra-of-statements}
 
The language we use in everyday life heavily relies on the use of the connectives $\vee$ (read "or") and $\wedge$ (read "and") together with the negation unary operator $\neg$. These connectives encode in the algebra they generate over the set of statements part of the structure of the mechanics of how we understand and perceive information. Studying how statements can be combined and manipulated using the operations $\vee$ or $\wedge$ or $\neg$ can feel like shedding some light on this mechanism. %This realisation is one of the creative impulses that led to the canonicalisation of logic as a fundamental field of study. 
\\
We propose an example where we need to consider just two statements $s_1$ and $s_2$. For sake of presentation we consider these two specific statements
\begin{align*}
s_1=&\text{"The apple is blue"},\\
s_2=&\text{"The dog is red"}.&
\end{align*}
An interesting question is how many new statements it is possible to obtain from $s_1$ and $s_2$ alone using $\vee$ and $\wedge$ and $\neg$ freely. 
We answer to this question quite easily. We can obtain the statements:
\begin{align*}
s_3:=\neg s_1=&\text{"The apple is not blue"},\\
s_4:=\neg s_2=&\text{"The dog is not red"},\\
s_5:=s_1 \vee s_2=&\text{"The apple is blue or the dog is red"},\\
s_6:=s_1 \vee \neg s_2=&\text{"The apple is blue or the dog is not red"},\\
s_7:=\neg s_1 \vee s_2=&\text{"The apple is not blue or the dog is  red"},\\
s_8:=\neg s_1 \vee \neg s_2=&\text{"The apple is not blue or the dog is not red"},\\
s_9:=s_1 \wedge s_2=&\text{"The apple is blue and the dog is red"},\\
s_{10}:=s_1 \wedge \neg s_2=&\text{"The apple is blue and the dog is not red"},\\
s_{11}:=\neg s_1 \wedge s_2=&\text{"The apple is not blue and the dog is  red"},\\
s_{12}:=\neg s_1 \wedge \neg s_2=&\text{"The apple is not blue and the dog is not red"},\\
s_{13}:=(s_1\wedge s_2) \vee (\neg s_1 \wedge \neg s_2)&=\text{"(The apple is blue and the dog is red) or (The apple is not blue or the dog is not red)"} ,\\
s_{14}:=(s_1\wedge \neg s_2)\vee (\neg s_1 \wedge s_2)&= \text{"(The apple is blue and the dog is not red) or (The apple is not blue or the dog is red)"},\\
s_{15}:=s_1\vee \neg s_1=&\text{"The apple is blue or the apple is not blue"},\\
s_{16}:= s_1 \wedge \neg s_1 = &\text{"The apple is blue and the apple is not blue"}.
\end{align*}
Can we write down any new combination of statements?
Yes of course, for example we can think of
\begin{align*}
s_{17}:=s_5\vee s_1&=\text{"The apple is blue or the dog is red or the dog is red"},
\end{align*}
but soon we realise that this statement is actually equivalent to statement of $s_5$. 
We can think of the statement
\begin{align*}
s_{18}:=s_2\vee \neg s_2=\text{"The dog is red or the dog is not red"}
\end{align*}
as a new statements, but soon we can convince ourselves  that the information this statement is conveying is the same as the one conveyed by statement $s_{15}$, or by the much more complex statement
\begin{align*}
\tilde{s}_{18}=\left((s_1\wedge s_2) \vee (\neg s_1 \wedge \neg s_2)\right)\vee \left( (s_1\wedge \neg s_2)\vee (\neg s_1 \wedge s_2) \right).
\end{align*}
We can try to play this game with any other statement, just to recognise that the $16$ statements $s_1,\dots,s_{16}$ are all the statements we can form.
%This fact seems quite strange. In fact if from a set of $2$ statements we can construct $10$ new statements we would imagine that from a set of $14$ statements we should be able to construct many more statements, and, potentially, iterating the procedure, obtaining as many statements as we desire!
%It seems our way of processing information does not allow for generating infinite statements from finite statements. This is already some interesting property of our mechanism of processing information.

In order to not rely simply on our intuition for the process of identifying equivalent statements, we  now formalise the rules we just used for combining statements.

We start by defining a partially artificial playground that will serve us for the scope of simplifying the treatment. In all our work we are going to consider the case where the total number of possible statements is finite. We consider this a quite artificial restriction, because when thinking intuitively one considers the set of statements as infinite. Every set of statements we are going to consider is \textit{closed} with respect to the operations $\wedge$, $\vee$ and $\neg$. 
We consider statements that form the structure of a Boolean lattice. A Boolean lattice is a distributive complemented lattice.
A lattice of statements is an algebraic structure $(\mathcal{S},\wedge,\vee)$ consisting of a set of statements $\mathcal{S}$ and the two binary operations $\wedge$, $\vee$ satisfying the following rules
\begin{subequations}
\label{eq:axioms-lattice}
\begin{align}
&x \vee x=x, \quad x \wedge x=x, \quad& \text{\textit{(Idempotency)}} \label{eq:idempotency}\\
&x \vee y=y \vee x, \quad x \wedge y=y \wedge x, \quad & \text{\textit{(Commutativity)}} \label{eq:commutativity-statements}\\
&x \vee(y \vee z)=(x \vee y) \vee z,\quad& \text{\textit{("or" associativity)}} \label{eq:or-associativity}\\
& x \wedge(y \wedge z)=(x \wedge y) \wedge z,\quad& \text{\textit{("and" associativity)}} \label{eq:and-associativity}\\
&x \vee(x \wedge y)=x \wedge(x \vee y)=x.&\text{\textit{(Absorption)}}\label{eq:absorbtion}
\end{align}
\end{subequations}
The symbol $"="$ is used meaning equivalence between statements. 
A lattice is a complete lattice $(\mathcal{S}, \wedge, \vee, \top, \perp)$ if it has a "top" element $\top$ and a "bottom" element $\perp$ such that for all statements  $s\in\mathcal{S}$
\begin{subequations}
\label{eq:axioms-complete}
\begin{align}
s \vee \top =\top,\\
s\wedge \perp = \perp.
\end{align}
\end{subequations}
The top and bottom element are somehow special elements, they are not included in the set of statements and they are present in every complete lattice.
In a complete lattice  $(\mathcal{S}, \wedge, \vee, \top, \perp)$ there can be statements $s_{i}\in\mathcal{S}$ such that there do not exist any $s_{j},s_{k}\in\mathcal{S}$, with $s_{j}\neq s_{k}$, such that $s_{i}=s_{j}\vee s_{k}$. These are called \textit{atomic statements}. We denote the set of atomic statements of $\mathcal{S}$ as $\mathcal{A}(\mathcal{S})$.
The top element can be thought as an equivalent statement to the combination of all the statements in $\mathcal{S}$ through the binary operator $\vee$, $\top=\bigvee_{s\in\mathcal{S}}s$, while the bottom element can be thought as equivalent to the combination of any two different atomic statements through the binary operator $\wedge$, $\perp=s_{i}\wedge s_{j}$ where $s_{i}\neq s_{j}$ are atomic statements.
The nomenclature "top" and "bottom" for $\top$ and $\perp$ comes directly from the graphical representation of lattice introduced in the next section.
A complete lattice of statements $(\mathcal{S}, \wedge, \vee, \top, \perp)$ is associative if the following rules apply
\begin{subequations}
\label{eq:axioms-distributivity}
\begin{align}
&x \wedge(y \vee z)=(x \wedge y) \vee(x \wedge z), \quad &\text{\textit{(Distributivity of $\wedge$ over $\vee$)}} \label{eq:distributivity-and}\\
&x \vee(y \wedge z)=(x \vee y) \wedge(x \vee z). \quad &\text{\textit{(Distributivity of $\vee$ over $\wedge$)}}\label{eq:distributivity-or}
\end{align}
\end{subequations}
A lattice of statements is complemented if for every statement $s\in \mathcal{S}$ there exists a unique statement $\neg s \in \mathcal{S}$ (the negation of $s$) such that
\begin{subequations}
\label{eq:axioms-complemented}
\begin{align}
&x\vee \neg x = \top,\quad&\text{\textit{("or" complementation)}}\\
&x\wedge \neg x = \perp. \quad&\text{\textit{("and" complementation)}}\label{eq:and-complementation}
 \end{align}
 \end{subequations}

These rules define logical equivalence between statements without relying on the concept of truth value. 

Often lattices are defined in the context of partially ordered sets as partially ordered sets in which each element has a unique supremum and infimum. 
One can pass from the lattice structure of the algebra of the statements to the lattice structure of partially ordered set defining the following partial order relation between statements: for any two statements $s_{i},s_{j}\in \mathcal{S}$
\begin{equation}
s_{i}\leq s_{j} \iff s_{i} \wedge s_{j} =s_{i} \text{ and } s_{i}\vee s_{j}=s_{j}.
\end{equation}
Furthermore, it is possible \cite{knuth2006}, and it will result convenient, to consider statements as sets. In this picture the binary operators $\wedge$ and $\vee$ corresponds respectively to the intersection $\cap$ and union $\cup$ of sets. Distributive lattices are lattices of sets ordered by set inclusion $\subseteq$.

One can represent graphically the structure of the algebra via an Hasse diagram. Because of their simplicity and intuitive representation, we make heavy use of the graphical language provided by the Hasse diagrams, and in Appendix \ref{appendix:graphical-notation} we provide a quick guide on how to use them.

Often, with an abuse of notation we will write $\mathcal{S}$ instead of $\mathcal{A}(\mathcal{S})$, for example, when writing $\mathcal{S}_3=\{s_1,s_2,s_3\}$ we are implicitly saying that $s_i\wedge s_j=\perp \; \forall i,j\in[1,3]\; ,i\neq j$ and that we are are considering all elements $\{\perp,s_1,s_2,s_3,s_1\wedge s_2,s_1\wedge s_3,s_2\wedge s_3,\top\}$.

We conclude this section observing that the statements $\{s_1,\dots,s_{16}\}$ introduced in the initial example form indeed a boolean lattice.

\subsection{Graphical notation for lattices}
\label{appendix:graphical-notation}
In this work we use the graphical notation of the Hasse diagrams.
Referring to Figure \ref{fig:first-Hasse}, statements are represented as filled circles with the label of the name of the statement inside (panel $(a)$). For any set of circles, that do not include the statements $\top$ or $\perp$, we have the following graphical rules. If the circle of a statement $s$ is connected with a line from below to the set of statements $\{s_1,s_2,\dots,s_N\}$ it means that $s=s_1\vee s_2\vee \dots \vee s_N$ (panel $(b)$). If the circle of a statement $s$ is connected with a line from above to the set of statements $\{s_1,s_2,\dots,s_n\}$ it means that $s=s_1\vee s_2 \vee \dots s_n$ (panel $(c)$). 

 \begin{figure}[h!]
\centering
\includegraphics[width=1\textwidth]{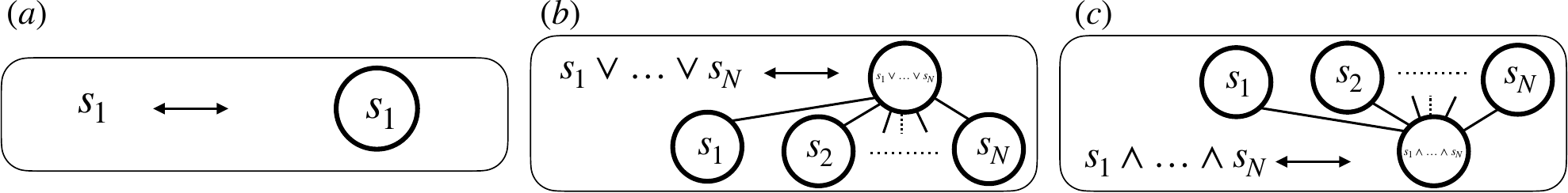}
\caption{Statements are represented as filled circles with the label of the name of the statement inside (panel $(a)$). For any set of circles, that do not include the statements $\top$ or $\perp$, we have the following graphical rules. If the circle of a statement $s$ is connected with a line from below to the set of statements $\{s_1,s_2,\dots,s_N\}$ it means that $s=s_1\vee s_2\vee \dots \vee s_N$ (panel $(b)$). If the circle of a statement $s$ is connected with a line from above to the set of statements $\{s_1,s_2,\dots,s_n\}$ it means that $s=s_1\vee s_2 \vee \dots s_n$ (panel $(c)$). \label{fig:first-Hasse}}
\end{figure}

Furthermore for any finite set of statements $\mathcal{S}$ with atomic elements $\{s_1,s_2,\dots,s_D\}$ one can introduce in the graphical representation the $\perp$ and the $\top$  elements as in Figure \ref{fig:top-bottom-levels}
\begin{figure}
\centering
\includegraphics[width=0.5\textwidth]{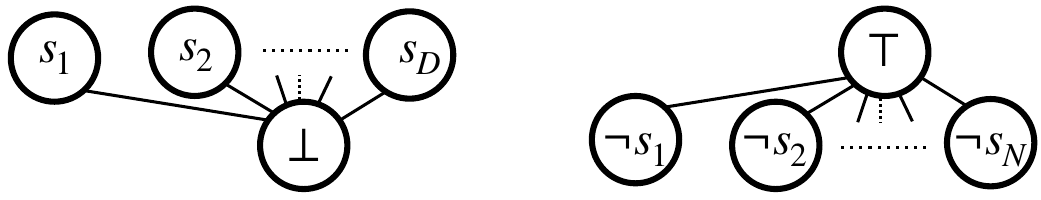}
\caption{Graphical representation of the top and bottom elements.\label{fig:top-bottom-levels}}
\end{figure}
where we used the notation $\neg s_i$ to indicate the union of all the atomic statements except $s_i$ (for example $\neg s_D=s_1\vee s_2\vee \dots \vee s_{D-1}$).

Using this graphical notation we can represent all statements of the finite set of statements $\mathcal{S}_3=\{a,b,c\}$ as in Figure \ref{fig:S_3}.

\begin{figure}
\centering
\includegraphics[width=0.25\textwidth]{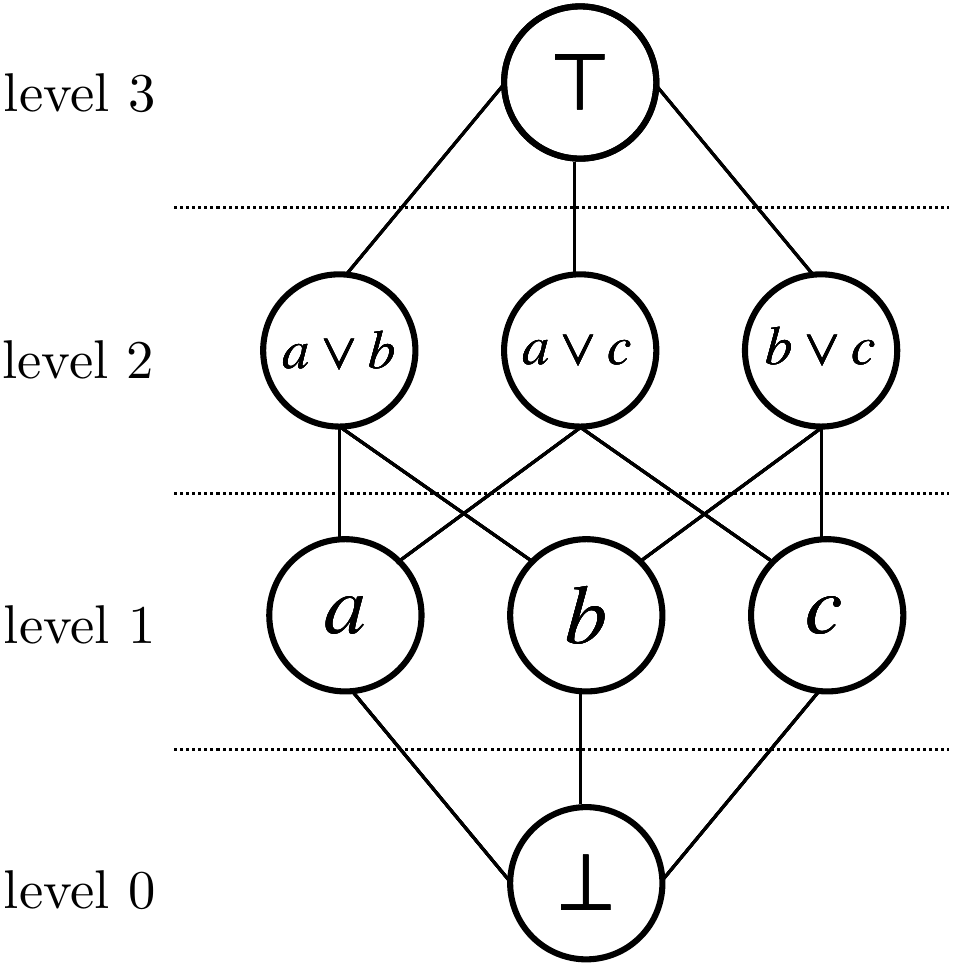}
\caption{Graphical representation of the finite set of statements $\mathcal{S}_3=\{a,b,c\}$. The atomic statements are collected on the first line above the line with the $\perp$ element. We denote each line of the lattice starting from the bottom with an integer number starting from zero, we call this number the level of the lattice. At level $0$ we have just the $\perp$ element. At level $1$ we have all the atomic statements, at level $1<l<3$ we have the statements obtained as the union of $l$ different atomic statements. At level $3$ we have the $\top$ element. Note how $s_1\vee s_2$ has two lines entering from below connected to $s_1$ and $s_2$ meaning that it is the union of $s_1$ and $s_2$. Note also how $s_2$ has two lines entering from above connected to $s_1 \vee s_2$ and $s_2 \vee s_3$, meaning that $s_2$ is the intersection of $s_1 \vee s_2$ and $s_2 \vee s_3$. In general going up corresponds to taking unions, going down corresponds to taking intersections. \label{fig:S_3}}
\end{figure}

For a generic set of statements obtained starting from $D$ atomic statements $\mathcal{S}_D=\{s_1,s_2,\dots,s_D\}$ the graphical representation looks as follows. The atomic statements are collected on the first line above the line with the $\perp$ element. Each line of the lattice starting from the bottom is labelled by an integer number starting from zero, we call this number the level of the lattice. For example at level $0$ we have just the $\perp$ element. At level $1$ we have all the atomic statements, at level $1<l<D$ we have the statements obtained as the union of $l$ different atomic statements. At level $D$ we have the $\top$ element.

We conclude this section with a graphical representation of the boolean lattice of the statements $\{s_1,\dots,s_{16}\}$ introduced in section \ref{appendix:algebra-of-statements}.

\begin{figure}[h!]
\centering
\includegraphics[width=0.35\textwidth]{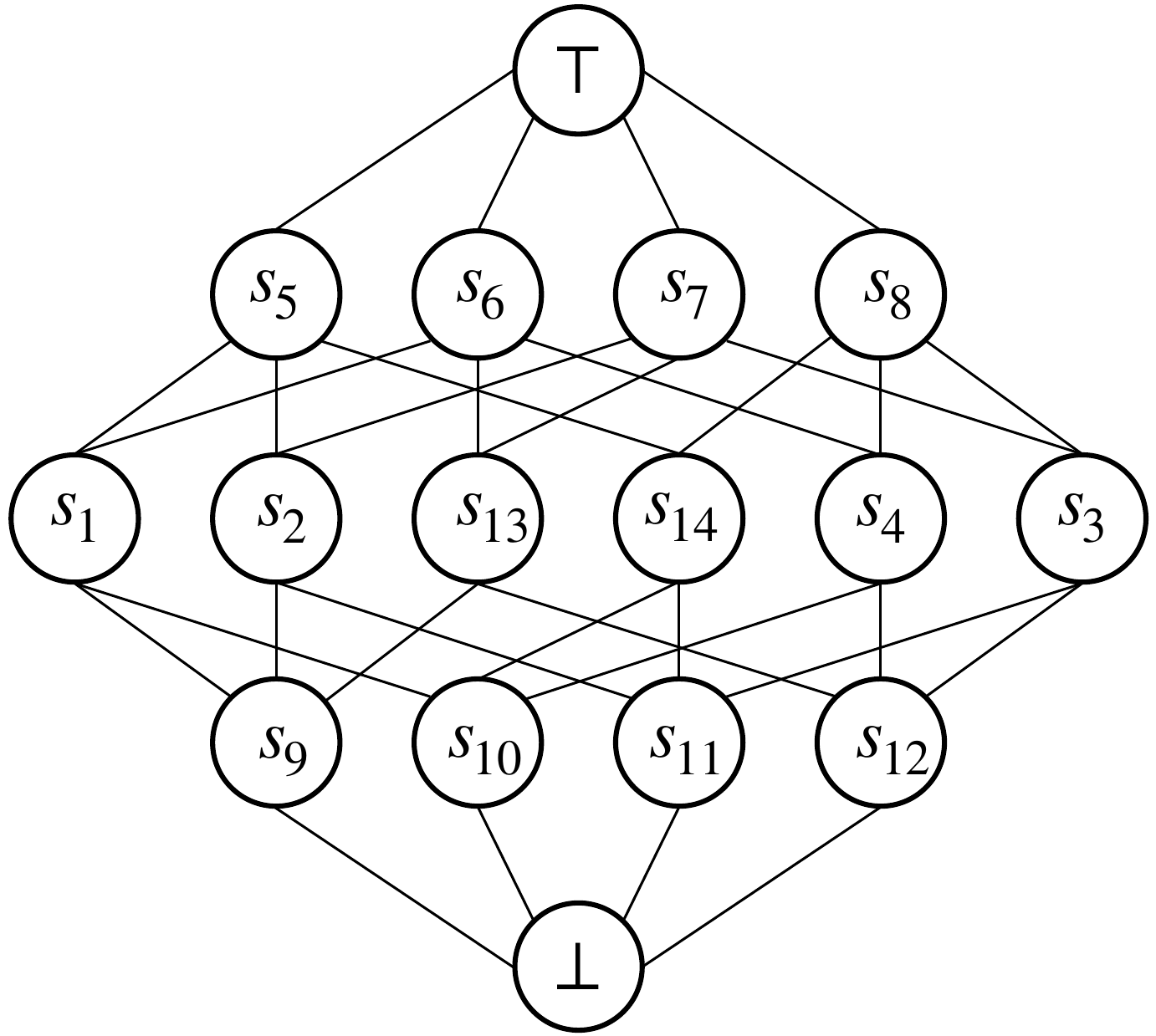}
\caption{Graphical representation of the statements of the example in section \ref{appendix:algebra-of-statements}.}
\end{figure}

 \subsection{Quasi-probabilities from the lattice of statements}
 \label{appendix:quasi-probabilities-from-lattice}
 
 \subsubsection{Probabilities from the lattice of statements}
 \label{appendix:probabilities-from-statements-intro}
Crucial to the derivation of probability theory from the lattice of statements is the concept of \textit{implication}. The idea of implication is usually understood in the common language as a concept connected with the truth value of statements. For example one would say that the statement $a$ implies the statements $b$ if the condition of $a$ being true is sufficient for $b$ to be true. 
One can abstract the concept of implication from the truth value and simply connect it to the concept of information contained in a statement, or even better to the kind of relation subsisting between two statements. 
 We say that a statements $s_i$ implies another statement $s_j$, written $s_i \to s_j$, if $s_i \vee s_j=s_j$ and $s_i \wedge s_j = s_i$.  One can play with this definition and see how it makes sense. In a poetic sense, one can say that a statement $s_i$ implies a statement $s_j$ whenever $s_i$ is a synecdoche for $s_j$ within the limited dictionary of the set of statement $\mathcal{S}$ considered.

In his works \cite{cox1946b,cox1961} R.T.Cox realised that by generalising the concept of implication to the concept of \textit{degree of implication} between statements one can obtain probability theory. The real numbers associated to these degrees of implication, when considered as representing the degree to which we believe one statement implies another statement, are equivalent to probabilities.

The methodologies and concepts introduced firstly by R. T. Cox.\cite{cox1961} for deriving probability theory from logic have been subsequently revisited, clarified and developed further by many other authors. A beautiful exposition of the topic  due to  E. T. Jaynes can be found in the short work \cite{jaynes1988} or in the first chapters of the book \cite{jaynes2003b} where probability is derived as the unique method for inference that a robot equipped with common sense (a somehow anachronistic concept) would use. Other expositions can be found in the many works of  A. Caticha (here I cite the unpublished \cite{caticha}  that collects many published works of him and \cite{caticha2009}), as well as in K. H. Knuth \cite{knuth2004,knuth2005a} and H. Reichenbach \cite{reichenbach1971} to cite some.
% The following presentation has not the pretentious to be exhaustive, as it can be found exhaustively studied in the cited work, but it is mainly a presentation meant to highlight some of the principal methodologies and requests necessaries to derive probability theory.
 In the next section, we are going to relax some requests made by these authors, ending up obtaining a derivation of quasi-probabilities instead of probabilities.
The derivation of quasi-probabilities we are going to present is by no means complete and exhaustively explained as the one of probability in the aforementioned works, nevertheless a complete treatment of the topic, where the differences with the derivation of probability are highlighted, is possible. 
 We reserve this analysis and the relative consideration for a future work, presenting here the bare minimum set of tools and consistency rules that we need in the current work.

  \subsubsection{Quasi-probabilities from the algebra of statements}
  \label{appendix:quasi-prob-derivation}
Similarly to the definition of the plausibility of Jaynes \cite{jaynes2003b} and the definition of degree of inclusion of Knuth \cite{knuth2005a}, let us define the "\AR" function $\mathcal{R}$ as a function that takes as argument a pair of statements (including $\top$) and returns a real number. The idea is for the function $\mathcal{R}$ to quantify in some way the relationship between two statements. Furthermore we ask for $\mathcal{R}(a,b)$ to return $1$ whenever $b$ implies $a$ and to return $0$ whenever $a\wedge b = \perp$ (i.e. $b$ implies $\neg a$),
 the function so defined is quite generic, but proceeding with further requests one is able to characterise $\mathcal{R}$ by some specific properties that would identify it as a "conditional quasi-probability function".
 
 These further requests divide in two groups.
The first important thing to keep in mind is that the arguments of $\mathcal{R}$ form and algebra with some specific properties. We want $\mathcal{R}$ to be consistent with these properties.
The second kind of assumptions on $\mathcal{R}$ are the one that we are going to call "\textit{common sense assumptions}", borrowing from the similar terminology used by Jaynes. Common sense assumptions are some qualitative properties that we would expect for $\mathcal{R}$ to satisfy in order to encode the idea we have about "plausibility" and that are used in the derivations of probability theory\footnote{Cox's theorem and deriving probability from propositional logic have long been debated. Different authors use various sets of additional assumptions to achieve the same result (the rules of probability theory). Common sense assumptions refer to any set of additional assumptions sufficient to derive these rules.}. Our derivation is an effort to use as less common sense assumptions as possible when talking about $\mathcal{R}$ allowing us to characterise a more general form for $\mathcal{R}$ that is constrained almost solely by the request to be consistent with the lattice structure.

\paragraph{The sum rule for atomic statements}
\label{appendix:sum-rule}
For example, we might ask how $\mathcal{R}$ behaves when one of its argument is a composition of statements.
Consider the atomic statements $a,b$ such that $a\wedge b=\perp$ and a third statement $t$ such that $a\vee t = t$ and $b\vee t=t$. 
We consider $\mathcal{R}(a\vee b,t)$  to be a function of the following arguments:
\begin{equation}
\mathcal{R}(a \vee b, t)=\tilde{F}\left[ \mathcal{R}(a,t),\mathcal{R}(b,t),\mathcal{R}(a,b\wedge t),\mathcal{R}(b,a\wedge t),\mathcal{R}(a\wedge b,t) \right],
\end{equation}
for a not yet specified function $\tilde{F}$ with domain $\mathbb{R}^5$ and codomain in $\mathbb{R}$.
The choice of these specific arguments is an additional assumption that we are making, this is an assumption that appears often in the literature, to some authors this assumption seems so natural that it is not even considered as an assumption.\footnote{For example, Knuth in \cite{knuth2004} considers as possible arguments only $\mathcal{R}(a,t),\mathcal{R}(b,t)$, justifying this choice before his equation (2) as follows "\textit{As $a\wedge b=\perp$, this degree of inclusion can only be a function of $\phi(a,t)$ and $\phi(b,t)$}", where his $\phi$, that he calls degree of inclusion, is our $\mathcal{R}$. In chapter 2 of \cite{caticha} it is highlighted how the choice of a smaller set of arguments with respect to ours is general enough with the following sentence "\textit{the arguments of $F$ [...] include all four possible degrees of belief in $a$ and $b$ in the context of $c$ and not any potentially questionable subset.}", where $c$ of Caticha corresponds to our $t$.}.In this work, we do not assert that this assumption is so self-evident as to require no justification. Indeed, for example, we do not consider the possibility of a dependence of $\mathcal{R}(a \vee b,t)$ on $\mathcal{R}(\neg a,t)$. However, this is a broad assumption that grants us considerable freedom and aligns with standard methodologies in the literature.\footnote{The derivations of probability theory that closely adheres to the principles of Cox's original work typically begin with the derivation of the product rule, rather than the sum rule. In this framework, a comparable assumption is made regarding the argument of the function analogous to our function 
$F$. However, in this context, the assumption has been regarded as far from obvious and has sparked extensive debate \cite{tribus1969,smith1990,caticha,jaynes2003b}.}

 In our case we have that $a\wedge b=\perp$ so 
 \begin{equation}
 \mathcal{R}(a \vee b, t)=\tilde{F}\left[ \mathcal{R}(a,t),\mathcal{R}(b,t),0,0,0 \right], 
 \end{equation}
thus effectively reducing the above functional relation to
 \begin{equation}
 \label{eq:sum-rul-phi}
 \mathcal{R}(a \vee b, t)= F\left[ \mathcal{R}(a,t),\mathcal{R}(b,t) \right], 
 \end{equation}
 where $F$ is a not yet specified function with domain in $\mathbb{R}^2$ and codomain in $\mathbb{R}$.
% At this point the common sense is evoked, to guess one property of $F$.
%This property is
%\begin{equation}
%F[x,1]=F[1,x]=1,
%\end{equation}
% that means that if $c$ implies $a$ or $c$ implies $b$, then $c$ must imply $a\vee c$.
 
 Now we consider a fourth statement $c$ such that $a\wedge b=a\wedge c=b\wedge c=\perp$. Because of associativity of the statements (property \eqref{eq:or-associativity}) we have:
 \begin{equation}
 a\vee (b \vee c)= (a \vee b) \vee c,
 \end{equation}
so that we should expect the following:
 \begin{equation}
 \mathcal{R}( a\vee (b \vee c),t)= \mathcal{R}((a \vee b) \vee c,t).
 \end{equation}
Using \eqref{eq:sum-rul-phi} we obtain the equality
\begin{equation}
\label{eq:transitivity-with-phi}
F\left[\mathcal{R}(a,t), F\left[\mathcal{R}(b,t),\mathcal{R}(c,t) \right] \right]=F\left[F\left[\mathcal{R}(a,t),\mathcal{R}(b,t) \right] ,\mathcal{R}(c,t) \right].
\end{equation}
Renaming $x=\mathcal{R}(a,t)$, $y=\mathcal{R}(b,t)$, and $z=\mathcal{R}(c,t)$ we can rewrite equation \eqref{eq:transitivity-with-phi} as
\begin{equation}
\label{eq:associativity-equation}
F[x,F[y,z]]=F[F[x,y],z].
\end{equation}
This functional equation is known as the associativity equation. Introduced by Abel in \cite{abel1826} and subsequently thoroughly studied by Aczel \cite{aczel1966,aczel2004}, the associativity equation lies at the core of the derivation of probability theory. In fact, if some common sense assumptions on $F$ are made \cite{aczel1966,aczel2004,alsina2006}(see section \ref{appendix:additional-assumptions}), the functional equation \eqref{eq:associativity-equation} is known to have a unique solution (up to multiplications of the arguments with a scalar), which has the form
\begin{equation}
\label{eq:sol-associativity-equation}
F(x,y)=f^{(-1)}\left(f(x)+f(y)\right),
\end{equation}
where $f$ is an invertible function with domain $\mathbb{R}$ and codomain $\mathbb{R}_{+}$ and $f^{(-1)}(x)$ is the inverse of $f(x)$.
This unique solution lies at the core of the characterisation of probability theory. 
In our derivation we are not asking for any common sense assumption on $F$ in general. Nevertheless one can check, by simply substituting \eqref{eq:sol-associativity-equation} in \eqref{eq:associativity-equation}, that, for any invertible  $f:\mathbb{R}\to \mathbb{R}$, equation \eqref{eq:sol-associativity-equation}  is a solution of $\eqref{eq:associativity-equation}$.   In this case, there is no uniqueness theorem for the functional form of $F$.
It will make sense to impose the common sense assumptions when the function $\mathcal{R}(x,y)$ is restricted to the domain of elements within an accessible sublattice, meaning that $x$ and $y$ are accessible statements. Consequently, within the accessible sublattice,  the form of the solution to the functional equation  \eqref{eq:associativity-equation} is indeed unique, taking the form \eqref{eq:sol-associativity-equation} and we anticipate to recover the principles of probability theory.

In the remaining parts of the lattice, the solution for $F$ is not unique. 
We shall assume the same functional form for  $F$ as in \eqref{eq:sol-associativity-equation}, with the justification for this choice deferred to future work. This assumption is not too restrictive and it  allows the general solution to include an $f$ that is an arbitrary invertible function from $\mathbb{R}$ to $\mathbb{R}$. 

In summary, we find that $F$ has a solution of the form \eqref{eq:sol-associativity-equation}, where $f$ is an invertible function with codomain restricted to $\mathbb{R}_{+}$ when evaluated on $\mathcal{R}(a,t)$ with $a$ and $t$ accessible statements. Additionaly, we fix $f$ such that $f(1)=1$ and $f(0)=0$.

%{ This unique solution lies at the core of the characterisation of probability theory. 
%In our derivation we are not asking for any additional assumption on $F$. In this case, to the best of the author's knowledge, there is no uniqueness theorem for the functional form of $F$. Nevertheless one can check, by simply substituting \eqref{eq:sol-associativity-equation} in \eqref{eq:associativity-equation}, that \eqref{eq:sol-associativity-equation} with  $f:\mathbb{R}\to \mathbb{R}$ is a solution of $\eqref{eq:associativity-equation}$. We can fix $f$ such that $f(1)=1$ and $f(0)=0$. Furthermore, since we will demand for our general solution to behave as a probability when further conditions are imposed, we can assume \eqref{eq:sol-associativity-equation} to be the functional form of the  unique solution of \eqref{eq:associativity-equation}.
%}

We consider the re-parametrisation  $\tilde{\mathcal{R}}(x,t)=f(\mathcal{R}(x,t))$ as the new "\AR" function, one obtains for generic statements $a$ and $b$ such that $a\wedge b = \perp$ that equation \eqref{eq:sum-rul-phi} can be recasted in the form:
\begin{equation}
\label{eq:sum-rule}
\tilde{\mathcal{R}}(a\vee b,t)=\tilde{\mathcal{R}}(a,t)+\tilde{\mathcal{R}}(b,t)
\end{equation}
with $\tilde{\mathcal{R}}$ with codomain $\mathbb{R}$.
Given that the function $f$ is invertible and has been defined such that $f(1)=1$ and $f(0)=0$, either $\mathcal{R}$ or $\tilde{\mathcal{R}}$  can be used interchangeably as they can both be interpreted as "amount of relation", one is simply a reparametrisation of the other. We choose to use $\tilde{\mathcal{R}}$ and the rules associated with this choice. 

The choice of $\tilde{\mathcal{R}}$ in place of $\mathcal{R}$ is, as it currently stands, arbitrary. Ideally, this choice should be justified by an additional, more "intuitive" assumption, deferred for future work. Such an approach would enable a \textit{fully} satisfactory derivation of quasi-probability rules based on a set of plausible assumptions.

From equation \eqref{eq:sum-rule} we see that the "\AR" obeys a sum rule analogous to the one of probability theory.

\paragraph{The sum rule}
We now extend the sum rule to non-atomic statements. We do this following \cite{knuth2004}.
Consider two statements $x$ and $y$ and write them as the composition of atomic statements
\begin{align}
x&=(p_1\vee p_2 \vee \dots \vee p_n)\vee (r_1\vee r_2 \vee \dots \vee r_k), \nonumber\\
y&=(r_1\vee r_2 \vee \dots \vee r_k) \vee (w_1\vee w_2 \vee \dots \vee w_m).
\end{align}
Clearly 
\begin{align}
x \wedge y & =(r_1\vee r_2 \vee \dots \vee r_k),\nonumber \\
x \vee y&=(p_1\vee p_2 \vee \dots \vee p_n)\vee(r_1\vee r_2 \vee \dots \vee r_k) \vee (w_1\vee w_2 \vee \dots \vee w_m).
\end{align}
Using the sum rule  \eqref{eq:sum-rule} repeatedly on non intersecting terms we can write
\begin{equation}
\REL(x \vee y, t)=\sum_{i=1}^{n} \REL(p_i,t)+\sum_{i=1}^{k} \REL(r_i,t)+\sum_{i=1}^{m} \REL(w_i,t).
\end{equation}
Adding and removing $\sum_{i=1}^{k} \REL(r_i,t)$ and grouping terms properly, one obtains
\begin{equation}
\REL(x\vee y,t)=\REL(x,t)+\REL(y,t)-\REL(x \wedge y, t).
\end{equation}
We obtained that $\REL$ obeys the general sum rule of probability theory.

\paragraph{Product rule}
In deriving the product rule for $\REL$ we again follow \cite{knuth2004}.
We want to find a function $G$, consistent with the structure of the algebra of statements, such that it allows us to write $\REL(x\wedge y, t)$ without referring to $\REL(x \vee y, t)$. Following the analysis of \cite{tribus1969,smith1990}  (for a short summary see section 2.3.1 of \cite{caticha}) we know that this function has to have as arguments
\begin{equation}
\label{eq:product-associativity-equation}
\REL(x \wedge y, t) = G\left[\REL(x,t), \REL(y,x\wedge t) \right].
\end{equation}
Similarly to the case of the sum rule, we consider two statements $a$ and $b$ such that $a \wedge b =\perp$ and two statements $r$ and $s$ such that $r \wedge s = \perp$.
Using the sum rule and the distributivity rule \eqref{eq:distributivity-and} we obtain
\begin{align}
\REL\left(a\wedge (r\vee s),t\right)&=G\left(\REL(a,t),\REL(r\vee s,a\wedge t) \right)=G\left(\REL(a,t),\REL(r,a\wedge t)+\REL(s,a\wedge t) \right),\label{eq:R-single-G}\\
\REL\left(a\wedge (r\vee s),t\right)&=G\left(\REL(a,t), \REL(a,r\wedge t) \right)+G\left(\REL(a,t),\REL(a,s\wedge t) \right)\label{eq:R-double-G}.
\end{align}
Equating \eqref{eq:R-single-G} with \eqref{eq:R-double-G} and  substituting $u=\REL(a,t)$, $v=\REL(r,a\wedge t)$, and $w=\REL(s,a\wedge t)$ we obtain the functional equation
\begin{equation}
\label{eq:functional-equation-product}
G(u,v+w)=G(u,v)+G(u,w).
\end{equation}
Now we define $z=v+w$ and we have
\begin{align}
\frac{\partial}{\partial v}&=\frac{\partial z}{\partial v}\frac{\partial }{\partial z}=\frac{\partial}{\partial z},\\
\frac{\partial}{\partial w}&=\frac{\partial z}{\partial w}\frac{\partial }{\partial z}=\frac{\partial}{\partial z},
\end{align}
thus $\frac{\partial^2}{\partial z^2}=\frac{\partial}{\partial v}\frac{\partial}{\partial w}$.
Applying this second derivative to equation \eqref{eq:functional-equation-product} we find that
\begin{equation}
\frac{\partial^2}{\partial z^2}G(u,z)=\frac{\partial}{\partial v}\frac{\partial}{\partial w}(G(u,v)+G(u,w))=0,
\end{equation}
thus $G$ is linear in its second argument so that the most general form is given by:
\begin{equation}
\label{eq:product-linear-form}
G(u,v)=A(u)v+B(u).
\end{equation}
Plugging equation \eqref{eq:product-linear-form} into equation \eqref{eq:functional-equation-product} we see that $B(u)=0$ for every $u$.
Proceeding with a similar reasoning starting from $\REL((a\vee b)\wedge r)$ we find out that $G$ is also linear in its first argument.
We finally obtain 
\begin{equation}
G(u,v)=Cuv,
\end{equation}
where $C$ is a constant that we are going to set to $1$.
We obtained that $\REL$ obeys the product rule of probability theory.
From the product rule, together with the commutativity property \eqref{eq:commutativity-statements}, it is possible to obtain Bayes' theorem.

\paragraph{Quasi-probabilities}
We are finally in the condition of defining the function $Q$ that maps the elements of a generic set of statements $\mathcal{S}$ to real numbers. For a generic statement $s\in \mathcal{S}$ the action of $Q$ is defined as
\begin{equation}
Q(s)=\REL(s,\top).
\end{equation}
Furthermore defining $Q(s_i\mid s_j)=\REL(s_i,s_j\wedge \top)=\REL(s_i,s_j)$, the following properties of $Q$ can be easily derived:
\begin{itemize}
\item $Q(\top)=1$,
\item $Q(\perp)=0$,
\item $Q(s_i)\in \mathbb{R}$,
\item $\sum_{s\in \mathcal{A}(\mathcal{S})}Q(s)=1$,
\item $Q(s_i \vee s_j)=Q(s_j \vee s_i)=Q(s_i)+Q(s_j)-Q(s_i\wedge s_j)$,
\item $Q(s_i \wedge s_j)= Q(s_j\wedge s_i)$,
\item $Q(s_i \wedge s_j)= Q(s_i)Q(s_i\mid s_j)$.
\end{itemize}
The function $Q$ is a \textit{valuation} defined on a distributive lattice \cite{klain1997}, and its value on every statement is uniquely defined by the set of values $\{Q(s)\}_{s\in\mathcal{A}(\mathcal{S})}$.

The function $Q$ associates a quasi-probability distribution value to each statement in $\mathcal{S}$. 
We see that from the algebra of statements there is a natural assignment of $Q$ of real numbers to each statement such that this assignment respects the sum rule, the product rule and Bayes theorem and thus behaves as a probability assignment. We remark here that we never assumed probability theory, we derived these rules from the request of consistency with the algebra of statements. 
Furthermore, for every lattice of statements we have that the assignments of $Q$ are completely specified by the assignment of $Q$ on the set of atomic statements. Thus the assignments of $Q$ are completely specified by the choice of a vector $\vec{q}=\{Q(s_1),Q(s_2),\dots, Q(s_D)\}\in\mathbb{R}^D$, where $\mathcal{A}(\mathcal{S})=(s_1,s_2,\dots,s_D)$ are the atomic statements. We have complete freedom in the choice of the elements of the vector $\vec{q}$, with the only constraint they sum to $1$, $\sum_{i=1}^{D}q_i=1$.
If, as lobbied by Jaynes, probability is a generalisation of Aristotelean logic, then language of quasi-probability is a \textit{further} generalisation of Aristotelean logic.

\subsubsection{Probability from the algebra of statements}
\label{appendix:additional-assumptions}
 In appendix \ref{appendix:probabilities-from-statements-intro} we introduced the idea of deriving probability theory as a relaxation of classical logic. In classical logic one associates to each statement a binary  truth value. A statement is either true or false. In probability theory one can assign to each statement a real number from the continuous interval $[0,1]$. 

We interpret the assignment of a real number to each element of the lattice as a relaxation of the binary truth assignment.  
Formally, for any lattice $L$, we consider the set $S^{L}_V$ of all functions from elements of the lattice, or statements, to binary truth values. 
We call  $V \in S^{L}_V$ a logical-valuation. Each logical-valuation corresponds to a different assignment of truth values to all the elements of the lattice. We denote $\tilde{S}^{L}_V$ the set of all the logical-valuations corresponding to admissible truth values' assignments (see section \ref{sec:section-II}), that is one has 
\begin{align}
\label{eq:logical-requests-1}
V(\top)=\TTT, \qquad V(\perp)=\FFF,
\end{align}
and for any $a,b\in L$
\begin{align}
\label{eq:logical-requests-2}
&(V(a)=\TTT \mbox{ and } ~V(b)=\FFF)  ~\mbox{ or }~ (V(a)=\FFF \mbox{ and }  ~V(b)=\TTT) & \implies & V(a\vee b)=\TTT \mbox{ and }  V(a \wedge b)=\FFF, \nonumber \\
&(V(a)=\FFF \mbox{ and } ~V(b)=\FFF) & \implies & V(a\vee b)=\FFF \mbox{ and } ~V(a \wedge b)=\FFF,  \nonumber \\
&(V(a)=\TTT \mbox{ and } ~V(b)=\TTT) & \implies & V(a\vee b)=\TTT  \mbox{ and } ~V(a \wedge b)=\TTT, \nonumber \\
&V(a)=\FFF & \implies & V(\neg a)=\TTT, \nonumber \\
&V(a)=\TTT & \implies & V(\neg a)=\FFF.
\end{align}
%Because of this requests every logical-valuation $V\in \tilde{S}^{L}_{V}$ corresponds to an assignement of truth labels to the element of the lattice $L$ that is compatible with the truth tables.
Furthermore, for any lattice $L$, we consider the set $S^{L}_{\mathcal{V}}$ of all functions $\mathcal{V}$  from elements of the lattice, or statements, to real values. We call $\mathcal{V} \in S^{L}_{\mathcal{V}}$ a real-valuation. We further define  $\tilde{S}^{L}_{\mathcal{V}}$, the subset of real-valuation that we consider relaxations of a logical-valuations $\tilde{S}^{L}_V$, as the set of real-valuations that satisfy compatibility with truth tables' rules and compatibility with the structure of the lattice.
A real valuation $\mathcal{V}$ is compatible with the truth tables if
\begin{align}
\label{eq:real-requests-1}
\mathcal{V}(\top)=\beta, \qquad \mathcal{V}(\perp)=\alpha,
\end{align}
for two fixed constants $\alpha,\beta \in \mathbb{R}$ (we assume $\alpha<\beta$, but all the following works for $\alpha>\beta$, while the case $\alpha=\beta$ is discarded if we want the analogy with the logical-valuation to hold),
%We denote $\tilde{S}^{L}_{\mathcal{V}}$ the set of all real-valuations such 
and for any two statements $a,b$ in the lattice $L$:
\begin{align}
\label{eq:real-requests-2}
&(\mathcal{V}(a)=\alpha \mbox{ and } ~\mathcal{V}(b)=\beta)  ~\mbox{ or }~ (\mathcal{V}(a)=\beta \mbox{ and } ~\mathcal{V}(b)=\alpha) & \implies & \mathcal{V}(a\vee b)=\beta \mbox{ and } \mathcal{V}(a \wedge b)=\alpha, \nonumber \\
&(\mathcal{V}(a)=\alpha \mbox{ and } ~\mathcal{V}(b)=\alpha) & \implies & \mathcal{V}(a\vee b)=\alpha \mbox{ and } ~\mathcal{V}(a \wedge b)=\alpha,  \nonumber \\
&(\mathcal{V}(a)=\beta \mbox{ and } ~\mathcal{V}(b)=\beta) & \implies & \mathcal{V}(a\vee b)=\beta \mbox{ and } ~ \mathcal{V}(a \wedge b)=\beta, \nonumber \\
&\mathcal{V}(a)=\alpha & \implies & \mathcal{V}(\neg a)=\beta, \nonumber \\
&\mathcal{V}(a)=\beta & \implies & \mathcal{V}(\neg a)=\alpha.
\end{align}
The concept of relaxation of logical-valuations comes from the intuitive correspondence between $\alpha$, $\beta$ and  $\FFF$,$\TTT$.  The requests of equations \eqref{eq:real-requests-2} correspond to the requirements for the real-valuation to be consistent with the truth tables when it can be reduced to a logical-valuation. %In this sense $\mathcal{V}\in\tilde{S}^L_{\mathcal{V}}$ are a relaxation of the logical-valuations $V\in\tilde{S}^L_{V}$.

%In fact one can ask for real-valuations to be compatible with the structure of the lattice. 
Logical-valuations are automatically compatible with the structure of the lattice just by demanding the properties we listed in equations \eqref{eq:logical-requests-1} and \eqref{eq:logical-requests-2}, but real-valuations are not because of their extended codomain.

Imposing the compatibility with the structure of the lattice on real-valuations we are able to pinpoint one of the differences between our derivation of quasi-probability  and the derivation of probability. 
We can do so restricting the discussion to real-valuations instead of the bi-valuation $\mathcal{R}$ introduced in the previous section, replacing $\mathcal{R}(x,t)$ with $\mathcal{V}(x)$.

Following the same methods as in section \ref{appendix:quasi-prob-derivation}, we aim to determine the general form of a function $\mathcal{F}$ such that for two statements $a,b$ such that $a\wedge b=\perp$:
\begin{equation}
\mathcal{V}(a\vee b)=\mathcal{F}\left[\mathcal{V}(a),\mathcal{V}(b) \right],
\end{equation}
analogously to equation \eqref{eq:sum-rul-phi}.
Proceeding as in the previous section we find out that $\mathcal{F}$ is such that
\begin{equation}
\label{eq:associativity-equation-univalue}
\mathcal{F}\left[x,\mathcal{F}[y,z] \right]=\mathcal{F}\left[\mathcal{F}[x,y],z \right],
\end{equation}
as in \eqref{eq:associativity-equation}.
As in the previous section, our goal is to find the solution to the functional equation \eqref{eq:associativity-equation-univalue}. We assume that $\mathcal{F}$ is a generic function from $\mathbb{R}^2$ to $\mathbb{R}$ and is jointly continuous.
The solution to this problem reveals how real-valuations operate on composite statements when compatibility with the rules of the lattice is demanded. The general solution to this problem does not demonstrate behaviour typical of probability.
Therefore, simply requiring compatibility with the structure of the lattice does not uniquely identify probability theory as a relaxed form of truth value assignments on a lattice.
We now explicitly impose some additional constraints on $\mathcal{V}$ and consequently on $\mathcal{F}$, demonstrating how these constraints lead to probability assignments. The justification for these additional constraints, or common sense assumptions, stems from interpreting $\mathcal{V}$ as a \textit{degree of plausibility} and the the qualitative intuition of its appropriate behaviour. Giving to $\mathcal{V}$ an interpretation involves imposing additional restrictions, hereby moving away from the most general notion of relaxing logical-valuations. 
It can be argued that the intuition on how a degree of plausibility should behave generates from our experience with experimental data and their treatment within a finite-frequentist interpretation of probability\footnote{This explanation introduces a flavour of frequentist interpretation in the derivation of probability theory. This is why some authors like Jaynes try to justify the intuition on how a degree of plausibility should behave introducing a more sophisticated philosophical framework. It is outside the scope of this appendix to analyse where these intuition comes from.}.
Since for us, non accessible statements are hidden from the experimentalist and so from any frequentist interpretation, we  do not require for valuations on inaccessible statements to satisfy these additional requests. However, for experimentally accessible statements, the concept of a degree of plausibility exhibits certain clear qualitative behaviours. The additional constraints or common-sense assumptions are as follows:

\paragraph{\textbf{Request 1:}}
If $\mathcal{V}$ must be interpreted as a "\textit{degree of plausibility}",  it is clear that it makes no sense to assign to a statement more plausibility than that of a true statement, nor less plausibility than that of a false statement.  Thus if $\alpha$ is the value associated by $\mathcal{V}$ to a false statement and $\beta$ is the value associated by $\mathcal{V}$ to a true statement, we want for the codomain of $\mathcal{V}$ to be restricted to the closed interval $[\alpha,\beta]$. This means that $\mathcal{F}$ must be a function from $[\alpha,\beta]^2$ to $[\alpha,\beta]$.
\paragraph{\textbf{Request 2:}}
If $\mathcal{V}$ must be interpreted as a "degree of plausibility", than if  $\mathcal{V}(a)=\beta$, that is the statement $a$ can be considered true, than also $\mathcal{V}(a\vee b)=\mathcal{F}\left[\mathcal{V}(a),\mathcal{V}(b)\right]=\beta$ must be considered true, independently of the value of $\mathcal{V}(b)$, even  if $\mathcal{V}(b)\neq \alpha$ and $\mathcal{V}(b)\neq \beta$.
\paragraph{\textbf{Request 3:}}
If $\mathcal{V}$ must be interpreted as a "degree of plausibility", than we expect that in general $\mathcal{V}(\alpha\vee \beta)\geq \max\left\{\mathcal{V}(a),\mathcal{V}(b)\right\}$, with equality only in the case $\mathcal{V}(a)=\alpha$ or $\mathcal{V}(b)=\alpha$, since the plausibility should grow composing statements.

The request for $\mathcal{F}$ to be a jointly continuous function that satisfy the associativity equation \eqref{eq:associativity-equation-univalue}, together with requests 1,2,3 corresponds to the hypothesis of the Representation Theorem 2.3.4.(d)\footnote{If we had chosen $\alpha > \beta$ than we would have applied the Representation Theorem 2.3.4.(a) to obtain analogous results with the difference that plausibility closer to the truth would correspond to smaller values.}~\cite{alsina2006} (In particular our Request 1 corresponds to the specification of the domain and codomain of $\mathcal{F}$, our Request 2 together with requests \eqref{eq:real-requests-2} corresponds to the requests $G(a,a)=a$ and $G(u,b)=G(b,u)=b$ in \cite{alsina2006} and our Request 3 implies that $\mathcal{F}$ do not have an interior idempotent that because of Lemma 1.3.9 in \cite{alsina2006} it is an equivalent condition to asking for $\mathcal{F}$ to be Archimedean). 

Thus function $\mathcal{F}$ admits a representation of the form 
\begin{equation}
\mathcal{F}(x,y)=f^{-1}\left(f(x)+f(y) \right),
\end{equation}
where $f$ is continuous and strictly increasing from $[\alpha,\beta]$ to $\mathbb{R^{+}}$, with $f(\alpha)=0$. Importantly, given $\mathcal{
V}$, the function $f$ is unique up to a multiplicative constant (corollary 2.2.6 in \cite{alsina2006}).
This allows for a regrading of $\mathcal{V}$ as $\hat{\mathcal{V}}=f\circ \mathcal{V}$ for which we have that for every statement $a$, $\hat{\mathcal{V}}(a)\geq 0$ and for every two statements $a,b$ such that $a\wedge b=\perp$, $\hat{\mathcal{V}}(a\vee b)=\hat{\mathcal{V}}(a)+\hat{\mathcal{V}}(b)$.
Thus we see that the positivity of the real-valuation comes from additional assumptions justified by the desire for the real-valuation to behave as a "degree of plausibility" a concept qualitatively similar to the concept of probability, while considering a real-valuation as a mere relaxation of a logical-valuation is less restrictive and allows for negative valued real-valuations.

Furthermore, we remark that asking for additional assumptions only on accessible sub-lattices encodes in the definition of $\hat{\mathcal{V}}$ information on the accessibility assignments of the lattice tethering the relaxed truth label with the accessible label.

Up to this point we have been talking about logical-valuations and real-valuations on the lattice, but in the previous section and in the main text, in order to derive quasi-probability from the rule of the lattice we made use of the bi-valuation $\mathcal{R}$.
%In the derivation of the sum rule the second argument of the bi-valuation $\mathcal{R}$. 
Since the bi-valuation $\mathcal{R}$ behaves as a real-valuation once its second argument is fixed to a constant statement $t\neq \perp$, the argument just presented can be directly mapped back to our derivation of quasi-probability presented in section \ref{appendix:quasi-prob-derivation}. With what we just presented, we highlighted the main point of departure of our derivation of quasi-probability from derivations of probability theory based on the methods of Knuth, Skilling and Caticha.

The task of highlighting the point of departure of our derivation of quasi-probability from the derivation of probability using the methods of Cox and Jaynes is more demanding.

Cox and Jaynes's methods begin with the derivation of the product rule. Both Jaynes and Cox, similar to our approach, introduce a bi-valuation, working with it without explicitly presenting a real-valuation. While introducing bi-valuations is not strictly necessary to reconstruct the product rule, it simplifies the process. This is why we also started with a bi-valuation in the main text.
Reconstructing the product rule using real-valuations is more complex and falls outside the scope of this work. However, because Cox's and Jaynes's derivations of probability from logic are the most well-known, we want to suggest that not adopting their common-sense assumptions means not adopting our common-sense assumptions either.

Following Cox and  Jaynes one starts by finding the function $G$ in equation \eqref{eq:product-associativity-equation} without knowing already the sum rule. 
Following the same kind of line of reasoning we used in the derivation of the sum rule, one can find that $G$ satisfies the associativity functional equation. 
Furthermore one can ask additional requests, or common sense assumption of the kind we asked in the case of the sum rule ending up with $G$ characterised as a jointly continuous function from $[0,1]^2$ to $[0,1]$, associative, archimedean and such that for every $x\in[0,1]$, $G(x,0)=G(x,0)=0$ and $G(1,1)=1$. Using theorem 2.3.4.(a) and lemma 2.2.1 in \cite{alsina2006} one obtains the product rule.
Jaynes, in his derivation, ask for different common sense assumption, in particular he asks for the function $G$ to be increasing in both arguments. Via lemma 2.1.1 and lemma  2.2.1 it can be shown that our common sense conditions imply that $G$ must be increasing in both arguments. Consequently, if we do not assume Jaynes' common sense condition, that is $G$ is increasing in both arguments, then this implies not assuming our own common sense assumption. 

We conclude by demonstrating that it is indeed possible to provide an example of a system with mixed accessibility values where $G$ is not increasing in both arguments.
Consider a system with $4$ atomic statements ${s_1,s_2,s_3,s_4}$, with associated quasi-probability vector $\vec{q}=\{Q(s_1),Q(s_2),Q(s_3),Q(s_4)\}=\{\frac{x(x-1)}{4},\frac{2-x}{4},-\frac{x-1^2}{4},\frac{3}{4}\}$ with  ${x\in[1-\sqrt{3},1+\sqrt{3}]}$. We consider $A=s_1\vee s_2$ accessible, $B=s1\vee s_3$ non accessible and ${C=A\vee B = s_1 \vee s_2 \vee s_3}$ non accessible.
We can compute ${Q(A\vee B|C)=x(x-1)}$, ${Q(B|C)=x-1}$, and ${Q(A|BC)=x}$ and see that ${Q(A\vee B|C)=G[Q(B|C),Q(A|BC)]}$. In the interval ${x\in[1-\sqrt{3},\frac{1}{2}]}$ the function $G$ is decreasing while both of his arguments are increasing.

\section{Scheme for the relaxation of the logical lattice}
\label{appendix:riassunto}
\begin{table}[h!]
\begin{tabularx}{\linewidth}{||c || L | L||} 
 \hline
 &Logical lattice & Quasi-probability lattice \\ [0.5ex] 
 \hline
Truth Value & Binary value assignment to each statement. & Real number assignment to each statment.\\
\hline
Accessibility & Binary value assignment to each statement. & Restriction to probabilities for accessible statements.\\
\hline
Assumption 1 & The number of inaccessible atomic statements  necessary to obtain an accessible statement is lower bounded by the accessibility-depth $d$. & The inaccessibility of the state $\vec{q}$ of the model is lower bounded by the accessibility-depth $d$.\\
\hline
\end{tabularx}
\caption{Comparison between logical lattice and the quasi-probability lattice.\label{tab:riassunto}}
\end{table}

 \section{Inaccessibility measure}
 
 \subsection{The inaccessibility measure}
 \label{appendix:distinguishability}
 
The main assumption on information inaccessibility tells us that for a given accessibility-depth $d$,  we have complete ignorance on at least $d$ elements, or better said, $d$ elements are inaccessible to our knowledge.
In this section we find a measure quantifying this amount of inaccessibility. We call this measure \textit{inaccessibility measure}. 
Having a well defined inaccessibility measure will allow us to formalise assumption 1 as a constraint on the allowed quasi-probability vectors $\vec{q}$. 

Let us denote the measure of inaccessibility with the greek letter $\uncert$. In the following we are going to characterise $\uncert$ as a function with domain over the probability vectors of the symplex $\Delta_D$, where $D$ is the dimension of the symplex\footnote{We are somehow forced to work on the space of probability vectors to not be tricked by the lack of intuition one has when reasoning on quasi-probabilities.  On the restricted domain of probability vectors  the concept of inaccessibility seems to not make sense since talking about inaccessibility requires the language of quasi-probabilities. Restricting to the space of probability vectors a better name for this information measure would be \textit{indstinguishability measure}. Nevertheless, since we are ultimately aiming at extending the domain of this information measure to the space of quasi-probability vectors, and in that context the name \textit{inaccessibility measure} is justified, we decide to already give to this measure the name of \textit{inaccessibility measure}.}. 
Clearly on $\Delta_D$ the maximum amount of inaccessible elements will be $D$ (all the elements are inaccessible, we know nothing) while the minimum amount will be $1$ (we can access every element, that is we have maximum knowledge). \\
We want to find a function $\inacc:\Delta_D\to [1,D]$ satisfying the following requests.

\paragraph{\textbf{Request (1): Counting.}}
For every integer $d\geq1$ one defines for any $\Delta_D$ the standard state with complete uncertainty on $d\leq D$ elements as 
\begin{equation}
	\mus{d} = (\underbrace{\frac{1}{d},\dots,\frac{1}{d}}_{d},0,\dots,0).
\end{equation}
We ask for
\begin{align}
	\uncert(\mus{d})=d,
\end{align}
the amount of inaccessible elements in $\mus{d}$ is $d$.

\paragraph{\textbf{Request (2): Monotonicity.}}
We ask $\mus{d}\in\Delta_D$ to be the vector with maximal inaccessibility in  $\Delta_D$
\begin{align}
	\forall \vec{p}\neq\mus{D} \in \Delta_D, \quad \uncert(\mus{D})>\uncert(\vec{p}).
\end{align}

\paragraph{\textbf{Request (3): Symmetry.}}
Relabelling the elements should not change the number of elements that are inaccessible. Thus we ask that the inaccessibility measure is invariant for any permutation of the probability vector.
Indicating with $P_{\pi}$ the permutation matrix associated to the permutation $\pi$ we ask for
\begin{align}
	\uncert(P_{\pi}\vec{p})=\uncert(\vec{p}),\qquad & \forall \text{ permutation }\pi.
\end{align}

\paragraph{\textbf{Request (4): Multiplicativity.}}
Consider the two configurations with complete uncertainty $\mus{d_1}\in \Delta_{D_1}$ and $\musq{d_2}\in \Delta_{D_2}$. If for $\mus{d_1}$ we have complete uncertainty over $d_1$ statements, while for $\musq{d_2}$ we have complete uncertainty over $d_2$ statements, we expect for the composite configuration $\mus{d_1}\otimes\musq{d_2}=\musr{d_1d_2}\in \Delta_{D_1D_2}$ to have complete uncertainty over $d_1\cdot d_2$ statements.
In general we expect for the number of inaccessible elements to multiply when composing two system.
%This implies that the measure of uncertainty $\uncert$ must be multiplicative with respect to the tensor product. That is
%\begin{align*}
%	\uncert(\musr{d_1d_2})=\uncert(\mus{d_1}\otimes\musq{d_2})=\uncert(\mus{d_1})\uncert(\musq{d_2}).
%\end{align*}
For this reason we request multiplicativity of the function $\uncert$:
\begin{align}
	\uncert(\vec{p}\otimes \vec{s})=\uncert(\vec{p})\uncert(\vec{s})
\end{align}
for all $\vec{p}$ and $\vec{s}$ belonging to any two finite probability simplicia.

It would be great to be able to characterise all the function obeying our requests. Unfortunately we are not able to provide a general proof and we need to rely on the choice of an ansatz. Nonetheless the chosen ansatz is quite general. \\

% TODO Ricordarsi di scrivere che il prodotto di funzioni moltiplicative è una funzione moltiplicativa, quindi da il prodotto di $n$ $\indisti$ posso costruire tutte quelle con $n$ elementi invece di 1. SAI che quelle con $3$ elementi sono quelle che ti servono. Pensaci.}

\paragraph{\textbf{Ansatz}}
We request for the function $\uncert$ to be of the form 
\begin{align}
	\uncert=(\sum_{i=1}^{D}f_i(p_i))^z,
\end{align}
for a specified set of  measurable functions $f_i$ and a parameter $z\in \mathbb{R}$.

% TODO (The use of this Ansatz is a requirement that I believe not necessary. Later I am going to relax it, but it would be great to have something similar in full generality. Shannon entropy can be characterised by similar requests, some additional natural requests and additivity instead of multiplicativity. Read carefully \cite{kuczma2008,aczel1975,aczel1974} to find how to adapt the results of the Cauchy functional equation in this case. In particular pay attention between the additivity with the "$+$" as in the two books and the additivity with the "$\otimes$" as in the paper of Aczel on the Shannon entropy).

\begin{lemma}
	For every $\vec{p}\in \Delta_{D\geq3}$, the most general $\uncert$ characterised by the requests (1,2,3) and written in the form of the Ansatz, form the $1$-parameter family $\uncert_{c}$ of uncertainty measure defined as
	\begin{align}
		\uncert_{c}(\vec{p})=\frac{1}{\sqrt[c-1]{\sum_{i=1}^{D}p_i^c}},
	\end{align}
	with $c\in \mathbb{R}_{+}$ and $c\neq 1$.
\end{lemma}
%\begin{proof}
%	Deriva direttamente dalle soluzioni generali della Cauchy multiplicative functional equation \cite{aczel1989,kuczma2008}
%\end{proof}
%\end{lemma}
\begin{proof}
	From the symmetry under permutations (request $(3)$)  we have that
	\begin{align}
	\label{eq:symmetric-ansatz}
	\uncert(\vec{p})=(\sum_{i=1}^{D}f(p_i))^z,
	\end{align}
	that is, there is one single function $f$.
	
	Now consider two probability vectors $\vec{p}\in \Delta_{D_1}$ and $\vec{q}\in \Delta_{D_2}$. From the request of multiplicativity (requests $(4)$) we have that
	\begin{align}
	\label{eq:Cauchy-mult-funct}
		 \sum_{(i,j)=(1,1)}^{(D_1,D_2)} f(p_iq_j) = \sum_{i=1}^{D_1} f(p_i) \sum_{j=1}^{D_2}f(q_j).
	\end{align}
	
	This is a functional equations that can be solved exploiting the solutions of the  Cauchy functional equation \cite{losonczi1981,kannappan1986,kannappan2009}. We have that, for $\vec{p}\in\Delta_{D_1\geq3}$ and $\vec{q}\geq \Delta_{D_2\geq 2}$ the most general measurable $f:[0,1]\to \mathbb{R}$ that solves the functional equation \eqref{eq:Cauchy-mult-funct} has one of the following forms \cite{losonczi1981,kannappan1986,kannappan2009}:
\begin{align}
\label{eq:fun-sol-1}&f(p)=a p+b \quad \text { if } p \in[0,1]\\
\label{eq:fun-sol-2}&f(p)= \begin{cases}0 & \text { if } p \in[0,1) \\ 1 & \text { if } p=1\end{cases}\\
\label{eq:fun-sol-3}&f(p)= \begin{cases}p^c & \text { if } p \in(0,1] \\ 0 & \text { if } p=0\end{cases},
\end{align}
where $c\in \mathbb{R}$ and $a,b$ are two constants such that $a+D_1 D_2 b= (a+D_1 b)(a+D_2 b)$. 
Solution \eqref{eq:fun-sol-1} cannot be accepted since the constants $a$ and $b$ depend on $D_1$ and $D_2$ and we do not allow for a dependence of $\indisti$ on the dimension of the simplicia. Solution \eqref{eq:fun-sol-2} has to be discarded too, since it cannot satisfy request $(1)$.
Thus we are left with solution \eqref{eq:fun-sol-3}. 

Plugging \eqref{eq:fun-sol-3} into \eqref{eq:symmetric-ansatz} we have a continuous set of solutions $\indisti_c$ parametrised by the real parameter $c$. Imposing the counting property (request $(1)$) corresponds to finding a combination of values $z$ and $c$ such that
\begin{align}
\uncert_c(\mus{d})=\left(\sum_{i=1}^{d}\frac{1}{d^c}\right)^z=\frac{1}{d^{z(c-1)}}\coloneqq d.
\end{align}
We find that $z=\frac{1}{1-c}$ and $c\neq 1$.
Finally for the monotonicity of the solution (request $(2)$) we have that $c>0$ and this completes the proof.
\end{proof}

We define the inaccessibility measure for probability vectors as follows
	\begin{definition}[Inaccessibility measure]
	For any probability simplex $\Delta_D$, the inaccessibility measure is the function $\indisti_c$ with domain $\Delta_D$ and codomain $\mathbb{R}_+$ defined as
		\begin{equation}
			\indisti_{c}(\vec{p})=\frac{1}{\sqrt[c-1]{\sum_{i=1}^{D}p_i^c}},
		\end{equation}
		where $c\in \mathbb{R}_+$ and $c\neq 1$.
	\end{definition}

\subsection{Inaccessibility measure extension to quasi-probability models of inaccessible information}
\label{appendix:extension-to-quasi-prob}

The family of inaccessibility measures $\uncert_{c}$ has been derived assuming we are interested in measuring the amount of inaccessibility of a probability vector. Nevertheless, each model of our theory of inaccessible information is described by its associated quasi-probability vector $\vec{q}$ and not by a probability vector. Furthermore, assumption $1$ refers to the amount of inaccessibility encoded in $\vec{q}$. \\
It is necessary to  find a continuation of $\indisti_c$  over the space of quasi-probability vectors of our inaccessible information theory. To do so we are going to exploit the property of the $\vec{q}$ of a MES models to be in one-to-one correspondence with the accessible probability distributions at level $d$ (see Lemma \ref{lem:MES-1-to-1}).\\

%\arrow[d, " ", rounded corners, to path={--([xshift=2ex]\tikztostart.east)|- (Z) [near,end]-| ([xshift=-2ex]\tikztotarget.west)-- (\tikztotarget)}]

%\arrow[d, controls={+(1.5,-0.2) and +(+1.5,+0.2)}]

%
%\begin{tikzcd}
%  A \arrow[r]
%    & B \arrow[r]
%        \arrow[d, phantom, ""{coordinate, name=Z}]
%      & C  \\
%                    D \arrow[r] \arrow[u,
%                 "\delta",	rounded corners,	to path={ -- ([xshift=2ex]\tikztostart.east)|- (Z) [near end]\tikztonodes-| ([xshift=+2ex]\tikztotarget.east) -- (\tikztotarget)}]
%    & E \arrow[r]
%&F 
%\end{tikzcd}

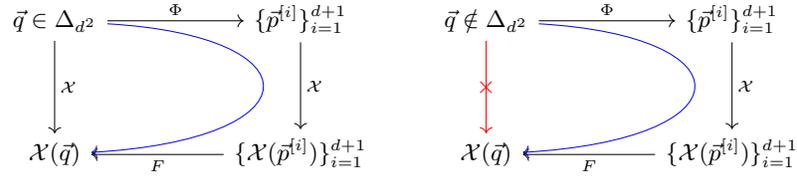
\begin{figure}[h!]
\begin{tikzcd}[row sep=huge, column sep=huge]
\vec{q} \in \Delta_{d^2} \arrow[d, blue, controls={+(3.5,-0.2) and +(3.5,+0.2)}]  \arrow[r, "\Phi"] \arrow[d, black, "\inacc"]
& \{\vec{p}^{[i]}\}_{i=1}^{d+1} \arrow[d, phantom, ""{coordinate, name=Z}]  \arrow[d, black,"\inacc" black] \\
\inacc(\vec{q}) 
&  \{\inacc(\vec{p}^{[i]})\}_{i=1}^{d+1}\arrow[l, black, "F" black]
\end{tikzcd}
 \qquad 
\begin{tikzcd}[row sep=huge, column sep=huge]
\vec{q} \notin \Delta_{d^2} \arrow[d, blue, controls={+(3.5,-0.2) and +(3.5,+0.2)}]  \arrow[r, "\Phi"] \arrow[d, red, "\times" marking]
& \{\vec{p}^{[i]}\}_{i=1}^{d+1} \arrow[d, black,"\inacc" black] \\
\inacc(\vec{q})
&  \{\inacc(\vec{p}^{[i]})\}_{i=1}^{d+1}\arrow[l, black, "F" black]
\end{tikzcd}

\caption{How to compute the inaccessibility of the state of a MES model of accessibility-depth $d$. When the vector $\vec{q}$ is a probability vector, then one can directly compute the inaccessibility $\inacc(\vec{q})$. Alternatively (following the blue arrow), one can use the one-to-one mapping $\Phi$ to map $\vec{q}$ to the collection of $d+1$ associated accessible probability vectors $\{\vec{p}^{[i]}\}_{i=1}^{d+1}$, compute the collection of the inaccessibility of each of these vectors and through the map $F$ retrieve the value of the inaccessibility computed on the original space of $\vec{q}$. For a MES model of accessibility-depth $d$ the vector $\vec{q}$ is not always an element of $\Delta_{d^2}$, in these cases the inaccessibility of $\vec{q}$ is not clearly defined (red arrow). However the blue arrow is still a viable path, as $\Phi$ always maps every $\vec{q}$ of a MES model of accessibility-depth $d$ to a set of legitimate probability vectors $\{\vec{p}^{[i]}\}_{i=1}^{d+1}$ for which the inaccessibility is well defined. The continuation of the inaccessibility measure over all the quasi-probability vectors allowed by a MES model of accessibility-depth $d$ is thus defined as the composition $F \circ \inacc \circ \Phi$ (the blue arrow).\label{fig:inacc-con-giro}}
\end{figure}

Restricting $\vec{q}$ to be a probability distribution we find that, for some values of $c$, we can express the inaccessibility $\indisc_c(\vec{q})$ as a function of the inaccessibilities $\indisc_c$ of $\{\vec{p}^{[i]}\}_{i=1,\dots,d+1}$, the associated accessible probability distribution at level $d$ of the MES model. The procedure is depicted in Figure \ref{fig:inacc-con-giro}. We call $\Phi$ the invertible map sending the quasi-probability vector $\vec{q}$ of a MES model of accessibility-depth $d$ to the collection of its associated accessible probability distributions $\{ \vec{p}^{[i]} \}_{i=1}^{d+1}$.  Since $\{\vec{p}^{[i]}\}_{i=1,\dots,d+1}$ are always probability distributions by definition, even when $\vec{q}$ assumes negative values, we can always compute their inaccessibility. Then we just need to find the function $F$ that maps the collection of inaccessibilities computed for  $\{\vec{p}^{[i]}\}_{i=1,\dots,d+1}$ to the inaccessibility of $\vec{q}$, in the case when $\vec{q}\in \Delta_{D}$. We use this method for computing the inaccessibility of quasi-probability vectors of MES models. We find that extending the inaccessibility measure to quasi-probability vectors corresponds to constraining the values of the parameter $c$ as expressed in the following definition

	\begin{definition}[Inaccessibility measure for MES models]
	For any MES model of accessibility-depth $d$, the measure of inaccessibility is defined by the function $\indisti$ mapping $\vec{q}$ to $\indisti(\vec{q})\in \mathbb{R}_+$ explicitly written as
		\begin{equation}
			\indisti(\vec{p})={\sum_{i=1}^{D}\frac{1}{p_i^2}},
		\end{equation}
	\end{definition}

In the following we are going to see in detail how the inaccessibility for MES models is derived from the general measure of inaccessibility for probability vectors.\\

For MES model the vector $\vec{q}$ is in one-to-one correspondence with the maximal set of accessible probability distributions $\{\vec{p}^{[i]}_{\vec{q}}\}_{i=1,\dots,d+1}$ at level $d$ (see Lemma \ref{lem:MES-1-to-1}).
We want to connect the amount of inaccessibility encoded in $\vec{q}$ to the amount of inaccessibility of every distribution $\{\vec{p}^{[i]}_{\vec{q}}\}_{i=1,\dots,d+1}$. 
	That is, we want to find a family of functions $F_c:\mathbb{R}^{d+1}\to[1,d^2]$, such that 
	\begin{equation}
		F_c\left(\{\inacc_c(\vec{p}^{[i]})\}_{i=1}^{d+1}\right)=\inacc_c(\vec{q}).
	\end{equation}

		\begin{lemma}
		For a MES model of accessibility-depth $d$, for $\vec{q}$ restricted to be a probability vector we have that, for $c\in\{2,3,\dots\}$, $\uncert_c(\vec{q})$ is computed from $\{\vec{p}^{[i]}_{\vec{q}}\}_{i=1,\dots,d+1}$ using the following recursive equations
		\begin{align}
		\label{eq:indisti-SMUB}
		\uncert_c(\vec{q})&=\sqrt[1-c]{ \frac{\sum_{i=1}^{d+1}\left(\uncert_c(\vec{p}^{[i]})\right)^{1-c}-G_c(\vec{q}) }{d+2-2^{c-1}}},
		\end{align}
		where
		\begin{align}
				G_c(\vec{q})&=\begin{cases}
		g_c(\vec{q}) &  \text { if } \mod(c,2)=1 \\
			g_c(\vec{q}) +\frac{1}{2}\binom{c}{\frac{c}{2}}\sqrt[1-\frac{c}{2}]{\uncert_{\frac{c}{2}}(\vec{q})} & \text { if } \mod(c,2)=0 
		\end{cases},\\	
		g_c(\vec{q})&=\sum_{i=1}^{\lerifloor{\frac{c-1}{2}}}\binom{c}{i}\sqrt[1-i]{\uncert_i(\vec{q})}\sqrt[1-c+i]{\uncert_{c-i}(\vec{q})},
		\end{align}
		and  $\sqrt[0]{\uncert_1(\vec{q})} \coloneqq 1$.
	\end{lemma}
	\begin{proof}
	It is useful to define the quantity 
	\begin{equation}
		S_c(\vec{p})\coloneqq\sum_{i}p_i^c=\left(\uncert_c(\vec{p})\right)^{1-c}.
	\end{equation}
	We can rewrite
	\begin{align}
		\indisti_c(\vec{q})=\frac{1}{\sqrt[c-1]{S_c(\vec{q})}}, \label{eq:Chic-recursive}\\
		\SMUB_c(\vec{q})=\sum_{i=1}^{d+1}S_c(\vec{p}^{[i]}).
	\end{align}
	In order to prove the lemma we express $S_c(\vec{q})$ in terms of $\SMUB(\vec{q})$.
	Since every $p^{[i]}_j$ is a sum of elements of $\vec{q}$, we have from the binomial theorem, that, for \textit{integer} $c$, $\SMUB$ can be expanded as a \textit{finite} sum of products of elements of $\vec{q}$ each at powers of integers between $0$ and $c$. 
	This fact allow us, by restricting to $c\in\{0,2,3,4,\dots\}$ and computing explicitly the expression of $\tilde{S}_c$, to obtain for $\vec{q}\in \Delta_d$:
	\begin{align}
		S_1&:=1,\nonumber \\
		\SMUB_2&=dS_2+S_1, \nonumber \\
		\SMUB_3&=(d-2)S_3+3S_2S_1, \nonumber \\
		\SMUB_4&=(d-6)S_4+4S_3S_1+3S_2S_2,\nonumber\\
		\SMUB_5&= (d-14)S_5+5S_4S_1+10S_2S_3,\nonumber\\
	\SMUB_6&=(d-30)S_6+6S_5S_6+15S_4S_2+10S_3S_3,\nonumber\\
	\SMUB_c&=(d-(2^{c-1}-2))S_c+\sum_{i=1}^{\lerifloor{\frac{c-1}{2}}}\binom{c}{i}S_i S_{c-i}+\frac{f_c}{2},\label{eq:Sc-recursive}
	\end{align}
	where 
	\begin{align}
		f_c=\begin{cases}
		0 &  \text { if } \mod(c,2)=1 \\
		\binom{c}{c/2}S_{c/2} S_{c/2} & \text { if } \mod(c,2)=0 
		\end{cases},
	\end{align}
%	\begin{align}
%		S_1(\vec{q})&=1,\nonumber \\
%		S_2(\vec{q})&=\frac{\SMUB_2(\vec{q})-S_1(\vec{q})}{d},\nonumber\\
%		S_3(\vec{q})&=\frac{\SMUB_3(\vec{q})-3S_2(\vec{q})S_1(\vec{q})}{d-2},\nonumber\\
%		S_4(\vec{q})&=\frac{\SMUB_4(\vec{q})-S_3(\vec{q})S_1(\vec{q})-3S_2(\vec{q})S_2(\vec{q})}{d-6},\nonumber\\
%S_5(\vec{q})&=\frac{\SMUB_5(\vec{q})-5S_4(\vec{q})S_1(\vec{q})-10S_2(\vec{q})S_3(\vec{q})}{d-14},\nonumber\\
%S_6(\vec{q})&=\frac{\SMUB_6(\vec{q})-6S_5(\vec{q})S_1(\vec{q})-15S_4(\vec{q})S_2(\vec{q})-10 S_3(\vec{q})S_3(\vec{q})}{d-30},\nonumber\\
%	\end{align}

	and where $S_i,\SMUB_i$ and $f_c$ are all functions of $\vec{q}$.  
	The formulas in the lemma are obtained plugging equation \eqref{eq:Sc-recursive} into equation \eqref{eq:Chic-recursive}.
	\end{proof}	
	
	\begin{figure}[h!]
	
\begin{tikzcd}[row sep=huge, column sep=huge]
\vec{q} \notin \Delta_{d^2} \arrow[d, blue, controls={+(3.5,-0.2) and +(3.5,+0.2)},"\inacc_2"]  \arrow[r, "\Phi"] \arrow[d, red, "\times" marking]
& \{\vec{p}^{[i]}\}_{i=1}^{d+1}  \arrow[d, black,"\inacc_2" black] \\
\inacc_2(\vec{q}) 
&  \{\inacc_2(\vec{p}^{[i]})\}_{i=1}^{d+1}\arrow[l, black, "F_2" black]
\end{tikzcd}
\qquad
	\begin{tikzcd}[row sep=huge, column sep=huge]
\vec{q} \notin \Delta_{d^2} \arrow[d, blue, controls={+(3.5,-0.2) and +(3.5,+0.2)},"\inacc_3?"]  \arrow[r, "\Phi"] \arrow[d, red, "\times" marking]
& \{\vec{p}^{[i]}\}_{i=1}^{d+1} \arrow[d, black,"\inacc_2" black, xshift=-4ex]\arrow[d, black,"\inacc_3" black, xshift=+3ex] \\
\inacc_3(\vec{q}) 
&  \{\inacc_2(\vec{p}^{[i]}),\inacc_3(\vec{p}^{[i]})\}_{i=1}^{d+1}\arrow[l, black, "F_3" black]
\end{tikzcd}
\caption{To continue $\inacc_2$ over the admissible states of a MES model of accessibility-depth $d$ we just need the definition of $\inacc_2$ over the probability simplex. To continue $\inacc_3$ over the admissible states of a MES model of accessibility-depth $d$ we need the definition of $\inacc_3$ over the probability simplex, but also of $\inacc_2$.\label{fig:al-3-serve-il-2}}
\end{figure}
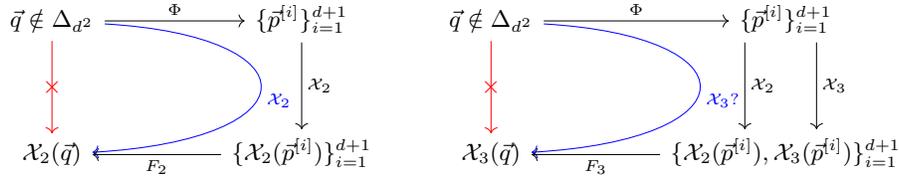

	So for example we have:
	\begin{align}
		\inacc_2(\vec{q})= \frac{d}{\sum_{i=1}^{d+1}\frac{1}{\inacc_2(\vec{p}^{[i]})}-1},\\
		\inacc_3(\vec{q})=\sqrt{\frac{d-2}{\sum_{i=1}^{d+1}\frac{1}{\inacc_3(\vec{p}^{[i]})^2}-	\frac{3}{\uncert_2(\vec{q})}}} \label{eq:inacc-3-recursive}.
	\end{align}
	We see that for $c=2$ the function $F_2$ that we were looking for exists, thus $\inacc_2$ can be continued over all the admissible quasi-probability vectors admitted in a MES model of accessibility-depth $d$.
	For $c>2$ we have that the function $F_c$ has to take as argument the accessibility of the associated accessible probability vectors $\{\inacc_{\tilde{c}}(\vec{p}^{[i]})\}_{i=1}^{d+1}$ computed for all  the $\tilde{c}\leq c$. 
	
	This means that for $c>2$ the inaccessibility measure $\indisti_c$ cannot be defined on quasi-probabilities without previously defining the inaccessibility measure $\indisti_{c-1}$, that in turn cannot be defined without previously defining the inaccessibility measure $\indisti_{c-2}$ and so on until one reach the definition of the inaccessibility measure $\indisti_2$.
	
	%So for $c>2$ the inaccessibility measure cannot be continued over all the states admitted by a MES model if one do not consider also all the inaccessibilities defined for smaller values of $c$.
	So for example, $\inacc_3(q)=F_3(\left\{\{\inacc_{2}(\vec{p}^{[i]})\}_{i=1}^{d+1},\{\inacc_{3}(\vec{p}^{[i]})\}_{i=1}^{d+1}\right\})$ as it is evident from equation \eqref{eq:inacc-3-recursive} (see Figure \ref{fig:al-3-serve-il-2}).
	
	This fact suggests a hierarchy of inaccessibilities for MES models. In fact $\inacc_{c}$ is defined for all the admissible quasi-probability vectors of a MES models only if all the $\inacc_{\tilde{c}}$ with $\tilde{c}<c$ are considered. A pictorial representation can be found in Figure \ref{fig:al-3-serve-il-2} In this work we focus on the inaccessibility $\inacc_2$ since it is the only one that can exists by its own, and in this sense it is unique\footnote{It is well known that 2-norms are "special" and somehow ubiquitous as there are many good reasons to prefer them. A beautifully written exposition of some of these reasons can be found in \cite{aaronson2004}.}.

%	
%	This is the case when $c$ is even, in fact, in the case of $c$ odd, the possibility of having an even root with negative argument, makes the inaccessibility measure not well defined.
%	Thus $\uncert_c$ with $c\in{2,4,6,\dots}$ can be continued to a measure of inaccessibility over the quasi-probability vectors $\vec{q}$ of MES models.

\subsection{Properties of the inaccessibility}
\label{appendix:properties-inaccessibility}
 As requested, the inaccessibility measure is always lower bounded by $1$ and it is upper bounded by the dimension of the symplex. In fact, the minimum number of elements between which one can be uncertain is $1$ and the maximum number is the maximum number of elements supported by the dimension of the symplex.

The inaccessibility is neither convex nor concave, but it is instead quasi-concave over the set of states of a MES model, as proven in the next lemma. 

%\begin{lemma}
%	Given a quasi probability vector $\vec{q}\in\mathbb{R}^{d^2}$ if $S_c(\vec{q})\geq 1/(\sqrt{d})^c$ for all $c \in {2,3,4,\dots,\tilde{c}}$, then $S_{\tilde{c}+1}(\vec{q})\geq 0$. 
%\end{lemma}

\begin{lemma}
\label{lem:quasi-concave}
	Given two states $\vec{p}$,$\vec{q}$  belonging to $\text{MES}_d$ and such that $\indisti_c(\vec{p})\geq d$ and $\indisti_c(\vec{q})\geq d$ with $c \in\{2,4,6,\dots\}$, then one has that for $\lambda \in [0,1]$
	\begin{equation}
		\chi_c\left(\lambda\vec{p}+(1-\lambda)\vec{q}\right)\geq \min \left(\chi_{c}(\vec{p}),\chi_{c}(\vec{q})\right).
	\end{equation}
	The inaccessibility is a quasi-concave function over the set of states of a MES model.
\end{lemma}
\begin{proof}
	Without loss of generality consider $\indisti_c(\vec{p})\leq \indisti_c(\vec{q})$.
	Using the fact that $\mathcal{X}_c(\vec{q})=\left(\|\vec{q}\|_c\right)^{\frac{c}{1-c}}$ we have that $\|\vec{q}\|_{c}\leq\|\vec{p}\|_{c}$ and
	\begin{align}
		\chi_c\left(\lambda\vec{p}+(1-\lambda)\vec{q}\right)&=\left(\|\lambda\vec{p}+(1-\lambda)\vec{q}\|_c\right)^{\frac{c}{1-c}}\geq \nonumber\\
		&\geq \left(\lambda\|\vec{p}\|_c+(1-\lambda)\|\vec{q}\|_c\right)^{\frac{c}{1-c}}\geq \nonumber\\
		 & \geq \left(\|\vec{p}\|_c\right)^{\frac{c}{1-c}}=\indisti_c(\vec{p}).
	\end{align}
	
%	Then we prove 
%	\begin{equation}
%		\min \left(\chi_{c}(\vec{p}),\chi_{c}(\vec{q})\right) \leq \chi_c\left(\lambda\vec{p}+(1-\lambda)\vec{q}\right).
%	\end{equation}
%	\js{
%	\begin{align}
%		\|\vec{p}\|_{c}&=\|\left( \lambda \vec{p}+(1-\lambda \vec{q})\right)+\left((1-\lambda)\vec{p}-(1-\lambda)\vec{q}\right)\|\leq \nonumber \\
%		&\leq \|\lambda \vec{p}+(1-\lambda \vec{q})\|+\|(1-\lambda)\vec{p}-(1-\lambda)\vec{q}\|\leq \nonumber \\
%		& \leq \|\lambda \vec{p}+(1-\lambda \vec{q})\|+(1-\lambda)\|\vec{p}-\vec{q}\|_{c}
%	\end{align}
%	Sciocco, come volevi che sto metodo funzionasse?}
\end{proof}

Furthermore, there is always a neighbourhood of $\mus{d}$ where $\uncert_c$ is concave for every admitted $c$.

\subsection{Relation between the inaccessibility measure and other information measures}
\label{appendix:comparisons-info-measure}
\begin{figure}[h!]
	\centering
	\includegraphics[width=0.45\linewidth]{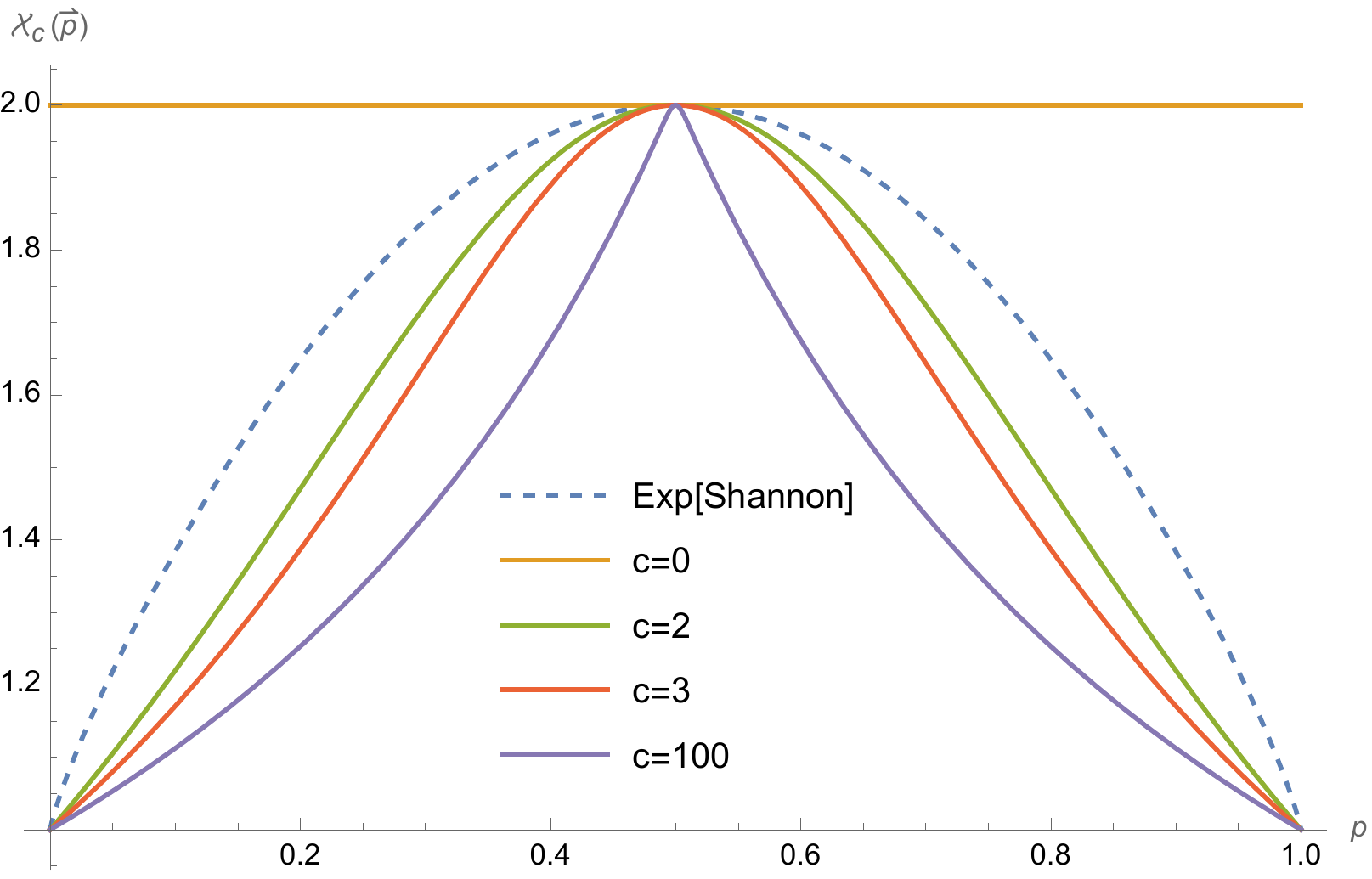}
	\caption{Inaccessibility $\uncert_c(\vec{p})$, where $\vec{p}=(p,1-p)$.}
	\label{fig:Distinguishability}
\end{figure}

In figure \ref{fig:Distinguishability} we plot the behaviour of the inaccessibility measure for different values of the parameter $c$.

The main difference between $\uncert$ and measures of entropy relies in request $(4)$. In fact measures of entropy are usually requested to be additive. 
%Nevertheless there is a trick to pass from additive to multiplicative, that is $(1-H)$. 
We can relate $\uncert$ to Shannon Entropy through the relation $\log(\lim_{c\to 1} \uncert_c)$, to Reny Entropies through the relation $\log(S_{\alpha})$, and  to Tsallis Entropies through the relation $1-(S_{\alpha})^{1-c}$.

%Properties of additivity (Shannon) and pseudo-additivity (Tsallis) can be easily mapped to the multiplicativity of the indistiguishability with the trick $(1-Ts)$.

In the literature the inaccessibility measure $\uncert$ has been used in different contexts, being called Hill's number or diversity index \cite{hill1973,pellens2016}. In the quantum literature a concept analogous to $\uncert_2$ is the effective dimension \cite{popescu2006,dunlop2021} or the operationally invariant information \cite{brukner1999}.  To the best of the author's knowledge a previous derivation of $\uncert$ from logical requests as the one of appendix \ref{appendix:distinguishability} has not been known and its use as Hill's number or diversity index has been justified only heuristically.

\section{Curiosities}

%\paragraph{Why not considering complex-quasi-probabilities?}
%Phase has to perfectly sum to 0 in the accessible statements. Negative values can sum to positive values and become positive in the interval [0,1]. There is freedom when considering negative probability, there is no freedom when considering complex quasi-probabilities.

\paragraph{Why just ideal configurations?}{}

Of course one can study any possible configuration, but we decided to focus on ideal configurations. We did so for two main reasons
\begin{enumerate}
\item We want to maximise the amount of experimentally relevant information;
\item Consider the configuration of the model $(4,1)$ of Figure \ref{fig:imbalance}. In this model there is imbalance between the “absolute” knowledge contained in the atomic statements of the classical analog. In fact, the knowledge of $c$ and $d$ is more fundamental of the knowledge of $a\vee b$ in the inaccessible model. This differentiation is lost once moving to the classical model.
To maintain this “symmetry” between classical models and their inflations to inaccessible information models we use the accessibility-depth $d$ of a system. Thus when inflating a classical model to a model of inaccessible information we do not assign different “absolute” knowledge to different statements.  
\end{enumerate}
\begin{figure}[h!]
	\centering
	\includegraphics[width=0.5\linewidth]{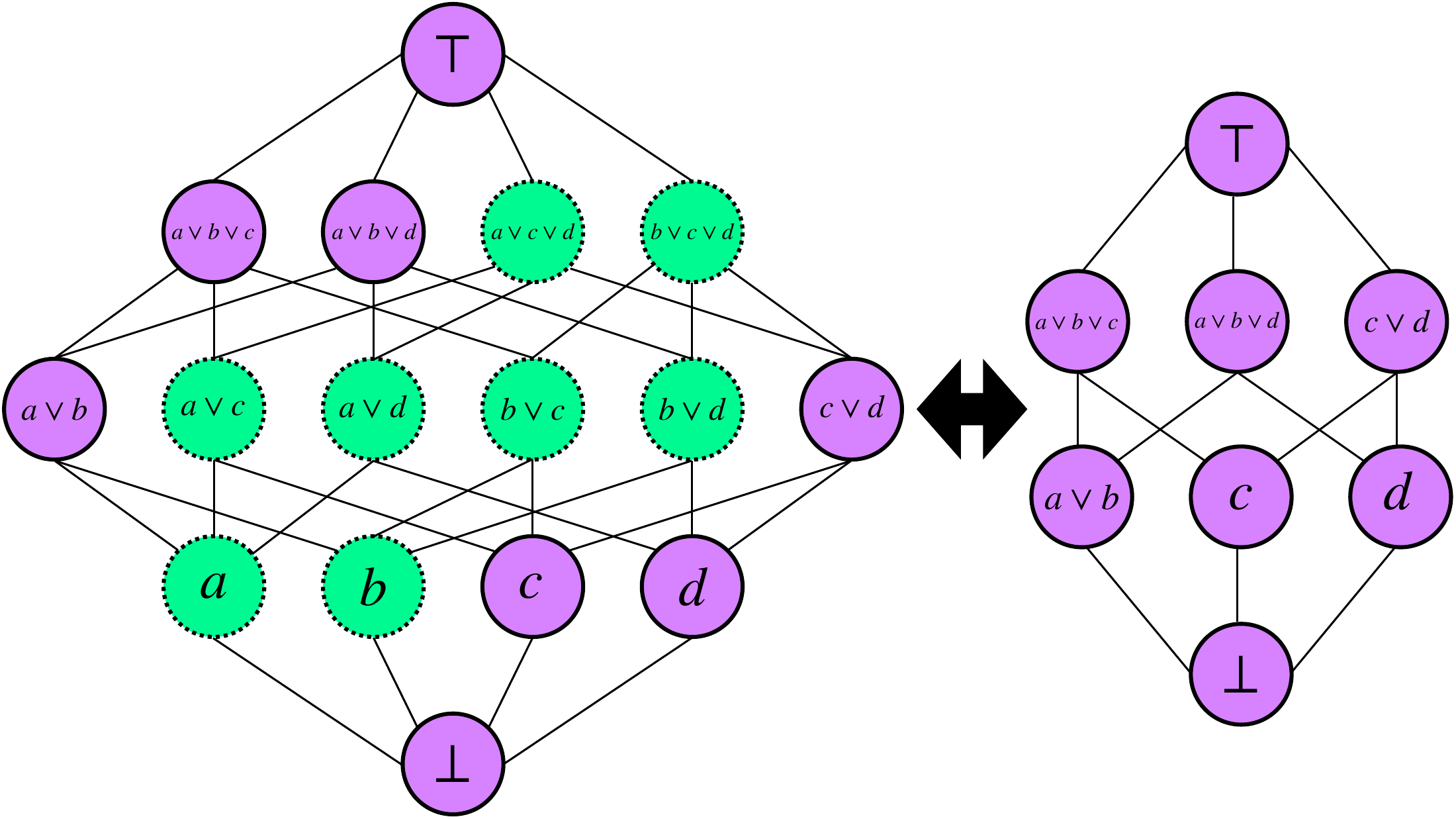}
	\caption{Asymmetric inflation with a non-ideal configuration of a $(4,1)$ model.}
	\label{fig:imbalance}
\end{figure}	

\paragraph{Why not complex quasi-probabilities?}
In the derivation of quasi-probabilities of appendix \ref{appendix:quasi-prob-derivation}, there are no apparent reasons to restrict the codomain of function $\mathcal{R}$  to the real numbers, one could instead consider a function $\mathcal{R}$ with codomain in $\mathbb{C}$. Doing so, one would obtain a function $Q$ that behaves exactly as the one we derived in appendix \ref{appendix:quasi-prob-derivation}, but that associates to each statement of the lattice a complex number. In this case the state $\vec{q}$ of a MES model would be a complex vector. Why we ruled out this possibility? For a MES model with accessibility-depth $d$ we want for the $d(d+1)$ accessible statements at level $d$ to be real numbers in the interval $[0,1]$. This means that the imaginary part of the atomic statements must sum to \textit{exactly} zero for each accessible statement. If we call $\mathfrak{q}_i=\mbox{Im}(q_i)$ the imaginary component of each element of the state vector $\vec{q}$, we obtain, considering the $d(d+1)$ accessible statements at level $d$, $d(d+1)$ constraints of the form $\mathfrak{q}_{i_1}+\mathfrak{q}_{i_2}+\dots+\mathfrak{q}_{i_d}=0$, where $\{i_1,\dots,i_d\}$ is a set of indices forming an accessible statement. Thus the imaginary part has to vanish.

 \end{appendix}

\end{document}